\documentclass[draft,12pt,reqno,a4paper]{amsart}

\usepackage[margin=3cm,footskip=1cm]{geometry}

\usepackage{amssymb} 
\usepackage{amsmath} 
\usepackage{amsfonts}
\usepackage{amsthm} 
\usepackage{mathtools}
\usepackage{color}

\usepackage{enumerate} 
\usepackage[latin1]{inputenc}
\usepackage[shortlabels]{enumitem}
\usepackage{color}
\usepackage{graphicx} 
\usepackage{verbatim}

\newcommand{\supp}{\operatorname{supp}}

\newcommand{\dist}{\operatorname{dist}}

\newcommand{\N}{{\mathbb{N}}}

\newcommand{\R}{{\mathbb{R}}}

\newcommand{\C}{{\mathbb{C}}}
\renewcommand{\S}{{\mathbb{S}}}

\DeclareFontFamily{U}{mathx}{\hyphenchar\font45}
\DeclareFontShape{U}{mathx}{m}{n}{
      <5> <6> <7> <8> <9> <10>
      <10.95> <12> <14.4> <17.28> <20.74> <24.88>
      mathx10
      }{}
\DeclareSymbolFont{mathx}{U}{mathx}{m}{n}
\DeclareFontSubstitution{U}{mathx}{m}{n}
\DeclareMathAccent{\widecheck}{0}{mathx}{"71}

\renewcommand\i{\mathrm{i}}

\newcommand{\p}{{\mathrm p}}

\renewcommand{\c}{{\mathrm c}}

\newcommand{\e}{{\mathrm e}}
\newcommand{\ess}{{\mathrm {ess}}}

\renewcommand{\d}{{\mathrm d}}

\newcommand{\pupo}{{\mathrm {pp}}}

\renewcommand{\Re}{\operatorname{Re}}
\renewcommand{\Im}{\operatorname{Im}}

\DeclarePairedDelimiter\inp\langle\rangle

\newcommand\paro[2][]{#1  ( #2#1 )}

\newcommand\parb[2][]{#1 \big ( #2#1\big )}

\newcommand\parbb[2][]{#1 \Big ( #2#1\Big )}
 
 \newcommand{\pp}{{\mathrm {pp}}}
  
\newcommand{\mand}{\text{ \,and\, }}

\DeclarePairedDelimiter\ket{\lvert}{\rangle}
\DeclarePairedDelimiter\bra{\langle}{\rvert}

\DeclareMathOperator*{\slim}{s-lim}
\DeclareMathOperator*{\wlim}{w-lim}

\DeclareMathOperator*{\wv2lim}{{w-\widetilde{\mathcal H}}-lim}

\DeclareMathOperator*{\vLlim}{{\mathcal L(\mathcal H)}-lim}

\DeclareMathOperator*{\swslim}{s- w^\star-lim}

\DeclarePairedDelimiter\abs\lvert\rvert
\DeclarePairedDelimiter\norm\lVert\rVert
\DeclarePairedDelimiter\set{\{}{\}}

\newcommand{\brT}{{\breve T}}

\newcommand{\bre}{{\breve \epsilon}}
\newcommand{\brp}{{\breve \psi}}

\newcommand{\bD}{{\mathbf D}}

\newcommand{\bY}{{\mathbf Y}}
\newcommand{\bX}{{\mathbf X}}

\newcommand{\vA}{{\mathcal A}}
\newcommand{\vB}{{\mathcal B}}

\newcommand{\vE}{{\mathcal E}}

\newcommand{\vG}{{\mathcal G}}

\newcommand{\vH}{{\mathcal H}}

\newcommand{\vL}{{\mathcal L}}

\newcommand{\vN}{{\mathcal N}}
\newcommand{\vO}{{\mathcal O}}

\newcommand{\vT}{{\mathcal T}}

\theoremstyle{plain}
\newtheorem{thm}{Theorem}[section]
\newtheorem{defn}[thm]{Definition} 
 
\newtheorem{proposition}[thm]{Proposition}
\newtheorem{lemma}[thm]{Lemma} \newtheorem{corollary}[thm]{Corollary}
\newtheorem{cond}[thm]{Condition}
\theoremstyle{definition}

\newtheorem{remarks}[thm]{Remarks}
\newtheorem*{remarks*}{Remarks}
\newtheorem*{remark*}{Remark}

\numberwithin{equation}{section}

\title {Stationary completeness:  the $N$-body
   short-range case}

\thanks{
Supported by DFF grant nr.\ 8021-00084B}

\author{E. Skibsted} \address[E. Skibsted]{Institut for Matematik\\
Aarhus Universitet\\ Ny Munkegade 8000 Aarhus C, Denmark}
\email{skibsted@math.au.dk}

\begin{document}

\begin{abstract} 
For a general class of $N$-body  Schr\"odinger operators with
short-range  pair-potentials the wave and
   scattering matrices as well as the restricted wave operators are
   all defined at any    non-threshold energy. This holds 
   without imposing  any a priori decay condition on  channel
   eigenstates and even for models including  long-range
   potentials of Derezi\'nski-Enss type. In
   this paper we improve for   short-range  models on the known \emph{weak continuity}
   properties    in that we show that \emph{all}  non-threshold
   energies are \emph{stationary complete}, resolving  in this  case a conjecture
   from \cite {Sk3}. A consequence is that the above
   scattering quantities depend \emph{strongly continuously} on the energy
   parameter  at \emph{all} non-threshold energies  (improving on 
   previously almost
   everywhere proven  properties). Another consequence is that the
   scattering matrix is unitary at any such  energy. As a side result
   we obtain  a new  and    purely stationary  proof of asymptotic completeness for $N$-body
   short-range systems. 
  \end{abstract}

\allowdisplaybreaks

\maketitle


\tableofcontents

\section{Introduction}\label{sec:Introduction}

We address a  conjecture of \cite{Sk3} for the stationary
scattering theory of $N$-body systems of quantum particles interacting
with long-range
   potentials of Derezi\'nski-Enss type
long-range pair-potentials. More precisely we focus on   the conjecture
in the short-range case, i.e. for models  defined by 
short-range pair-potentials only.  In this setting    we
will  show that indeed 
 \emph{all}
non-threshold energies are \emph{stationary complete}, resolving  the conjecture
  to the affirmative.

\subsection{The standard   $N$-body short-range model, results}\label{subsec: Atomic 3-body model}
Consider a system of  $N$  (possibly charged) particles, $i=1,\dots,N$,  of dimension $n\ge 1$
interacting by pair-potentials 
\begin{equation*}
	V_{ij} (x_i -x_j)=\vO(\abs{x_i -x_j}^{-1-\epsilon}), \quad \epsilon>0.
\end{equation*} Here $x_j\in\R^n$ is the position of particle $j$, and denoting   $m_j$  the corresponding mass, the 
Hamiltonian  reads 
\begin{equation*}
H=-\sum_{j = 1}^N \frac{1}{2m_j}\Delta_{x_j} + \sum_{1 \le i<j \le N} V_{ij} (x_i -x_j).
\end{equation*}

 This Hamiltonian $H$ is regarded as a self-adjoint operator
on $L^2(\bX)$, where $\bX$ is the $(N-1)n$ dimensional real vector
space  $ \{ x=(x_1,\dots ,x_N)\mid\sum_{j=1}^{N} m_j x_j = 0\}$.
Let $\vA$  denote the set of all cluster
decompositions of the system. The notation $a_{\max}$ and
$a_{\min}$ refers to the $1$-cluster and $N$-cluster decompositions,
respectively.
  Let for $a\in\vA$ the notation  $\# a$ denote the number of
clusters in $a$.
For $i,j \in\{1, \dots, N\}$, $i< j$, we denote by $(ij) $ the
$(N-1)$-cluster decomposition given by letting $C=\{i,j\}$ form a
cluster and the other  particles form  singletons. We write $(ij) \leq a$ if $i$ and $j$ belong to the same cluster
in $a$.   More general, we write $b\leq  a$ if each cluster of $b$
is a subset of a cluster of $a$. If $a$ is a $k$-cluster decomposition, $a= (C_1, \dots, C_k)$,
we let
\begin{equation*}
\bX^a = \set[\big]{ x\in\bX\mid \sum_{l\in C_j } m_l x_l = 0,  j = 1, \dots,
k}=\bX^{C_1}\oplus\cdots \oplus\bX^{C_k},
\end{equation*}
and
\[
\bX_a  =  \set[\big]{ x\in\bX\mid  x_i = x_j \mbox{ if } i,j \in C_m  \mbox{ for some }
m \in \{ 1, \dots, k\}  }.
\]
 Note that $b\leq a\Leftrightarrow \bX^b\subseteq\bX^a$, and that the
 subspaces  $\bX^a$ and $\bX_a$  define  an orthogonal decomposition
 of  $\bX$
equipped  with
the quadratic form
$q(x)=\Sigma_j \,2m_j|x_j|^2,  \, x\in {\bX}$.
 Consequently any  $x\in \bX$ decomposes orthogonally as 
 $x =x^{a} + x_{a}$ with $x^a =\pi^a x\in\bX^a$ and $x_a =\pi_a x\in
 \bX_a$. The subspaces $\bX^a$ and $\bX_a$ are usually referred to as 
 `intra- and inter-cluster configuration spaces', respectively.

With these notations the  $N$-body Schr\"odinger operator
  takes the form $ H = H_0 + V$, 
where  $H_0=p^2$ is (minus)  the Laplace-Beltrami operator on   the
Euclidean space  $(\bX, q)$ and
$V=V(x) =  \sum_{b=(ij)\in\vA} V_{b}(x^{b}) $ with $ V_b (x^b) =
V_{ij} (x_i - x_j)$ for the  $(N-1)$-cluster decomposition
$b=(ij)$. Note for example  that 
\begin{equation*}
  x^{(12)}=\parb{\tfrac{m_2}{m_1+m_2}(x_1-x_2),-\tfrac{m_1}{m_1+m_2}(x_1-x_2),0,\dots,0}.
\end{equation*}

More generally for any cluster 
decomposition $a\in\vA$ we introduce a   Hamiltonian $H^a$ as follows. 
For $a=a_{\min}$  we define
$H^{a_{\min}}=0$ on $\mathcal H^{a_{\min}}:=\mathbb C. $
For $a\neq a_{\min}$ we introduce the potential 
\begin{equation*}
V^a(x^a)=\sum_{b=(ij)\leq a} V_{b}(x^b);\quad
x^a\in \bX^a.
\end{equation*} 
Then 
\begin{equation*}
 H^a=-
\Delta_{x^a} +V^a(x^a)=
(p^a)^2 +V^a\ \ 
\text{on }\mathcal H^a=L^2(\bX^a).
\end{equation*}

A channel $\alpha$ is by
definition given as $\alpha =(a,\lambda^\alpha, u^\alpha)$, where
$a\in\vA'=\vA\setminus \{a_{\max}\}$ and  $u^\alpha\in \mathcal H^a$ obeys
$\norm{u^\alpha}=1$ and 
$(H^a-\lambda^\alpha)u^\alpha=0$ for a real number
$\lambda^\alpha$, named a threshold. The set of thresholds is denoted
$\vT(H)$,  and including the eigenvalues of $H$ we introduce
 $\vT_{\p}(H)=\sigma_{\pp}(H)\cup
  \vT(H)$.  For any $a\in \vA'$ the intercluster potential is by definition
\begin{align*}
  I_a(x)=\sum_{b=(ij)\not\leq a}V_b(x^b).
\end{align*} Next we recall  the  short-range channel wave operators
\begin{equation}\label{eq:Atomwave_op}
  W_{\alpha}^{\pm}=\slim_{t\to \pm\infty}\e^{\i
  tH}\parb{{u^\alpha}\otimes \e^{-\i tk_\alpha}
  (\cdot)};\quad k_\alpha=p_a^2+\lambda^\alpha.
\end{equation} The $N$-body scattering theory originates from physics
and has a long history, also considered as a mathematical subject of
its own right.   By
now there exist several   proofs of the existence of the
channel wave operators as well as of  their completeness
\cite{GM,SS,Gr,En, De, Ta, HS, Ya3, Zi} (the list is not complete), see
also the monographs \cite{DG, Is5}. Several of the known proofs  are
manifestly time-dependent. The work by Yafaev  appears in this regard
as an  exception (at least  partially), although \cite{Ya3}  is not entirely
stationary neither.

  By the intertwining property $H W_{\alpha}^{\pm}\supseteq
W_{\alpha}^{\pm}k_\alpha $  and the fact that $k_\alpha$ is diagonalized by
the unitary map  $F_\alpha:L^2(\mathbf X_a)\to L^2(I^\alpha
;\vG_a)$, $I^\alpha=(\lambda^\alpha,\infty)$, $\vG_a=L^2(\mathbf{S}_a)$,  $\mathbf{S}_a=\mathbf X_a\cap\S^{d_a-1}$ with $d_a=\dim
\mathbf X_a$,  given by
\begin{align*}
  (F_\alpha \varphi)(\lambda,\omega)=(2\pi)^{-d_a/2}2^{-1/2}
  \lambda_\alpha^{(d_a-2)/4}\int \e^{-\i  \lambda^{1/2}_\alpha \omega\cdot
  x_a}\varphi(x_a)\,\d x_a,\quad \lambda_\alpha=\lambda-\lambda^\alpha,
\end{align*} we can for {any} two given  channels $\alpha$ and
$\beta$  write 
\begin{equation*}
  \hat
  S_{\beta\alpha}:=F_\beta(W_{\beta}^+)^*W_{\alpha}^-F_\alpha^{-1}=\int^\oplus_{
  I_{\beta\alpha} }
  S_{\beta\alpha}(\lambda)\,\d \lambda,\quad
  I_{\beta\alpha}=I^\beta\cap I^\alpha.
\end{equation*} The fiber operator $ S_{\beta\alpha}(\lambda)\in
\vL(\vG_a,\vG_b)$ is from an abstract point of view 
a priori defined only for
a.e. $\lambda \in I_{\beta\alpha}$, however it is known to exist  away from the set of thresholds (see Theorem \ref{thm:chann-wave-matrD} stated below). It is the
{$\beta\alpha$-entry of the   scattering matrix}
$S(\lambda)=\parb{S_{\beta\alpha}(\lambda)}_{\beta\alpha}$ (here the  `dimension' of the
`matrix' $S(\lambda)$ is $\lambda$-independent on
any interval not containing thresholds).

Introducing the standard notation for weighted spaces 
\begin{align*}
  L_s^2(\mathbf X)=\inp{x}^{-s}L^2(\mathbf X),
\quad s\in\R, \quad\inp{x}=\parb{1+\abs{x}^2}^{1/2},
\end{align*} we recall the following result. (See the paragraph after
the theorem for some physics interpretation.)

\begin{thm}[\cite {Sk3}]\label{thm:chann-wave-matrD}
  \begin{enumerate}[1)]
  \item 
Let $\alpha$ be a  given channel  $\alpha
    =(a,\lambda^\alpha, u^\alpha)$,     $f:I^\alpha \to \C$ be
  continuous and compactly supported away from  $\vT_{\p}(H)$, and let $s>1/2$. For any $\varphi\in
  L^2(\mathbf  X_a)$   
  \begin{align}\label{eq:wavD1}
  W^\pm_{\alpha}
  f\paro{ k_\alpha}\varphi=\int_{I^\alpha \setminus \vT_\p(H)} \,\,f(\lambda)
    W^\pm_{\alpha}(\lambda) \paro{F_\alpha \varphi)(\lambda,
  \cdot} 
    \,\d \lambda\in  L_{-s}^2(\bX), 
              \end{align} where the `channel wave matrices'
              $W^\pm_{\alpha}(\lambda)\in\vL\parb{\vG_a,L_{-s}^2(\bX)}$
              with a strongly continuous dependence on $\lambda$.
 In particular for  $\varphi\in
  L_s^2(\mathbf  X_b)$   
the integrand is a continuous compactly supported 
$L_{-s}^2(\bX)$-valued function. In general the integral
has  the weak interpretation of an integral of a measurable 
$L_{-s}^2(\bX)$-valued function.
\item The operator-valued function
  $I_{\beta\alpha}\setminus \vT_{\p}(H)\ni \lambda \to S_{\beta\alpha}(\lambda)\in
\vL(\vG_a,\vG_b)$ is weakly continuous.  
\end{enumerate}
\end{thm}

The integral in \eqref{eq:wavD1} has the interpretation  of being a
`wave-packet' of generalized eigenfunctions, which in turn may be seen
as a manifestation of the mentioned property $H
W^\pm_{\alpha}\supseteq  W^\pm_{\alpha}k_\alpha$. Similarly the
angular dependence of the factor $\paro{F_\alpha \varphi)(\lambda,
  \omega} $  in the integrals  determines the large-time 
asymptotics  (more precisely  the asymptotics  for $t\to \pm \infty$) of the corresponding time-depending wave-packets in conical
regions of the inter-cluster configuration spaces. Hence
the angle $\omega$ has 
the interpretation of being `incoming or outgoing directions', and the
operator $S_{\beta\alpha}(\lambda)$ (possibly thought of as having kernel
$S_{\beta\alpha}(\lambda;\omega_b,\omega_a)$) links conical $\alpha$-incoming
data with 
conical $\beta$-outgoing  scattering data.

Using  the notation $ L^2_\infty(\mathbf X)=\cap_s L^2_s(\mathbf X)$ the
delta-function of $H$ at $\lambda$ is given by 
\begin{equation*}
   \delta(H-\lambda) =\pi^{-1}\Im
{(H-\lambda-\i
0)^{-1}}\text{ as a quadratic form on } L^2_\infty(\mathbf X).
\end{equation*} The adjoint operators $\Gamma^\pm_{\alpha}(\lambda)=W^\pm_{\alpha}(\lambda)^*$ are
referred to as `restricted channel wave operators'.
\begin{defn}\label{defn:scatEnergy0}  
An  energy $\lambda
  \in \vE:=(\min \vT(H),\infty)\setminus\vT_\p(H)$  is  {stationary
    complete} for $H$ if  
\begin{equation}\label{eq:ScatEnergy22233}
  \forall \psi\in  L^2_\infty(\mathbf X):\,\, \sum_{\lambda^\beta<
  \lambda}\,\norm{\Gamma^\pm_{\beta}(\lambda) \psi}^2= 
  \inp{\psi,\delta(H-\lambda){\psi}}.
\end{equation} 
  \end{defn}

In this paper  we take the existence of the channel wave operators for
granted, although this property is an easy consequence of our method
(see Remark \ref{remark:R} \ref{item:2Sk}). With this in place  the
assertion  of their
completeness is equivalent to the assertion  that Lebesgue almost all energies are  stationary
complete (see \cite{Sk3}). In particular we obtain asymptototic completeness as  a consequence of  the following main result of the  paper (stated here for the standard   $N$-body short-range model only).
\begin{thm}\label{thm:Parsconcl-gener0} All $\lambda
  \in \vE$ are stationary complete for the 
Hamiltonian $H$.
  \end{thm} 

Note that while
\eqref{eq:wavD1} may be taken as a definition   of the channel wave
matrices (although  of course incomplete
and implicit),
\eqref{eq:ScatEnergy22233} is a non-trivial derived property.  The notion of stationary complete energies was  studied in  \cite{Sk3}  for a model including long-range  potentials of the type studied previously in \cite{De, En}, in particular for  the  atomic $N$-body
Hamiltonian.  Even  in the long-range  setting  Lebesgue almost all non-threshold energies  are  stationary
complete, however it was left as
an  open problem to show stationary completeness at fixed energy in that case. For  
  the atomic $N$-body  model the problem  remains open   for $N\geq 4$,
  however it was recently solved for Derezi\'nski-Enss type
 long-range potentials for  $N=3$ \cite{Sk2}.

One can regard \eqref{eq:ScatEnergy22233} as an `on-shell  Parseval
formula'. By integration it yields asymptotic completeness, hence
providing an alternative to other proofs.  Moreover 
there are  immediate consequences  for the discussed scattering
quantities considered as operator-valued functions on $\vE$ (recalled for the  generalized  $N$-body short-range model 
in Section \ref{sec:-body effective potential and a $1$-body radial limit}):
\begin{enumerate}[I)]
\item \label{item:stroS} The scattering matrix $S(\cdot)$ is a  strongly continuous unitary operator
    determined uniquely  by asymptotics of minimum
  generalized eigenfunctions (at any given  energy). The   latter are
  taken from  the ranges
  of the channel wave matrices (at this energy). 
\item \label{item:stroS2} The  restricted channel wave operators $\Gamma^\pm_{\alpha}(\cdot)$ are strongly continuous.
\item \label{item:stroS3}  The scattering matrix links the incoming
  and outgoing channel wave 
  matrices by the following (strongly convergent) summation formula,
\begin{equation*}
     W^-_{\alpha}(\lambda)g=\sum_{\lambda^\beta<\lambda}
     W^+_{\beta}(\lambda)S_{\beta\alpha}(\lambda)g;\quad \lambda \in
     I^\alpha\setminus\vT_\p(H)\mand g\in \vG_a.
  \end{equation*}
\end{enumerate}

The many-body stationary scattering theory appears so far incompletely
developed. In particular, to our knowledge, the results
\eqref{eq:ScatEnergy22233} and \ref{item:stroS}--\ref{item:stroS3} are
all new for $N\geq 4$, while they are substantial improvements  of known
results for $N= 3$ (see the discussion in Subsection \ref{subsec:
  Extensions and comparison with the literature}). They may in fact be
viewed as a `completion' of the $N$-body short-range stationary
scattering theory in that they settle various  fundamental issues not
covered by Theorem \ref{thm:chann-wave-matrD}.

The 
  proof of \eqref{eq:ScatEnergy22233}  relies on a characterization
of this property from  \cite {Sk3} (recalled in
\eqref{eq:asres29}).  In particular we derive the top-order asymptotics of any vector
on the form $(H-\lambda-\i
0)^{-1}\psi$,  $\psi\in  L^2_\infty(\mathbf X)$, yielding
\eqref{eq:ScatEnergy22233}. 
 The asymptotics  is also
valid in the $N$-body setting of \cite{Sk3}, however there proven only 
away from a Lebesgue null-set. The  procedure
of \cite{Sk3}  is very different, and  we stress that it  excludes any stronger
conclusion.

\subsection{Comparison with the literature, content of the paper}\label{subsec: Extensions and comparison with the literature}

For $3$-body systems there is a fairly big literature on  stationary scattering theory both on the mathematical side  and the physics side. This
is to a large extent  based on the Faddeev method or some modification of
that, see for example \cite{GM,Ne}. The Faddeev method, as for
example used in the mathematically rigorous paper \cite{GM}, requires
fall-off like $V_a=\vO(\abs{x^a}^{-2-\epsilon})$. Moreover  there are
additional complications in that the threshold zero needs be be
regular for the two-body systems (i.e.  zero-energy eigenvalues
and  resonances are excluded) and  `spurious poles' cannot be ruled 
out  (these poles
would arise from lack of solvability of a certain Lippmann--Schwinger
type equation).

The work on the $3$-body stationary scattering theory  
\cite{Is3} (see also its partial precursor
\cite{Is2} and  \cite[Chapter 5]{Is5})  is different. In fact Isozaki does not assume any
regularity at zero energy for the two-body
systems  and his theory does not have spurious
poles. He overcomes these deficencies by avoiding the otherwise
prevailing Faddeev method. On the other hand \cite{Is3} still
needs, in some comparison argument, a very detailed  information on the
spectral theory of the (two-body) sub-Hamiltonians at zero energy, and
this requires the  fall-off condition 
$V_a=\vO(\abs{x^a}^{-5-\epsilon})$ (as well as a restriction on the
particle dimension). The  present paper does not
involve such detailed  information. In fact in  comparison with
Isozaki's works  a  Borel
calculus argument  suffices. Otherwise the overall spirit of the present paper
and \cite{Is3} is  the same, in particular the use of resolvent
equations is   kept at a minimum (solvability of Lippmann--Schwinger
type equations is not an issue) and both works employ intensively
 Besov spaces not only in proofs but also in the  formulation of
 various results.

In conclusion we  recover all of Isozaki's results 
by a new  method that works down  to the critical exponent $1$,
i.e. for $V_a=\vO(\abs{x^a}^{-1-\epsilon})$, hence for general short-range potentials. Although this is not widely
recognized in the literature on stationary scattering theory, various sharp
Besov space assertions appear desirable and (beyond their own right of being
sharp)   are clearly  useful. For the approach of the present   paper in
fact  their sharpness 
plays a  crucial role.

The literature on $N$-body stationary scattering theory for $N\ge 4$ is very limited, see though \cite{Ya1,Ya3,Ya4} which overlaps with \cite {Sk3} regarding Theorem \ref{thm:chann-wave-matrD}.

We remark that the strong continuity assertion for the scattering
matrix, cf. 
\ref{item:stroS}, cannot in general be replaced by norm continuity, see
\cite[Subsection 7.6]{Ya1} for a counterexample  for  a short-range
potential. In this sense the stated regularity  of the scattering
matrix in \ref{item:stroS} is optimal. 

We also remark that from a physics point of view, well-definedness and continuity of basic
scattering quantities are fundamemntal issues. However once settled there are
several   interesting  problems left, like the structure of the singularities of the kernel of the scattering
matrix (possibly considered as function of the  energy parameter), see
for example \cite{ Is1, Is5, SW} for  results  on 
this topic.

The paper is organized as follows. In Section \ref{sec:preliminaries} we define a generalized $N$-body short-range model and state various basic stationary theory. In particular we introduce in Remarks  \ref{remark:RQ} a notion of 
`$Q$-bounds' that will be vital for our procedure. Such bounds arise
in our application of vector field constructions, essentially due to
Yafaev, as well as for other issues. The Yafaev type   constructions
are  studied in   detail in the
somewhat technical  Section \ref{sec:
  Yafaev type constructions} (in particular presented in a self-contained
and slightly  extended form). The reader may in a first reading   skip
this section (and come back to it  when necessary for further reading)  and go directly to  Section \ref{sec:-body effective
  potential and a $1$-body radial limit}, where we give an account of various relevant results on scattering theory from \cite{Sk3}. In particular our strategy of proof of the pointwise stationary completeness is explained in Section \ref{sec:-body effective potential and a $1$-body radial limit}. In Section \ref{sec::A partition of unity} we introduce a convenient phase-space partition of unity, which then is used in Section \ref{sec:Squeezing argument}  in combination with the Borel calculus (trivializing  threshold problems). We are left with  certain energy-localized states for which a  `$Q$-bound' of Section \ref{sec:Non-threshold analysis}  will apply in an iteration procedure, eventually to be explained and completed  in Section
\ref{sec:Proof of stationary completeness}. Finally we devote Appendix
\ref{sec:Strong bounds} to an elaboration of some technical details
left out in Section \ref{sec::A partition of unity}.

Although we do give an account of the
most relevant parts of  \cite{Sk3}, the present paper contains proofs
for which the reader presumably would  benefit from independent parallel
consultance of \cite{Sk3} (see the discussion after \eqref{82a0}).

\section{Preliminaries}\label{sec:preliminaries}

We explain our abstract setting  and give an account of some basic stationary  $N$-body theory. 
\subsection{$N$-body Hamiltonians, assumptions  and
  notation}\label{subsec:body Hamiltonians, limiting absorption
  principle  and notation}

First we explain our model and then we introduce a number of
useful general notation.

\subsubsection{Generalized $N$-body short-range  Hamiltonians}\label{subsubsec::Short-range N-body Hamiltonians}

Let $\bX$ be                    
a (nonzero) finite dimensional real inner product space,
equipped with a
finite family $\{\bX_a\}_{a\in \vA}$ of subspaces closed under intersection:
For any $a,b\in\mathcal A$ there exists $c\in\mathcal A$, denoted
$c=a\vee b$,  such that $\bX_a\cap\bX_b=\bX_c$.
 We
 order $\vA$ by  writing $a\leq b$ (or equivalently as $b\geq a$) if
$\bX_a\supseteq \bX_b$. 
It is assumed that there exist
$a_{\min},a_{\max}\in \vA$ such that 
$\bX_{a_{\min}}=\bX$ and 
$\bX_{a_{\max}}=\{0\}$. The subspaces $\bX_a$, $a\neq a_{\min} $,  are called  \emph{collision
 planes}. We will  use the notation $d_a=\dim
\mathbf X_a$ and in particular  the abbreviated notation
$d=d_{a_{\min}}=\dim \mathbf X$.  Let $\vA'=\vA\setminus
  \{a_{\max}\} $. 
 
The $2$-body model (or more correctly named `the one-body model')  is
based on the structure 
 $\vA=\set{a_{\min},a_{\max}}$. 
The scattering  theory for such models is well-understood,
in fact (here including long-range potentials) there are several doable approaches, see for example the monographs \cite{DG,Is5} for  accounts  on time-dependent  and stationary
   scattering  theories. In this paper we  consider short-range potentials only, see the below Condition
   \ref{cond:smooth2wea3n12}, and the scattering theory for the $2$-body as well as for the $N$-body model has a  canonical meaning (to be discussed in 
   Section \ref{sec:-body effective potential and a $1$-body radial limit}).

Let $\bX^a\subseteq\bX$ be the orthogonal complement of $\bX_a\subseteq \bX$,
and denote the associated orthogonal decomposition of $x\in\bX$ by 
$$x=x^a\oplus x_a=\pi^ax\oplus \pi_ax\in\bX^a\oplus \bX_a.$$ 
The vectors $x^a$ and $x_a$ may be  called the \emph{internal
  component}  and the 
\emph{inter-cluster component}  of $x$, respectively. The momentum operator
$p=-\i \nabla$ decomposes similarly, $p=p^a\oplus p_a$. For $a\neq a_{\min} $ the  unit sphere  in $\mathbf X^{a}$ is denoted
by $\mathbf{S}^{a}$. For $a\neq a_{\max} $ the  unit sphere  in $\mathbf X_{a}$ is denoted
by $\mathbf{S}_{a}$.

A real-valued measurable function $V\colon\bX\to\mathbb R$ is 
a \textit{potential of many-body type} 
if there exist real-valued measurable functions
$V_a\colon\bX^a\to\mathbb R$ such that 
\begin{equation*}
V(x)=\sum_{a\in\mathcal A}V_a(x^a)\ \ \text{for }x\in\mathbf X.
\end{equation*} We take $V_{a_{\min}}=0$ and  impose throughout the paper the
following condition for    $a\neq a_{\min}$.  
\begin{cond}\label{cond:smooth2wea3n12}
    There exists $\mu\in (0,1/2)$  such that for all $a\in \vA\setminus\set{a_{\min}}$ the
    potential $V_a=V_a(x^a)$ fulfills:
    \begin{enumerate}
    \item \label{item:shortr}$V_a(-\Delta_{x^a}+1)^{-1}$ is compact. 
    \item \label{item:shortl}$\abs{x^a}^{1+2\mu}
      V_a(-\Delta_{x^a}+1)^{-1}$ is bounded.
\end{enumerate}
\end{cond}

 For any $a\in\vA$  we  introduce   associated  {Hamiltonians} $H^a$
 and $H_a$   as follows. 
For $a=a_{\min}$  we define
$H^{a_{\min}}=0$ on $\mathcal H^{a_{\min}}=L^2(\set{0})=\mathbb C $
and $H_{a_{\min}} =
p^2$ on $L^2(\bX)$, respectively. 
For $a\neq a_{\min}$ 
we let 
\begin{equation*}
V^a(x^a)=\sum_{b\leq a} V_{b}(x^b)
,\quad
x^a\in \bX^a,
\end{equation*} 
and  define  then 
\begin{equation*}
 H^a=-
\Delta_{x^a} +V^a\ \ 
\text{on }\mathcal H^a=L^2(\bX^a)\mand H_a=H^a \otimes I +I\otimes
p^2_a\text{ on }L^2(\bX).
\end{equation*} 
We  abbreviate 
\begin{align*}
V^{a_{\max}}=V,\quad
 H^{a_{\max}}=H,\quad 
 \mathcal H^{a_{\max}}=\mathcal H=L^2(\bX).
\end{align*}
 The operator $H$ (with domain $
\mathcal D(H)=H^2(\mathbf X)$) is  the full
Hamiltonian of the $N$-body model,  and the \textit{thresholds} of $H$ are by
definition the
eigenvalues of the  sub-Hamiltonians $H^a$; 
$a\in  \vA'$.
 Equivalently stated  the set of thresholds is 
\begin{equation*}
 \vT (H):= \bigcup_{a\in\vA'} \sigma_{\pupo}( H^a).
\end{equation*}
 This set 
is closed and countable. Moreover the set of non-threshold eigenvalues
  is discrete in $\R\setminus \vT (H)$,  and it 
 can only  accumulate  at  points in
$\vT (H)$  from below.  
The essential spectrum is given by the formula
$\sigma_{\ess}(H)= \bigl[\min \vT(H),\infty\bigr)$. 
We introduce the notation $\vT_{\p}(H)=\sigma_{\pp}(H)\cup
  \vT(H)$, and more generally   $\vT_{\p}(H^a)=\sigma_{\pp}(H^a)\cup
  \vT(H^a)$. Denote $R_a(z)=(H_a-z)^{-1}$ for
$z\notin \sigma(H_a)$ and $R(z)=R_{a_{\max}}(z)$.

\subsubsection{General 
  notation}\label{subsubsec:General 
  notation}

Any  function $f\in C^\infty_\c(\R)$ taking values in $[0,1]$ is
referred to as a  \emph{standard support function} (or just a `support function'). For  any such
functions  $f_1$ and $f_2$ we write
$f_2\succ f_1$, if $f_2=1$ in a neighbourhood of $\supp f_1$.

If $T$ is a self-adjoint  operator on a
Hilbert space and $f$ is a support function we can represent the
operator $f(T)$ by the well known Helffer--Sj\"ostrand formula
\begin{align}\label{82a0}
  \begin{split}
 f(T) 
&=
\int _{\C}(T -z)^{-1}\,\mathrm d\mu_f(z)\, \text{ with}\\
&\mathrm d\mu_f(z)=\pi^{-1}(\bar\partial\tilde f)(z)\,\mathrm du\mathrm dv;\quad 
z=u+\i v.   
  \end{split}
\end{align} Here $\tilde{f}$ is an `almost analytic' extension of
$f$, which in this case may be taken compactly supported. The
formula \eqref{82a0} extends to more general  classes of  functions
and  serves as a standard tool for commuting  operators. Since  it
will be used only tacitly in this paper  the interested  reader might benefit from
consulting 
 \cite[Section 6]{Sk3} which  is devoted to
applications of \eqref{82a0} to  $N$-body Schr\"odinger operators, hence being equally relevant for the present paper.

Consider   and fix $\chi\in C^\infty(\mathbb{R})$ such that 
\begin{align*}
\chi(t)
=\left\{\begin{array}{ll}
0 &\mbox{ for } t \le 4/3, \\
1 &\mbox{ for } t \ge 5/3,
\end{array}
\right.
\quad
\chi'\geq  0,
\end{align*} and such that the following properties are fulfilled:
\begin{align*}
  \sqrt{\chi}, \sqrt{\chi'}, (1-\chi^2)^{1/4} ,
  \sqrt{-\parb{(1-\chi^2)^{1/2} }'}\in C^\infty.
\end{align*} We define correspondingly $\chi_+=\chi$ and
$\chi_-=(1-\chi^2)^{1/2} $ and record that
\begin{align*}
  \chi_+^2+\chi_-^2=1\mand\sqrt{\chi_+}, \sqrt{\chi_+'}, \sqrt{\chi_-}, \sqrt{-\chi_-'}\in C^\infty.
\end{align*}

We shall use the notation 
$\inp{x}=\parb{1+\abs{x}^2}^{1/2}$ for  $x\in \mathbf X$ (or more generally for
any $x$ in a normed space).  For $x\in \mathbf X\setminus\set{0}$ we
write $\hat x=x/\abs{x}$. If $T$ is a self-adjoint  operator on a
Hilbert space $\vG$ and $\varphi\in \vG$ then
$\inp{T}_\varphi:=\inp{\varphi,T\varphi}$. We denote the space of  bounded operators 
from one  (abstract) Banach space $X$ to another one $Y$ by $\vL(X,Y)$ 
and abbreviate $\mathcal L(X,X)=\mathcal L(X)$. The dual space of $X$
is denoted by $X^*$.

To define \emph{Besov spaces 
associated with the multiplication operator
$|x|$ on $\vH$}  
let
\begin{align*}
F_0&=F\bigl(\bigl\{ x\in \mathbf X\,\big|\,\abs{x}<1\bigr\}\bigr),\\
F_m&=F\bigl(\bigl\{ x\in \mathbf X\,\big|\,2^{m-1}\le \abs{x}<2^m\bigr\} \bigr)
\quad \text{for }m=1,2,\dots,
\end{align*}
where $F(U)=F_U$ is the sharp characteristic function of any given  subset
$U\subseteq {\mathbf X}$. 
The Besov spaces $\mathcal B =\mathcal B(\mathbf X)$, $\mathcal
B^*=\mathcal B(\mathbf X)^*$ and $\mathcal B^*_0=\mathcal
B^*_0(\mathbf X)$ are then given  as 
\begin{align*}
\mathcal B&=
\bigl\{\psi\in L^2_{\mathrm{loc}}(\mathbf X)\,\big|\,\|\psi\|_{\mathcal B}<\infty\bigr\},\quad 
\|\psi\|_{\mathcal B}=\sum_{m=0}^\infty 2^{m/2}
\|F_m\psi\|_{{\mathcal H}},\\
\mathcal B^*&=
\bigl\{\psi\in L^2_{\mathrm{loc}}(\mathbf X)\,\big|\, \|\psi\|_{\mathcal B^*}<\infty\bigr\},\quad 
\|\psi\|_{\mathcal B^*}=\sup_{m\ge 0}2^{-m/2}\|F_m\psi\|_{{\mathcal H}},
\\
\mathcal B^*_0
&=
\Bigl\{\psi\in \mathcal B^*\,\Big|\, \lim_{m\to\infty}2^{-m/2}\|F_m\psi\|_{{\mathcal H}}=0\Bigr\},
\end{align*}
respectively.
Denote the standard \emph{weighted $L^2$ spaces} by 
$$
L_s^2=L_s^2(\mathbf X)=\inp{x}^{-s}L^2(\mathbf X)\ \ \text{for }s\in\mathbb R ,\quad
L_{-\infty}^2=\bigcup_{s\in\R}L^2_s,\quad
L^2_\infty=\bigcap_{s\in\mathbb R}L_s^2.
$$ 
Then for any $s>1/2$
\begin{equation*}
 L^2_s\subsetneq \mathcal B\subsetneq L^2_{1/2}
\subsetneq \mathcal H
\subsetneq L^2_{-1/2}\subsetneq \mathcal B^*_0\subsetneq \mathcal B^*\subsetneq L^2_{-s}.
\end{equation*} The abstract quotient-norm on the Banach space
$\vB^*/\vB_0^*$ and 
  \begin{equation*}
    \norm{\psi}_{\rm quo}:=\limsup_{n\to \infty}\,2^{-n/2}\norm[\Big]{\sum_{m=0}^n F_m\psi}_{{\mathcal H}},\quad \psi\in \vB^*,
  \end{equation*} are 
equivalent   norms.

Let $T$ be an  operator on $\mathcal H=L^2(\mathbf X)$ such that
$T,T^*:L^2_\infty\to L^2_\infty$, and let $t\in\mathbb R$.  Then we
say that   $T$ is an {\emph{operator of order $t$}}, if 
 for each $s\in\mathbb R$  the restriction  $T_{|L^2_\infty}$ extends to
 an operator $T_s\in\vL(L^2_{s}, L^2_{s-t})$. Alternatively stated,
 \begin{subequations}
  \begin{equation}\label{eq:defOrder}
\|\inp{x}^{s-t}T\inp{x}^{-s}\psi\|\le C_s\|\psi\| \text{ for all }\psi\in
L^2_\infty.
\end{equation} 
 If   $T$ is of {order $t$}, we write 
\begin{equation}
  T=\vO(\inp{x}^t).
\label{eq:1712022}
\end{equation}  
\end{subequations}
\begin{subequations}

For any given  $ a\in \vA\setminus\{a_{\min},a_{\max}\}$,
$\kappa>0$ (considered small)  and any given operator  $T$  of {order $t$}
we write
\begin{equation}
T=\vO^ {a}_\kappa(\inp{x}^t),
\label{eq:1712022kap}
\end{equation} 
if for some $\bre>0$ and some operator $\brT$ also of order $t$
\begin{equation}
  \label{eq:kapO}
  T-\chi_+(\abs{x})\chi_-(|x^{a}|/2\kappa|x|)\brT\chi_+(\abs{x})\chi_-(|x^{a}|/2\kappa|x|)=\vO(\inp{x}^{t-\bre}).
\end{equation}

 \end{subequations}

\subsection{Basic  $N$-body stationary  theory}\label{subsec:body preliminaries}
\begin{subequations}
Under a rather weak condition (in particular weaker  than Condition
\ref{cond:smooth2wea3n12}) it is demonstrated in  \cite{AIIS} that the
following limits exist   locally
                                                              uniformly
                                                              in $\lambda\not\in \vT_{\p}(H)$:
\begin{equation}\label{eq:LAPbnda}
  R(\lambda\pm \i
  0)=\lim _{\epsilon\to 0_+} \,R(\lambda\pm\i
    \epsilon)\in \vL\parb{L^2_s,L^2_{-s}}\text{ for any }s>1/2.
\end{equation} Furthermore  
    (the strong weak$^*$-topology is  explained in Theorem \ref{microLoc}) 
  \begin{align}\label{eq:BB^*a}
    \begin{split}
    	\sup_{0<\epsilon\le 1}\norm{R(\lambda\pm\i
    		\epsilon)}_{\vL\parb{\vB,\vB^*}}&<\infty \text{ with 
    		locally
    		uniform
    		bounds in 
    	}\lambda\not
    \in \vT_{\p}(H),\\
      &\text{and there exist the limits}\\
                                                              R(\lambda\pm \i
                                                              0)&=\swslim _{\epsilon\to 0_+} \,R(\lambda\pm\i
                                                              \epsilon)\in \vL\parb{\vB,\vB^*};\quad \lambda\not
                                                              \in \vT_{\p}(H).
    \end{split}
\end{align}
\end{subequations}
\begin{subequations}
 These results follow from a `Mourre estimate' for the operator $A=r^{1/2}Br^{1/2}$,
where
\begin{equation}\label{eq:Mourre}
   B=2\Re(\omega\cdot p)=-\i\sum_{j\leq d}\parb{\omega_j\partial_{x_j}+\partial_{x_j}\omega_j},\quad  \omega=\mathop{\mathrm{grad}} r   
\end{equation}  and  $r $ is   a certain  smooth positive
function  on $\mathbf X$   fulfilling   the
following property, cf.  \cite[Section 5]{Sk3}:

 For all $k\in\N_0:=\N\cup\set{0}$ and $\beta\in \N^{d_0}_0$   there
  exists $C>0$ such that 
\begin{equation}\label{eq:compa00}
  \abs[\big]{\partial ^\beta (x\cdot \nabla)^k \parb{r(x)-\inp{x}}} \leq C\inp{x}^{-1}.
\end{equation}

Fixing $\lambda\not\in \vT_{\p}(H)$ we can indeed find such a function
(possibly taken rescaled, meaning that we replace $r(x)$ by
$Rr(x/R)$ with  $R\geq 1$   large)  and  an
open neighbourhood $U$ of $\lambda$ such that
\begin{equation}\label{eq:mourreFull}
  \forall \text{ real } f\in C_\c^\infty(U):\quad{f}(H)
  \i[H,r^{1/2}Br^{1/2}]{f}(H)\geq
 c\,{f}(H)^2;
\end{equation} here the constant $c>0$ depends on the number
\begin{equation}\label{eq:optimal}
  d(\lambda,H)=4\dist\parb{\lambda ,\set{\mu\in\vT (H)\mid \mu<\lambda}}.
\end{equation} In fact $c$  can be taken arbitrarily smaller than
the number $d(\lambda,H)$ (upon correspondingly adjusting the scaling parameter $R$
and the  neighbourhood $U$). 

The Mourre estimate \eqref{eq:mourreFull} has  consequences beyond
\eqref{eq:LAPbnda} and 
\eqref{eq:BB^*a}. In particular, cf.   \cite [Subsection 5.2]{Sk3},   for any $\lambda\not\in
 \vT_{\p}(H)$ and  with $2\epsilon_0=\sqrt{d(\lambda,H)}$
    \begin{equation}\label{eq:MicroB0}
      R(\lambda\pm \i0)\psi-\chi_+( \pm B/\epsilon_0) R(\lambda\pm \i0)\psi\in \vB_0^*
      \text{ for all }\psi\in \vB.
    \end{equation}  For completeness of presentation we note that
    $\chi_-(  \pm B/\epsilon_0)\in \vL(\vB)$, cf.  \cite[Theorem
    14.1.4]{H1}. We will need various partially extended versions of
    \eqref{eq:MicroB0} from \cite{Sk3}, which will be stated in
    Appendix \ref{sec:Strong bounds}. See also the related bound
    \eqref{eq:2boundobtain338} stated below. 
    
Of course the above assertions for $H$ are also valid for the
Hamiltonians $H_a$, $a\in\vA'$, and by self-similarity there are
completely similar  assertions  for
the sub-Hamiltonains $H^a$ upon replacing $(H,r)$ by analogue pairs $(H^a,r^a)$.

\end{subequations}

We recall the following scheme of estimating the resolvent from \cite[Appendix B]{Sk3}.

\begin{lemma}\label{lemma:comSTA}
  Let $f_1$ and $f_2$  be standard support
  functions with $f_2\succ f_1$.  Suppose $P$ is a symmetric form on $\vH$, form-bounded relatively
to $H$, and given such that  the  bounded self-adjoint  operator
$\Psi=f_1(H)Pf_1(H)$ restricts  to a  bounded operator on $\vB$  and
such that
\begin{subequations}  
\begin{equation}\label{eq:Cyy}
  \i [H, \Psi]\geq f_1(H)\parbb{Q^*Q-T^*T}f_1(H) 
\end{equation} for $H$-bounded operators $Q$ and $T$ (possibly taking
values in an abstract Hilbert space). Then the following  estimates hold for all $z\in\C\setminus\R$ and
all $\psi\in \vB$.
 \begin{align*}
  \begin{split}
   \norm{&QR(z) f_1(H)\psi}^2\leq  \norm{TR(z) f_1(H)\psi}^2  \\&+2\parb{\norm{\Psi}_{\vL (\vH)}+\norm{\Psi}_{\vL (\vB)}}
  \norm{R(z)f_2(H)}_{\vL (\vB,\vB^*)}\norm{\psi}^2_{\vB} .
  \end{split}
\end{align*} In particular, if  also 
\begin{equation}\label{eq:conextra}
  \sup _{\Im z\neq 0}\norm{\abs{T{f_1} (H)}{R(z)}}_{\vL(\vB,\vH)}<
  \infty\mand  \sup _{\Im z\neq 0}\norm{R(z)f_2(H)}_{\vL (\vB,\vB^*)},
\end{equation} \end{subequations}
then  
\begin{equation}\label{eq:2boundobtain}
 \sup _{\Im z\neq 0}\norm{\abs{Q{f_1} (H)}{R(z)}}_{\vL(\vB,\vH)}< \infty.
\end{equation} 

 \end{lemma}
\begin{proof}  The estimate \eqref{eq:Cyy}
  leads for $\phi=R(z) f_1(H)\psi$ to
  \begin{align*}
    \norm{Q\phi}^2-\norm{T\phi}^2 &\leq 2\parb{(\Im
    z)\inp{R(z)\psi, \Psi R(z) \psi}  -\Im \inp{\psi, \Psi R(z) \psi}
                                    }\\
&\leq 2\parb{\norm{\Psi}_{\vL (\vH)} \abs{\Im\inp{ \psi, R(z)f_2(H)
  \psi}}  +\abs{\Im \inp{\Psi \psi, R(z)f_2(H) \psi}} }
\\
&\leq 2\parb{\norm{\Psi}_{\vL (\vH)} +\norm{\Psi}_{\vL (\vB)}}\norm{R(z)f_2(H)}_{\vL (\vB,\vB^*)}\norm{ \psi}^2_{\vB} .
  \end{align*}
\end{proof}

\begin{remarks} \label{remark:RQ}
  \begin{enumerate}[i)]
  \item \label{item:1Q}
The bound \eqref{eq:2boundobtain} 
 combines well  with the  limiting
absorption principle expressed by 
\eqref{eq:LAPbnda} and  \eqref{eq:BB^*a} (and similarly when applied
to $H_a$ rather than to $H$ for any $a\in
    \vA'$).   Note
for example that if $Q$ maps  into $\vH$, $Q{f_1} (H)=\vO(\inp{x}^0)$
(in the sense of \eqref{eq:1712022}) and $\lambda\not\in \vT_{\p}(H)$,
then  \eqref{eq:2boundobtain}  leads  to the existence of  the weak limits
\begin{equation}\label{eq:weakQ0}
{Q{f_1} (H)}R(\lambda\pm \i
  0)=\wlim _{\epsilon\to 0_+} Q{f_1} (H)R(\lambda\pm\i
    \epsilon)\in \vL\parb{\vB,\vH}. 
\end{equation}

\item \label{item:2Q} Bounds of the  form \eqref{eq:2boundobtain} are
   very useful for us (there will be several of those, see \ref{item:8Q} below). We shall use the generic notation $Q$ for
  $H$-bounded operators obeying \eqref{eq:2boundobtain}, in all cases
  with $f_1$ being  any  narrowly supported standard support
  function obeying  $f_1=1$ in a
neighbourhood of any given $\lambda\not\in \vT_{\p}(H)$.
  Such  $Q=Q_1$ may show up from a direct application  of  Lemma \ref{lemma:comSTA}, 
  or possibly the expression $Q_1^*Q_1$ may be bounded from above by $Q_2^*Q_2$ where
  the bound for $Q_2$ results from such  direct application (obviously
  yielding the bound for $Q_1$ also). In the latter case $Q_1f_1 (H)=\vO(\inp{x}^{-1/2})$ in the sense of
  \eqref{eq:1712022} (or possibly only for `components'),  and in   commutator calculations this product  will
   generically arise as  the combination 
  (for $a\in\vA$)
  \begin{equation}\label{eq:rel1and2}
   f_1 (H_a)Q_1^*B_aQ_1f_1 (H);\quad \, B_a\text{ bounded}.
  \end{equation} Such expressions will   be  treated by using
  the  bound  \eqref{eq:2boundobtain} with $Q=Q_1$ for
  $H$ (as well as for  $H$ replaced by $H_a$).

\item \label{item:3Q} The following
  basic example of a `$Q$-bound' does not fit into the scheme of Lemma
  \ref{lemma:comSTA}:
\begin{align}\label{eq:2boundLapRes}
  \begin{split}
   \sup _{\Im z\neq 0}\norm{Q_s{f_1}& (H)R(z)}_{\vL(\vB,\vH)}<
 \infty;\\&\quad Q_s=r^{-s},\,\, s\in (1/2,1). 
  \end{split}
\end{align} Here $ f_1$ is given as in
\ref{item:2Q}, and of course the bound  
 is  a consequence of \eqref{eq:BB^*a}.
\item \label{item:4Q} 
Another example of a bound of the  form \eqref{eq:2boundobtain} is
the following  version 
of \eqref{eq:MicroB0} from \cite{Sk3}.
We apply Lemma \ref{lemma:comSTA}  with
\begin{subequations}
\begin{equation}\label{eq:PropBasic8}
  \Psi=\pm f_1(H)\chi_+( \pm B/\epsilon) f_1(H),
\end{equation} where 
$\epsilon\in(0, \epsilon_0]$ and $ f_1$ is given as in
\ref{item:2Q}.  Then
\eqref{eq:LAPbnda}, \eqref{eq:BB^*a} and   \eqref{eq:mourreFull} lead
 to a version of  \eqref{eq:2boundobtain} from which we can deduce  the following concrete 
 bounds (with `prime' denoting the derivative), cf. \ref{item:2Q}:
\begin{align}\label{eq:2boundobtain338}
  \begin{split}
  \sup _{\Im z\neq 0}\norm{Q_{\epsilon\pm}  &{f_1}
   (H)R(z)}_{\vL(\vB,\vH)}< \infty;\\&\quad Q_{\epsilon\pm}
 =r^{-1/2}\sqrt{\chi_+'}( \pm B/\epsilon ),\,\, 0<\epsilon\leq
 \epsilon_0.
\end{split}  
  \end{align}
\item \label{item:8Q} A  complete list of `$Q$-bounds' to be applied
  in this paper to  expressions of the form  \eqref{eq:rel1and2}  may
  be provided by  
  supplementing  \eqref{eq:2boundLapRes} and
  \eqref{eq:2boundobtain338} stated above with the bounds \eqref{eq:2boundobtain33},
  \eqref{eq:2boundobtain33NOT}  and \eqref{eq:QanontreshB}.
\end{subequations}
 \end{enumerate}
\end{remarks}

\section{Yafaev type constructions}\label{sec: Yafaev type
  constructions}

In \cite{Ya3} various  real functions $m_a\in C^\infty(\mathbf X)$,
$a\in\vA'$,  and $m_{a_{\max}}\in C^\infty(\mathbf X)$  are 
constructed. They  are  homogeneous of degree $1$ for $\abs{x}> 5/6$,
and in addition $ m_{a_{\max}}$ is (locally) convex in that region. These
functions are constructed as depending  of a  small parameter
$\epsilon>0$, and once given, they  are  used in the constructions
\begin{equation}\label{eq:M_a0}
 M_a=2\Re(w_a\cdot p)=-\i\sum_{j\leq d}\parb{(w_a)_j\partial_{x_j}+\partial_{x_j}(w_a)_j};\quad  w_a=\mathop{\mathrm{grad}} m_a.   
\end{equation}
The operators $M_a$, $a\in\vA'$, may  be considered as `channel
localization operators', while the operator $M_{a_{\max}}$ enters as a technical quantity controlling  commutators of  
the Hamiltonian and the channel
localization operators. 
We are going to use and extend  these constructions/ideas  in 
partly generalized settings. 

A basic ingredient of \cite{Ya3} are the cones 
\begin{equation}\label{eq:collPla1}
  \mathbf
  X_a(\varepsilon)=\set{x\in\bX\mid\abs{x_a}>(1-\varepsilon)\abs{x}};\quad
  a\in\vA',\,
  \varepsilon \in (0,1).
\end{equation} Note that $\mathbf X_a \setminus\set{0}\subseteq
\mathbf X_a(\varepsilon)$ and that  in fact $\mathbf X_a
\setminus\set{0}=\cap_{\varepsilon\in(0,1)}\,\mathbf X_a(\varepsilon)$. Note also that $\mathbf X_0(\varepsilon)=\mathbf X\setminus \set{0}$.  

We are going to consider families of such cones having a  `width'
roughly being proportional to $\epsilon^{d_a}$ for a 
parameter $\epsilon>0$ (recall  $d_a=\dim
\mathbf X_a$). This parameter is taken sufficiently small   as primarily 
determined (to be elaborated on in  Subsection
\ref{subsec: Homogeneous Yafaev type functions})  by the following
geometric property: There exists $C>0$ such that for all $a,b \in \vA'$ and $x\in\mathbf X$
\begin{equation}\label{eq:3.2}
  \abs{x^c}\leq C\parb{\abs{x^a}+\abs{x^b}};\quad c=a\vee b.
  \end{equation}

\subsection{Homogeneous Yafaev type functions for a general $
  a_0$}\label{subsec: Homogeneous Yafaev type functions} 
In this subsection we  fix any $ a_0\in \vA\setminus\{a_{\min}\}$ and
 let $\vA_{ a_0}=\set{a\leq  a_0}$, $\vA'_{ a_0}=\set{a\lneq  a_0}$
   and for $a\in \vA'_{ a_0}$
\begin{equation}\label{eq:collPla2}
  \mathbf X^{a_0}_a(\varepsilon)=\set{x\in \bX\mid \abs{x^{a_0}_a}>(1-\varepsilon)\abs{x^{a_0}}};\quad
  \varepsilon \in (0,1).
\end{equation} Note that \eqref{eq:collPla2} coincides with
\eqref{eq:collPla1} if $ a_0=a_{\max}$. We record
\begin{equation}\label{eq:nearCol}
  \mathbf
  X^{a_0}_a(\varepsilon)\subseteq \set{\abs{x^{a}}< \sqrt{2\ \varepsilon}\abs{x^{a_0}}}.
  \end{equation}

For small $\epsilon>0$ and any $a\in \vA'_{ a_0}$ we define
\begin{equation*}
  \varepsilon^a_k= k\epsilon^{d_a},\quad k=1,2,3,
\end{equation*} and view the numbers  $\varepsilon_a$ in the
interval 
$(\varepsilon^a_2, \varepsilon^a_3)$  as 
`admissible'. Alternatively, viewed as a free parameter, $\varepsilon_a$ is
called admissible if 
\begin{equation}\label{eq:admis}
  \varepsilon^a_2=2\epsilon^{d_a}<\varepsilon_a<3\epsilon^{d_a}=\varepsilon^a_3.
\end{equation} Using an  arbitrary
ordering of $\vA'_{ a_0}$ we introduce for these  parameters  the `admissible' vector
\begin{equation*}
  \bar\varepsilon=(\varepsilon_{a_1}, \dots
\varepsilon_{a_n}),\quad n=n(a_0)= \# \vA'_{ a_0},
\end{equation*} and denote by $\d \bar\varepsilon$ the corresponding
Lebesgue measure. The smallness of $\epsilon>0$ will be needed at
various points below; it will be determined by geometric  considerations (primarily  \eqref{eq:3.2}) 
and in particular  taken independently of the considered vectors $x\in\mathbf X\setminus \mathbf X_{a_0}$.

Let for $\varepsilon>0$ and $a\in \vA'_{ a_0}$
\begin{equation*}
  h_{a,\varepsilon}(x)=h^{a_0}_{a,\varepsilon}(x)=\parb{(1+\varepsilon)^2\abs{x_a^{a_0}}^2+\abs{x_{a_0}^2}}^{1/2};\quad
  x\in \mathbf X\setminus \mathbf X_{a_0}.
\end{equation*} (We prefer for simplicity below to use the indicated short-hand notation $h_{a,\varepsilon}(x)$.)
\begin{lemma}\label{lem:Y3.1} Let $a,b\in \vA'_{ a_0}$, ${a\not\leq
    b}$ and $x\in\mathbf X^{a_0}_a(\varepsilon^a_1)$.  Then (for all
  small $\epsilon>0$)
  \begin{equation}
    \label{eq:psi1}
    h_{b,\varepsilon^b_3}(x)<\max_{\vA'_{ a_0}\ni  c\geq a } h_{c,\varepsilon^c_2}(x).
  \end{equation}
\end{lemma}
\begin{proof}
  Note that while here  $a,b,c\in \vA'_{ a_0}$ it could be that 
  $\tilde c:=a\vee b=a_0$.  By assumption $b\lneq \tilde c$. This leads us to consider three cases, depending
  on whether $\tilde c=a_0$
   or $\tilde c\in \vA'_{ a_0}$. We can assume that $\abs{x^{a_0}}=1$.

{\bf I}. Suppose $\tilde c\in \vA'_{ a_0}$ and  $x\in \mathbf X^{a_0}_{\tilde c}(\varepsilon^{\tilde c}_1)$.  We then verify that 
\begin{equation}
    \label{eq:psi2a}
    h_{b,\varepsilon^b_3}(x)< h_{{\tilde c},\varepsilon^{\tilde c}_2}(x),
  \end{equation} proving \eqref{eq:psi1}.  By assumptions
  \begin{equation*}
   \abs{x^{a_0}_b}\leq 1\mand \abs{x^{a_0}_{\tilde c}}>(1-\varepsilon^{\tilde c}_1).
  \end{equation*}  So it suffices to check that 
  \begin{equation*}
    (1+\varepsilon^b_3)\leq (1+\varepsilon^{\tilde c}_2)(1-\varepsilon^{\tilde c}_1).
  \end{equation*} Since $b\lneq \tilde c$, this is indeed valid for
  small $\epsilon> 0$.

{\bf II}. Suppose $\tilde c\in \vA'_{ a_0}$ and  $x\in \mathbf
X^{a_0}_a(\varepsilon^a_1)\setminus \mathbf
X^{a_0}_{\tilde c}(\varepsilon^{\tilde c}_1)$. We then verify that 
\begin{equation}
    \label{eq:psi2b}
    h_{b,\varepsilon^b_3}(x)< h_{{a},\varepsilon^{a}_2}(x),
  \end{equation} proving \eqref{eq:psi1}. Thanks to the bounds
\begin{equation*}
    \abs{x^{a_0}_{a}}>(1-\varepsilon^{a}_1)\mand \abs{x^{a_0}_{\tilde c}}\leq(1-\varepsilon^{\tilde c}_1),
  \end{equation*} we deduce
  \begin{equation*}
    \abs{x^{a}}^2<1 -(1-\varepsilon^{a}_1)^2<\varepsilon^{a}_2
\mand \abs{x^{\tilde c}}^2\geq 1 -(1-\varepsilon^{\tilde c}_1)^2\geq \varepsilon^{\tilde c}_1.
\end{equation*} Also it follows from our assumption on $x$ that  $a\lneq \tilde c$. Due to \eqref{eq:3.2} and these bounds we can find
$\kappa>0$ independent of small $\epsilon>0$ (and $x$)  such that
\begin{equation*}
  \abs{x^{a_0}_{b}}^2\leq 1-2\kappa \varepsilon^{\tilde c}_1.
\end{equation*} Then \eqref{eq:psi2b} is  trivially fulfilled, since 
\begin{equation*}
  (1+\varepsilon^{b}_3)(1-\kappa \varepsilon^{\tilde c}_1)<1 < (1+\varepsilon^{a}_2)   (1-\varepsilon^{a}_1).
\end{equation*}

  {\bf III}. Suppose $\tilde c=a_0$ and $x\in\mathbf
  X^{a_0}_a(\varepsilon^a_1)$. Then using  \eqref{eq:3.2} as before we   find
$\kappa>0$ independent of small $\epsilon>0$ such  that 
\begin{equation*}
  \abs{x^{a_0}_{b}}^2\leq 1-2\kappa.
\end{equation*} This leads  to \eqref{eq:psi2b}, noting  that similarly
\begin{equation*}
  (1+\varepsilon^{b}_3)(1-\kappa )< (1+\varepsilon^{a}_2)   (1-\varepsilon^{a}_1),
\end{equation*} and then in turn to \eqref{eq:psi1}.
\end{proof}

 Letting  $\Theta= 1_{[0,\infty)}$,  we define for  any $a\in \vA'_{
   a_0}$ and any admissible vector $\bar\varepsilon$
  \begin{equation*}
    m_a(x, \bar \varepsilon)= m^{a_0}_a(x, \bar \varepsilon)=h_{a,\varepsilon_a}(x)\Theta\parb{
      h_{a,\varepsilon_a}(x)-\max_{\vA'_{ a_0} \ni c\neq a } h_{c,\varepsilon_c}(x)};\quad
  x\in\mathbf X\setminus \mathbf X_{a_0}.
  \end{equation*} Here the second factor can be replaced by the product
  \begin{equation*}
    \Theta(\cdot)=\Pi_{ b\in \vA'_{ a_0}}\,\Theta\parb{ h_{a,\varepsilon_a}(x)-h_{b,\varepsilon_b}(x)}.
  \end{equation*}
\begin{lemma}\label{lem:Y3.2} For all  $a,b\in \vA'_{ a_0}$,
  $x\in\mathbf X^{a_0}_b(\varepsilon^b_1)$ and  admissible vectors $\bar\varepsilon$:
  \begin{enumerate}[1)]
  \item\label{item:Y3.2.1} If ${b\not\leq
    a}$, then $m_a(x, \bar \varepsilon)=0$.
  \item \label{item:Y3.2.2} If ${b\leq
    a}$, then $m_a(x, \bar \varepsilon)=m_a(x_b, \bar
  \varepsilon)$ and 
\begin{equation*}
    m_a(x, \bar \varepsilon)=h_{a,\varepsilon_a}(x)\Theta\parb{
      h_{a,\varepsilon_a}(x)-\max_{ \vA'_{ a_0}\ni c\geq b  } h_{c,\varepsilon_c}(x)}.
  \end{equation*}
  \end{enumerate}
\end{lemma}
\begin{proof} Fix $x\in\mathbf X^{a_0}_b(\varepsilon^b_1)$. 

By \eqref{eq:admis} and Lemma \ref{lem:Y3.1}, if  ${b\not\leq
    a}$
\begin{equation}
    \label{eq:psi12}
    h_{a,\varepsilon_a}(x)\leq h_{a,\varepsilon^a_3}(x)<\max_{\vA'_{ a_0}\ni c\geq b 
      } h_{c,\varepsilon^c_2}(x)\leq \max_{\vA'_{ a_0} \ni c\geq b } h_{c,\varepsilon_c}(x),
  \end{equation} showing \ref{item:Y3.2.1}.

If on the other hand ${b\leq
    a}$, then it suffices for \ref{item:Y3.2.2} to show that 
  \begin{equation*}
    \max_{c\in \vA'_{ a_0}} h_{c,\varepsilon_c}(x)=\max_{\vA'_{
    		a_0}\ni c\geq  b  } h_{c,\varepsilon_c}(x),
  \end{equation*} noting  that the right-hand side only depends on
  $x_b$. Obviously  it
  suffices to show that for any $\breve c\in \vA'_{ a_0}$ 
\begin{equation*}
    h_{\breve c,\varepsilon_{\breve c}}(x)\leq \max_{\vA'_{
    		a_0} \ni c\geq  b } h_{c,\varepsilon_c}(x),
  \end{equation*} and  we can assume that $ b\not\leq
  \breve c$. The stated  bound then  follows  from \eqref{eq:psi12} applied to
  $(a,b)=(\breve c,b)$.
\end{proof}

Next we  apply Lemma \ref{lem:Y3.2} \ref{item:Y3.2.2} to $b=a$. 
\begin{corollary}\label{cor:aligb} For any  $x\in\mathbf
  X^{a_0}_a(\varepsilon^a_1)$, $a\in \vA'_{ a_0}$, and any admissible vector $\bar\varepsilon$ 
  \begin{equation*}
    m_a(x, \bar \varepsilon)=h_{a,\varepsilon_a}(x)\Theta\parb{
      h_{a,\varepsilon_a}(x)-\max_{ c\in \vA'_{ a_0},\,c\geq a} h_{c,\varepsilon_c}(x)}.
  \end{equation*}
\end{corollary}

To average the functions $ m_a(x, \bar \varepsilon)$ over $\bar
\varepsilon$  we introduce for any  $c\in \vA'_{ a_0}$ a
non-negative function $\varphi_c\in C^\infty_\c(\R)$ with
\begin{equation*}
  \supp \varphi_c\subseteq (\varepsilon^c_2,\varepsilon^c_3)
\mand \int \varphi_c(\varepsilon)\,\d \varepsilon=1.
\end{equation*} We then define for each  $a\in \vA'_{ a_0}$
\begin{equation}\label{eq:defma}
  m_a(x)=m^{a_0}_a(x)=\int_{\R^n} m^{a_0}_a(x, \bar \varepsilon)\,\Pi_{c\in \vA'_{
      a_0}}\, \varphi_c(\varepsilon_c)\,\d \bar\varepsilon;\quad x\in\mathbf X\setminus \mathbf X_{a_0}.
\end{equation} Recall the notation $a_{\min}=0$.
\begin{lemma}\label{lem:3.4} Suppose $a\in \vA'_{ a_0}\setminus \set{a_{\min}}$,
   $x\not\in\mathbf
  X^{a_0}_a(\varepsilon^a_3)$ and $x^{a_0}\neq 0$. Then for any admissible vector $\bar\varepsilon$
  \begin{equation*}
    m_a(x)=m_a(x, \bar
  \varepsilon)=0.
  \end{equation*}
\end{lemma}
\begin{proof} We show that $m_a(x, \bar
  \varepsilon)=0$. With  our assumptions we can assume that $\abs{x^{a_0}}=1$
  and it suffices to show that 
  \begin{equation*}
    h_{a,\varepsilon_a}(x)-h_{{a_{\min}},\varepsilon_{a_{\min}}}(x)<0.
  \end{equation*} This bound follows from estimating 
  \begin{equation*}
    (1+\varepsilon_a)\abs{x_a^{a_0}}\leq
    (1+\varepsilon^a_3)\abs{x_a^{a_0}}\leq
    (1+\varepsilon^a_3)(1-\varepsilon^a_3)<  (1+\varepsilon_{a_{\min}}).
  \end{equation*}
\end{proof}
\begin{lemma}
  \label{lem:3.5} Let $a\in \vA'_{ a_0}$ and suppose $x\in\mathbf
  X^{a_0}_a(\varepsilon^a_1)$ obeys that for all  $b\in \vA'_{ a_0}$
  with  $b\gneq
  a $ the vector $x\not\in\mathbf
  X^{a_0}_b(\varepsilon^b_3)$. Then 
  \begin{equation*}
    m_a(x)=\int_\R h_{a,\varepsilon}(x)\,\varphi_a(\varepsilon)\,\d \varepsilon.
  \end{equation*}
\end{lemma}
\begin{proof}
  We check that $m_a(x, \bar
  \varepsilon)=h_{a,\varepsilon_a}(x)$, assuming that $\abs{x^{a_0}}=1$. Using Corollary
  \ref{cor:aligb} it suffices to show that
  \begin{equation*}
    h_{a,\varepsilon_a}(x)\geq\max_{ \vA'_{ a_0}\ni c\geq a } h_{c,\varepsilon_c}(x).
  \end{equation*}
In turn it suffices to show that for  any $b$ as in the lemma 
\begin{equation*}
    h_{a,\varepsilon_a}(x)\geq h_{b,\varepsilon_b}(x).
  \end{equation*} This bound follows from
  \begin{equation*}
    (1+\varepsilon_a)\abs{x_a^{a_0}}>
    (1+\varepsilon^a_2)(1-\varepsilon^a_1)\geq (1+\varepsilon^b_3)(1-\varepsilon^b_3)\geq (1+\varepsilon_b)\abs{x^{a_0}_b}.
  \end{equation*}
\end{proof}
\begin{lemma}\label{lem:3.6}
  For any $a\in \vA'_{ a_0}$ the function  $m_a\in C^\infty(\bX\setminus{\bX_{a_0}})$.
\end{lemma}
\begin{proof} By homogeneity it suffices to consider $x$ obeying
  $\abs{x^{a_0}}>2$. By Lemma \ref{lem:3.4} we
  can then also assume that $\abs{x^{a_0}_a}>1$.
  We represent
\begin{align}\label{eq:form00}
  \begin{split}
   m_a(x)=\int \,\d \varepsilon_a
  h_{a,\varepsilon_a}(x)&\,\varphi_a(\varepsilon_a)\, \\ &\prod_{  \vA'_{ a_0}\ni b\neq a}\,\parbb{\int\,\Theta\parb{ h_{a,\varepsilon_a}(x)-h_{b,\varepsilon_b}(x)}\varphi_b(\varepsilon_b)\, \d\varepsilon_b}. 
  \end{split}
\end{align} Clearly
\begin{equation*}
  \Theta\parb{ h_{a,\varepsilon_a}(x)-h_{b,\varepsilon_b}(x)}=\Theta\parb{ (1+\varepsilon_a)\abs{x_a^{a_0}}-(1+\varepsilon_b)\abs{x_b^{a_0}}},
\end{equation*} so introducing $\varPhi_b(t)=\int_0^t\,
\varphi_b(\varepsilon)\,\d \varepsilon$,
\begin{align*}
   m_a(x)&=\int \,\d \varepsilon_a
   h_{a,\varepsilon_a}(x)\,\varphi_a(\varepsilon_a)\, \prod_{ b\neq
     a}\,\parbb{\int\,\Theta\parb{
       (1+\varepsilon_a)\abs{x_a^{a_0}}/\abs{x_b^{a_0}}-1-\varepsilon_b)}\varphi_b(\varepsilon_b)\,
     \d\varepsilon_b}\\
&=\int \,\d \varepsilon_a h_{a,\varepsilon_a}(x)\,\varphi_a(\varepsilon_a)\, \Pi_{ b\neq a}\,\parb{\varPhi_b((1+\varepsilon_a)\abs{x_a^{a_0}}/\abs{x_b^{a_0}}-1)}.
\end{align*} This is a smooth in $x\in\set{\abs{x^{a_0}_a}>1}$.
\end{proof}

We summarize most of the derived  features as follows.

\begin{lemma}\label{lemma:ma1}
For any $a\in\vA'_{a_0}$ the  function 
 $m_a:\mathbf X\setminus \mathbf X_{a_0}\to \R$ fulfills  the following
properties for any sufficiently small  $\epsilon>0$ and any $b\in\vA'_{a_0}$: 
\begin{enumerate}[1)]
\item\label{item:10a} $m_a$ is homogeneous of degree $1$.
\item\label{item:11a} $m_a\in C^\infty (\bX\setminus{\bX_{a_0}})$.
\item\label{item:12a} If $b\leq a$  and  $x\in \mathbf X^{a_0}_b(\varepsilon^b_1)$, then
  $m_a(x)=m_a(x_b)$.
\item\label{item:13a} If ${b\not\leq a}$ and $x\in \mathbf
  X^{a_0}_b(\varepsilon^b_1)$, then   $m_a(x)=0$.
\item\label{item:14a} If $a\neq {a_{\min}}$,  $x\not\in \bX_{a_0}$  and $\abs{x^{a}}\geq \sqrt{2\ \varepsilon^a_3}\abs{x^{a_0}}$, then   $m_a(x)=0$.
\end{enumerate} 
\end{lemma} 
\begin{proof}
  The property \ref{item:10a} is obvious, and \ref{item:11a} coincides
  with Lemma \ref{lem:3.6}. The properties \ref{item:12a} and
  \ref{item:13a} follow from applying Lemma \ref{lem:Y3.2} to the
  defining expression \eqref{eq:defma}. The property \ref{item:14a}
  follows from Lemma \ref{lem:3.4} and \eqref{eq:nearCol}.
  \end{proof}

Following \cite{Ya3} we introduce for  $ x\in\mathbf X\setminus \mathbf X_{a_0}$
and admissible vectors $\bar\varepsilon$
\begin{align*}
   m_{a_0}(x,\bar\varepsilon)&=m_{a_0}^{a_0}(x,\bar\varepsilon)=\max_{a\in \vA'_{ a_0}} h_{a,\varepsilon_a}(x),\\  m_{a_0}(x)&=m_{a_0}^{a_0}(x)=\int_{\R^n} m_{a_0}^{a_0}(x, \bar \varepsilon)\,\Pi_{c\in \vA'_{
       a_0}}\, \varphi_c(\varepsilon_c)\,\d \bar\varepsilon.
\end{align*} 
\begin{lemma}
  \label{lem:decM} For all $ x\in\mathbf X\setminus \mathbf X_{a_0}$
  \begin{equation}
    \label{eq:decM0}
     m_{a_0}(x)=\Sigma_{a\in \vA'_{ a_0}}\,m_a(x).
  \end{equation}
\end{lemma}
  \begin{proof}
    To check \eqref{eq:decM0} it suffices by continuity and density to consider
    \begin{equation*}
      x\in \cap_{a\in \vA'_{a_0}}\set{ {x^{a_0}_a}\neq 0}\subseteq \mathbf X\setminus \mathbf X_{a_0}.
    \end{equation*}
 For any such
    $x$ there is a Lebesgue null-set of  admissible vectors  such that
    away  from  the   null-set
\begin{equation}
    \label{eq:decM}
     m_{a_0}(x,\bar\varepsilon)=\Sigma_{a\in \vA'_{ a_0}}\,m_a(x,\bar\varepsilon).
  \end{equation} Indeed for $a,b\in \vA'_{a_0}$, $a\neq
  b$, the relationship 
    $h_{a,\varepsilon_a}(x)=h_{b,\varepsilon_b}(x)$
   is only possible at a null-set in
   the variable $(\varepsilon_a,\varepsilon_b)$. Consequently we can for $\bar
   \varepsilon$ away from  a null-set find $a=a(x, \bar\varepsilon)\in\vA'_{ a_0}$ such that 
\begin{equation*}
    h_{a,\varepsilon_a}(x)> \max_{  \vA'_{ a_0}\ni b\neq a} h_{b,\varepsilon_b}(x).
  \end{equation*} Then $m_b(x,\bar\varepsilon)=0$ for  $b\neq a$ and
  $m_{a_0}(x,\bar\varepsilon)=m_a(x,\bar\varepsilon)$, showing
  \eqref{eq:decM}. Clearly \eqref{eq:decM0} follows from
  \eqref{eq:decM} by integration.
\end{proof}

We summarize  the known  features of $m_{a_0}$ as follows.
\begin{lemma}\label{lemma:m1} The  function 
 $m_{a_0}:\mathbf X\setminus \mathbf X_{a_0}\to \R$ fulfills  the following
properties for any sufficiently small  $\epsilon>0$: 
\begin{enumerate}[i)]
\item\label{item:10b} $m_{a_0}$ is convex and homogeneous of degree $1$.
\item\label{item:11b} $m_{a_0}\in C^\infty (\bX\setminus{\bX_{a_0}})$.
\item\label{item:12b} If $b\in\vA'_{ a_0}$  and  $x\in \mathbf X^{a_0}_b(\varepsilon^b_1)$, then
  $m_{a_0}(x)=m_{a_0}(x_b)$.
\item\label{item:9b} $m_{a_0}=\Sigma_{a\in \vA'_{ a_0}}\,m_a$.
\item\label{item:13b} Let $a\in \vA'_{ a_0}$ and suppose $x\in\mathbf
  X^{a_0}_a(\varepsilon^a_1)$ obeys that for all  $b\in \vA'_{ a_0}$
  with  $b\gneq
  a $ the vector $x\not\in\mathbf
  X^{a_0}_b(\varepsilon^b_3)$. Then 
  \begin{equation}\label{eq:goodEps}
    m_{a_0}(x)=m_a(x)=\int_\R h_{a,\varepsilon}(x)\,\varphi_a(\varepsilon)\,\d \varepsilon.
  \end{equation}
\item\label{item:14b} The derivative of $m_{a_0}\in C^\infty
  (\bX\setminus{\bX_{a_0}})$ is given by the formula
\begin{align}\label{eq:1derFirst}
  \begin{split}
   \quad \nabla m_{a_0}(x)\quad \quad &\\=\sum_{a\in \vA'_{ a_0}}\,\int &\, \d \varepsilon_a
   \parb{\nabla h_{a,\varepsilon_a}(x)}\,\varphi_a(\varepsilon_a)\,\\ & \prod_{\vA'_{ a_0}\ni b\neq a }\,\parbb{\int\,\Theta\parb{ h_{a,\varepsilon_a}(x)-h_{b,\varepsilon_b}(x)}\varphi_b(\varepsilon_b)\, \d\varepsilon_b}. 
  \end{split}
\end{align}
\item\label{item:15b} The derivative of $m_{a_0}\in C^\infty
  (\bX\setminus{\bX_{a_0}})$ obeys 
\begin{align}\label{eq:1derFirstb}
  \begin{split}
   \quad \quad\nabla \parb{m_{a_0}(x)-\abs{x}}&\\=\sum_{a\in \vA'_{ a_0}}\,\int
   \,
 \d \varepsilon_a&
   \nabla \parb{{h_{a,\varepsilon_a}(x)}-\abs{x}}\,\varphi_a(\varepsilon_a)\\  &\prod_{\vA'_{ a_0}\ni b\neq a    }\,\parbb{\int\,\Theta\parb{ h_{a,\varepsilon_a}(x)-h_{b,\varepsilon_b}(x)}\varphi_b(\varepsilon_b)\, \d\varepsilon_b}. 
  \end{split}
\end{align} In particular there exists $C>0$ (being independent of the
parameter  $\epsilon$) such that for all $x\in
\bX\setminus{\bX_{a_0}}$ 
\begin{equation}
  \label{eq:compa}
  \abs{\nabla \parb{m_{a_0}(x)-\abs{x}}}\leq C\sqrt{\epsilon}.
\end{equation}
\end{enumerate} 
\end{lemma}
\begin{proof}
  The property \ref{item:9b} coincides with \eqref{eq:decM0}. Since
  clearly $m_{a_0}(\cdot,\bar\varepsilon)$ is  convex, so is $m_{a_0}$, and apart from this \ref{item:10b} and
  \ref{item:11b} follow from \ref{item:9b} and Lemma \ref{lemma:ma1} \ref{item:10a} and
  \ref{item:11a}.  

  By Lemma \ref{lem:Y3.2} and \ref{item:9b}  it follows under the conditions of
  \ref{item:12b} that
  \begin{equation*}
     m_{a_0}(x)=\Sigma_{\vA'_{ a_0}\ni a\geq b}\,m_a(x_b)=\Sigma_{\vA'_{ a_0}\ni a}\,m_a(x_b)= m_{a_0}(x_b).
  \end{equation*} Here we used  the  fact that $x^{a_0}_b\neq 0$ since $x\in \mathbf
  X^{a_0}_b(\varepsilon^b_1)$, and hence trivially also $x_b\in \mathbf
  X^{a_0}_b(\varepsilon^b_1)$. We have verified \ref{item:12b}.

  As for \ref{item:13b} we use Lemma \ref{lem:3.5}. It suffices to
  show that $m_{a_0}(x)=m_a(x)$ under the given conditions. Usng again \ref{item:9b} we need to show
  that $m_b(x)=0$ for any $b\neq a$. If  ${b\not\geq
    a}$, then
  $m_b(x)=0$ thanks to Lemma \ref{lemma:ma1} \ref{item:13a}. If $b\gneq
  a $, we can apply 
  Lemma \ref{lem:3.4}, and consequently we conclude  again that $m_b(x)=0$. 

As for \ref{item:14b} obviously \eqref{eq:1derFirst} is related  to
\ref{item:9b} and 
\eqref{eq:form00}  (it  may be considered as a variant of \cite[Lemma
7.6]{GT}).  First we note  that  the proof of the smoothness of the 
right-hand side 
expression of \eqref{eq:form00}
shows that also the right-hand side of \eqref{eq:1derFirst} is 
smooth.  Hence by continuity and density it suffices to verify
\eqref{eq:1derFirst} as a formula for the  distributional derivative of $m$ as a function on the set
\begin{equation*}
  \cap_{a,b\in \vA'_{a_0},\,a\neq b}\set{x\mid
  {x^{a_0}_a}\mand x^{a_0}_b \text{ span a two-dimensional subspace}}.
\end{equation*}
  So let us consider a small open neighbourhood
$U$ of any point $y$ in this set. Pick a basis $e_1,\dots,e_d$ of $\bX$
such that for all $j=1,\dots,d$ and all different $a,b\in \vA'_{a_0}$
the two vectors 
\begin{equation}\label{eq:lineInd1}
  (e_j\cdot
  y^{a_0}_a,e_j\cdot y^{a_0}_b),
  \,(\abs{
    y^{a_0}_a}^2, \abs{
    y^{a_0}_b}^2)\in\R^2\,\text{ are linear independent.}
\end{equation}
To see this is doable it suffices to show the linear independence for
a single vector $e_1$ (since then we can supplement with vectors $e_2,
\dots e_d$ close to $e_1$). In turn (by the same reasoning) it
suffices to show that \eqref{eq:lineInd1} can be realized for $j=1$
and a fixed pair $(a,b)$. The latter is   obviously  doable for a suitable
vector in  $ \mathrm{span} (y^{a_0}_a,y^{a_0}_b)$. Upon possibly  shrinking $U$
\eqref{eq:lineInd1} holds with $y$ replaced by an arbitrary  $x\in U$, implying that for all
different $a,b\in \vA'_{a_0}$ and all admissible vectors
$\bar\varepsilon$
\begin{equation*}
  U\cap\set{h_{a,\varepsilon_a}(x)=h_{b,\varepsilon_b}(x)}\text{ is a
    regular hypersurface.}
\end{equation*} In fact for any $j=1,\dots,d$, $a\neq b$ and $\bar\varepsilon$
\begin{equation}\label{eq:goodVaria}
  x\in
  U\cap\set{h_{a,\varepsilon_a}(x)=h_{b,\varepsilon_b}(x)}\Rightarrow \partial_j
  h_{a,\varepsilon_a}(x)\neq \partial_j
  h_{b,\varepsilon_b}(x);\quad \partial_j:=e_j\cdot \nabla. 
\end{equation} Here and below we  represent the variable $x$ in terms of the
constructed basis, make the corresponding linear change of
variables and note that in  the new coordinates $\nabla$ is represented as
$(\partial_1,\dots, \partial_d)$. With \eqref{eq:goodVaria} at disposal
we can  easily  for each $j=1,\dots,d$
do the associated one-dimensional
integration by parts using the Fubini theorem (varying  in  the
direction of $e_j$  while fixing the   other variables, including
$\bar\varepsilon$). The obtained formulas amount to
\eqref{eq:1derFirst}.

As for \ref{item:15b} the assertion \eqref{eq:1derFirstb} follows from a similar reasoning
as given above for \eqref{eq:1derFirst}, and \eqref{eq:compa} follows
from an easy direct computation and estimation of
\eqref{eq:1derFirstb}, cf. Lemma \ref{lemma:ma1} \ref{item:14a} (and
its proof).
\end{proof}

\subsection{Geometric considerations for $a_0=a_{\max}$}\label{subsec: Geometric considerations for a0=a{max}}
Throughout this subsection we consider the case $a_0=a_{\max}$
only. For this example the functions $m_a$ from Lemma \ref{lemma:ma1}
and $m_{a_0}$ from Lemma \ref{lemma:m1} are smooth on
$\bX\setminus{0}$. We can introduce smooth modifications by
multiplying them by a suitable factor, say specifically by the factor
$\chi_+(2|x|)$. We adapt these modifications and will use (slightly
abusively) the same notation $m_a$ and $m_{a_0}$ for the smoothed out
versions.  We may then consider the corresponding first order
operators $M_a$ and $M_{a_0}$ from \eqref{eq:M_a0} realized as self-adjoint
operators. In our application, given in Section \ref{sec::A partition
  of unity}, these operators will be used to construct a suitable
phase-space partition of unity. In the present subsection we provide
technical details on how to control their commutator with the
Hamiltonian $H$.

We need  to consider various conical
subsets of  $\mathbf X\setminus{0}$ (recalling the basic example \eqref{eq:collPla1}).
 Let for  $a\in\vA'$ and $\varepsilon,\delta\in (0,1)$
\begin{align}\label{eq:primes}
  \begin{split}
    \mathbf X'_a&={\mathbf X_a }\setminus\cup_{{  b\gneq a,\,b\in \vA} }\,\mathbf
                 X_b={\mathbf X_a }\setminus\cup_{{b\not\leq a,\,b\in
                     \vA}\,
                 }\mathbf X_b,\\
\mathbf X_a(\varepsilon)&=\set{x\in \bX\mid\abs{x_a}>(1-\varepsilon)\abs{x}},\\
\mathbf \Gamma_a(\varepsilon)&=\parb{\mathbf X \setminus\set{0}}\setminus\cup_{{b\not\leq a,\,b\in \vA'}
}\,\mathbf X_b(\varepsilon),\\
\mathbf Y_a(\delta)&=\mathbf X_a(\delta)\setminus\cup_{b\gneq a,\,b\in \vA'}
\,\overline{\mathbf X_b(3\delta^{1/d_a})}.
\end{split}
\end{align} Here and henceforth the overline means topological closure in
$\mathbf X$. 
The structure of the sets $\mathbf X_a(\varepsilon)$, $\mathbf
\Gamma_a(\varepsilon)$ and $\mathbf Y_a(\varepsilon)$ is
${\R_+ V}$, where  $V$ is a
subset of the unit sphere $\mathbf{S}^{a_{\max}}=\S^{d-1}$ in
$\mathbf X$. For $\mathbf
X_a(\varepsilon)$ and $\mathbf Y_a(\varepsilon)$ the set $V$ 
 is relatively open, while for  $\mathbf
\Gamma_a(\varepsilon)$ the  set is  compact.

\begin{lemma}\label{lemma:cover} For  any $a\in\vA'$ and $\varepsilon\in (0,1)$
\begin{equation}\label{eq:Gamma2}
   \mathbf \Gamma_a(\varepsilon)\subseteq\cup_{b\leq a}  \,\mathbf X'_b.
\end{equation} 
  \end{lemma}
  \begin{proof}
    For any $x\in 
\mathbf \Gamma_a(\varepsilon)$ we introduce $b(x)\leq a$ by  $\mathbf X_{b(x)}=\cap_{b\leq
a,\,x\in \mathbf X_b} \,\mathbf X_b$ and then check that $x\in \mathbf
X'_{b(x)}$: If not, $x\in \mathbf X_b \cap \mathbf X_{b(x)}$ for some
$b$ with $b(x)\lneq b\leq a$. Consequently 
    $\mathbf X_b=\mathbf X_b\cap \mathbf X_{b(x)}=\mathbf X_{b(x)}$,  contradicting that
  $b\neq b(x)$. 
  \end{proof}

 Thanks to   \eqref{eq:Gamma2} we can for  any $a\in\vA'$ and  any  $\varepsilon, \delta_0\in (0,1)$  write
 \begin{subequations}
  \begin{equation}\label{eq:deltaNBHa}
    \mathbf \Gamma_a(\varepsilon)\subseteq\cup_{b\leq a} \cup_{\delta\in (0,\delta_0]} \,\mathbf Y_b(\delta).
  \end{equation} 
  By compactness \eqref{eq:deltaNBHa} leads for
  any fixed $\varepsilon,\delta_0\in(0,1)$ to the existence of a finite
  covering
\begin{equation} \label{eq:deltaNBHb}
     \mathbf \Gamma_a(\varepsilon)\subseteq\cup_{j\leq J} \,\,\mathbf Y_{b_j}(\delta_j),
  \end{equation} where  $\delta_1,\dots, \delta_J\in
(0,\delta_0]$ and  $b_1,\dots, b_J\leq a$. Here and in the following we suppress
 the dependence of quantities on the given $a\in\vA'$.
 \end{subequations}

To make contact to  Subsection \ref{subsec: Homogeneous Yafaev
  type functions} we recall   that  the
functions $m_a$, $a\in \vA'$, depend on a
positive parameter $\epsilon$ and  by Lemma  \ref{lemma:ma1}
\ref{item:13a}  fulfill 
\begin{equation}\label{eq:suppa}
  \supp m_a\subseteq 
\mathbf \Gamma_a(\epsilon^{d}).
\end{equation}
    Having  Lemmas \ref{lemma:ma1} and \ref{lemma:m1} at our
    disposal we also recall   that the (small) positive  parameter
    $\epsilon$ appears independently in the two lemmas. First we choose and fix the same   small $\epsilon$ for the  lemmas.
     Then we need   applications of Lemma \ref{lemma:m1} for $\epsilon_1,\dots,\epsilon_J\leq \epsilon$
     fixed  as follows. The latter devices will be used to control commutators, cf. Remark \ref{remark:RQ} \refeq{item:11b}.

For each  $a\in\vA'$ we apply \eqref{eq:deltaNBHb} with
    $\varepsilon=\delta_0=\epsilon^{d}$
    (cf. \eqref{eq:suppa}). Since $\delta_j\leq \delta_0$, we can
    introduce positive $\epsilon_1,\dots,\epsilon_J\leq \epsilon$ by  the
    requirement $\epsilon_j^{d_{b_j}}=\delta_j$. The 
    inputs $\epsilon=\epsilon_j$ in Lemma \ref{lemma:m1}  yield
    corresponding   functions, say denoted   $m_j$. In
    particular in the region $\bY_{b_j}(\delta_j)$ the function $m_a$
    from Lemma \ref{lemma:ma1} only
    depends on $x_{b_j}$ (thanks to   Lemma
    \ref{lemma:ma1} \ref{item:12a} and the  property
    $\bX_{b_j}(\delta_j)\subseteq \bX_{b_j}(\epsilon^{d_{b_j}})$),
    while (thanks to Lemma \ref{lemma:m1} \ref{item:13b} and recalling  $a_0= a_{\max}$)
    \begin{subequations}
    \begin{equation}\label{eq:good0}
      m_j(x)=C_j\abs{x_{b_j}},\quad C_j=\int_\R (1+\varepsilon)\,\varphi_{b_j}(\varepsilon)\,\d
      \varepsilon.
    \end{equation}  Obviously $\abs{y'}$
    is
    non-degenerately  convex in $y'\in\bX_{b_j}\setminus \set{0}$, meaning that  the  restricted Hessian
    \begin{equation}\label{eq:HesRestri0}
      \parb{\nabla_{y'}^2
        \abs{y'}}_{|\bX_{b_j}\cap\set{y}^\perp}\text{ is
      positive definite at any }y\in \bX_{b_j}\setminus\set{0}.
    \end{equation}   
    \end{subequations}

These properties   can be applied as follows  using   
the vector-valued first order operators 
\begin{equation}\label{eq:Ffield}
  G_a=\vH_a(x_a)\cdot
  p_a,\text{ where } \vH_a(x_a)=\chi_+(2\abs{x_a})\abs{x_a}^{-1/2}\parb{I-\abs{x_a}^{-2}\ket{ x_a}\bra{ x_a}}.
\end{equation}  
 We  choose for the considered   $a\in\vA'$ 
   a quadratic partition $\xi_1,\dots, \xi_J\in
C^\infty(\mathbf{S}^{a_0})$ (viz $\Sigma_j \,\xi_j^2=1$)
subordinate to the covering \eqref{eq:deltaNBHb} (recalling the discussion before
    Lemma \ref{lemma:cover}). Then
we can write, using  the support properties \eqref{eq:suppa} and \eqref{eq:deltaNBHb}, 
\begin{align*}
  m_a(x)=\Sigma_{j\leq J}\,\,m_{a,j}(x);\quad m_{a,j}(x)=\xi^2_j( \hat x) m_a(x),\quad \hat x=x/\abs{x},
\end{align*} and from the previous discussion it follows  that 
\begin{align*}
  \chi^2_+(\abs{x})m_{a,j}(x)=\xi_j^2( \hat x)\chi^2_+(\abs{x}) m_{a}(x_{b_j}),
\end{align*} as well as 
\begin{subequations}
\begin{align}\label{eq:Hes0}
  \begin{split}
 &p\cdot\parb{\chi^2_+(\abs{x})\nabla^2m_a(x)} p=\Sigma_{j\leq J}\,\,
                                 p\cdot\parb{\xi^2_j( \hat x)\chi^2_+(\abs{x}) \nabla^2m_a(x_{b_j})}p,\\&
=\Sigma_{j\leq J}\, \,G^*_{b_j}\parb{\xi^2_j( \hat x)\chi^2_+(\abs{x})\vG_j}G_{b_j};\quad 
                                                                  \vG_j=\vG_j(x_{b_j})\text{ 
                                                                  bounded}.   
  \end{split}
\end{align}
In turn using the convexity property of $m_j$ and the previous
discussion (cf. \eqref{eq:good0} and
\eqref{eq:HesRestri0})  we conclude that 
\begin{align}\label{eq:Hes_est}
   G^*_{b_j}\xi^2_j( \hat x)\chi^2_+(\abs{x})
  G_{bj}\leq  2 p\cdot \parb{\chi^2_+(\abs{x})\nabla^2m_j(x)} p.
\end{align}

Here the right-hand side  is the `leading term' of  the commutator
$\tfrac 12 \i[
p^2,M_j]$, where $M_j$ is given by \eqref{eq:M_a0} for  the
modification of 
$m_j$ given by the  function $\chi_+(2|x|)m_j(x)$. More precisely, in the sense of \eqref{eq:1712022}, for any real $f\in C_\c^\infty(\R)$
\begin{align}\label{eq:1712022apll}
	\begin{split}
		f(H)\parbb{2 p\cdot \parb{\chi^2_+(\abs{x})\nabla^2m_j(x)} p-\tfrac 12 \i[
H,&\chi_+(\abs{x})M_j\chi_+(\abs{x})] }f(H) \\&=\vO(\inp{x}^{-1-2\mu}).
\end{split}
\end{align}  
\end{subequations} We note that a basic  application of
\eqref{eq:Hes0}--\eqref{eq:1712022apll} comes from combining these
features with Lemma \ref{lemma:comSTA} applied concretely with
\begin{subequations}
\begin{equation}\label{eq:PropBasic}
  \Psi=\Psi_j=\tfrac 12 f_1(H)\chi_+(\abs{x})M_j\chi_+(\abs{x})f_1(H),
\end{equation} where 
$ f_1$ is any  narrowly supported standard support
  function obeying  $f_1=1$ in a
neighbourhood of a given $\lambda\not\in \vT_{\p}(H)$. This leads
in fact to \eqref{eq:2boundobtain} with
\begin{equation}\label{eq:HessionPos}
  {\abs{Qf_1 (H)}}^2=2f_1(H){ p\cdot \parb{\chi^2_+(\abs{x})\nabla^2m_j(x)} p }f_1(H),
\end{equation}  and therefore in turn  to the `$Q$-bounds' 
\begin{align}\label{eq:2boundobtain33}
  \begin{split}
  \sup _{\Im z\neq 0}\norm{Q_j{f_1}&
   (H){R(z)}}_{\vL(\vB,\vH^{d})}< \infty;\\
&\quad Q_j=\xi_j( \hat x)\chi_+(\abs{x})G_{b_j},\,\,j\le J.
 \end{split}
\end{align} 
\end{subequations}

 Finally  we introduce for $a\in\vA'$ functions $\xi^+_a$ and $\tilde\xi^+_a$ as follows.
First choose any $\xi_a\in
C^\infty(\mathbf{S}^{a_0})$ (recall $a_0=a_{\max}$) such that $\xi_a=1$ in
$\mathbf{S}^{a_0}\cap\mathbf  \Gamma_a(\varepsilon)$ and $\xi_a=0$ on
$\mathbf{S}^{a_0}\setminus  \mathbf \Gamma_a(\varepsilon/2)$. Choose then any $\tilde\xi_a\in
C^\infty(\mathbf{S}^{a_0})$ using this recipe   with $\varepsilon$
replaced by  $\varepsilon/2$. Finally  let  $\xi^+_a(x)=\xi_a(\hat
x)\chi_+(4\abs{x})$ and  $\tilde\xi^+_a(x)=\tilde\xi_a(\hat
x)\chi_+(8\abs{x})$,  and note  that
$\tilde\xi^+_a\xi^+_a=\xi^+_a$. Applied to $\varepsilon=\epsilon^{d}$ it follows from Lemma
\ref{lemma:ma1} \ref{item:13a} that the channel localization operators $M_a$ fulfill   
\begin{align}\label{eq:partM}
  M_a= M_a\xi^+_a= \xi^+_aM_a=M_a\tilde\xi^+_a=\tilde \xi^+_aM_a,
\end{align} which in applications  provides `free factors' of
$\xi^+_a$ and $\tilde\xi^+_a$
where  convenient. In particular it will be  useful for replacing $H$ by $H_a$ (or vice versa) in the presence of a factor $M_a$.

\subsection{Geometric considerations for
  $a_0\neq a_{\max}$}\label{subsec: Geometric considerations a0 ej}
For $a_0\neq a_{\max}$ the functions $m^{a_0}_a$, $a\in \vA'_{ a_0}$,
from Lemma \ref{lemma:ma1} and $m_{a_0}^{a_0}$ from Lemma
\ref{lemma:m1} are smooth on $\bX\setminus{\bX_{a_0}}$ however not on
$\bX\setminus\set{0}$. We  introduce smooth modifications by
multiplying them by a factor of $\chi_+(2|x|)$ (as before) as well as
by a factor of $\chi_+(2|x^{a_0}|/\kappa|x|)$ for a small parameter
$\kappa>0$. Hence we consider 
\begin{subequations}
 \begin{equation}\label{eq:modmfuncs}
  m_{a,\kappa}(x)=m_{a,\kappa}^{a_0}(x)=\chi_+(2|x|)
  \chi_+(2|x^{a_0}|/\kappa|x|)m^{a_0}_a(x); 
\quad a\in \vA_{ a_0}.
\end{equation}
Note that \eqref{eq:modmfuncs} applies to $a\in \vA'_{ a_0}$ as well
as to $a=a_0$.
We may then consider the corresponding first order
operators \eqref{eq:M_a0}, say denoted by
\begin{equation}\label{eq:Masss}
  M_{a,\kappa}=M^{a_0}_{a,\kappa}\mand M_{{a_0},\kappa}=M^{a_0}_{{a_0},\kappa},
\end{equation} 
\end{subequations}
again realized as self-adjoint operators.  They will be used in
Section \ref{sec:Proof of stationary completeness} in a similar way as for
the case $a_0= a_{\max}$ discussed in Subsection \ref{subsec:
  Geometric considerations for a0=a{max}}, however now in combination
with an additional input from Section \ref{sec:Non-threshold
  analysis}.

Mimicking  the procedure of Subsection \ref{subsec: Geometric considerations for
  a0=a{max}}  we 
consider the following  conical subsets of $\mathbf X\setminus{\bX_{a_0}}$.
 Let for  $a\in\vA_{a_0}'$ and $\varepsilon,\delta\in (0,1)$
\begin{align}\label{eq:primes2}
  \begin{split}
    \mathbf X^{'a_0}_a&={\mathbf X_a }\setminus\cup_{{  b\gneq a,\,b\in \vA_{a_0}} }\,\mathbf
                 X_b={\mathbf X_a }\setminus\cup_{{b\not\leq a,\,b\in
                     \vA_{a_0}}\,
                 }\mathbf X_b,\\
\mathbf X^{a_0}_a(\varepsilon)&=\set{x\in \bX\mid\abs{x^{a_0}_a}>(1-\varepsilon)\abs{x^{a_0}}},\\
\mathbf \Gamma^{a_0}_a(\varepsilon)&=\parb{\mathbf X \setminus{\bX_{a_0}}}\setminus\cup_{{b\not\leq a,\,b\in \vA_{a_0}'}
}\,\mathbf X^{a_0}_b(\varepsilon),\\
\mathbf Y^{a_0}_a(\delta)&=\mathbf X^{a_0}_a(\delta)\setminus\cup_{b\gneq a,\,b\in \vA_{a_0}'}
\,\overline{\mathbf X^{a_0}_b(3\delta^{1/d_a})}.
\end{split}
\end{align}

The structure of the sets $\mathbf X^{a_0}_a(\varepsilon)$, $\mathbf
\Gamma^{a_0}_a(\varepsilon)$ and $\mathbf Y^{a_0}_a(\varepsilon)$ is
$\parb{\R_+ V^{a_0}}\oplus \mathbf X_{a_0}$, where  $V^{a_0}$ is a
subset of the unit sphere $\mathbf{S}^{a_0}$ in $\mathbf X^{a_0}$. For $\mathbf
X^{a_0}_a(\varepsilon)$ and $\mathbf Y^{a_0}_a(\varepsilon)$ the set
$V^{a_0}\subseteq \mathbf{S}^{a_0} $ is (relatively) open, while for  $\mathbf
\Gamma^{a_0}_a(\varepsilon)$ the  set $V^{a_0}$ is compact. 

In particular we obtain from Lemma \ref{lemma:cover} (now applied to $a_0$
 rather than to 
$a_{\max}$) 
\begin{equation}\label{eq:Gamma22}
   \text{For  any } a\in\vA_{a_0}' \mand \varepsilon\in (0,1):\quad \mathbf \Gamma^{a_0}_a(\varepsilon)\subseteq\cup_{b\leq a}  \,\mathbf X^{'a_0}_b.
\end{equation}  
Thanks to   \eqref{eq:Gamma22} we can for  any $a\in\vA_{a_0}'$ and  any  $\varepsilon, \delta_0\in (0,1)$  write
 \begin{subequations}
  \begin{equation}\label{eq:deltaNBHa2}
    \mathbf \Gamma^{a_0}_a(\varepsilon)\subseteq\cup_{b\leq a} \cup_{\delta\in (0,\delta_0]} \,\mathbf Y^{a_0}_b(\delta),
  \end{equation}  and then  compactness  leads to the existence of a finite covering 
\begin{equation}\label{eq:deltaNBHb2}
     \mathbf \Gamma^{a_0}_a(\varepsilon)\subseteq\cup_{j\leq J} \,\,\mathbf Y^{a_0}_{b_j}(\delta_j),
  \end{equation} where  $\delta_1,\dots, \delta_J\in
(0,\delta_0]$ and  $b_1,\dots, b_J\leq a$.
 \end{subequations}

Again, to  make contact to  Subsection \ref{subsec: Homogeneous Yafaev
  type functions}, we recall   that  the
functions $m^{a_0}_a$, $a\in \vA'_{a_0}$, depend on a
positive parameter $\epsilon$ and  by Lemma  \ref{lemma:ma1}
\ref{item:13a}  fulfill 
\begin{equation}\label{eq:suppa2}
  \supp m^{a_0}_a\subseteq 
\mathbf \Gamma^{a_0}_a(\epsilon^{d}).
\end{equation}
    Having  Lemmas \ref{lemma:ma1} and \ref{lemma:m1} at our
    disposal we also recall   that the (small) positive  parameter
    $\epsilon$ appear independently in the two lemmas. As before  we first choose and fix the same   small $\epsilon$ for the  lemmas. 
     Then we need  applications of Lemma \ref{lemma:m1} along the  pattern of that in  Subsection \ref{subsec: Geometric considerations for a0=a{max}}.
      Hence for each  $a\in\vA_{a_0}'$   we apply \eqref{eq:deltaNBHb2} with
    $\varepsilon=\delta_0=\epsilon^{d}$
    (cf. \eqref{eq:suppa2}). Since $\delta_j\leq \delta_0$, we can
    introduce positive $\epsilon_1,\dots,\epsilon_J\leq \epsilon$ by  the
    requirement $\epsilon_j^{d_{b_j}}=\delta_j$. The 
    inputs $\epsilon=\epsilon_j$ in Lemma \ref{lemma:m1}  yield
    corresponding   functions, say denoted   $m_j$. In
    particular in the region $\bY^{a_0}_{b_j}(\delta_j)$ the function $m_a$
    from Lemma \ref{lemma:ma1} only
    depends on $x_{b_j}$ (thanks to   Lemma
    \ref{lemma:ma1} \ref{item:12a} and the  property
    $\bX^{a_0}_{b_j}(\delta_j)\subseteq \bX^{a_0}_{b_j}(\epsilon^{d_{b_j}})$),
    while (thanks to  Lemma \ref{lemma:m1} \ref{item:13b})
    \begin{subequations}
    \begin{equation}\label{eq:good}
      m_j(x)=\int_\R
      h_{b_j,\varepsilon}(x)\,\varphi_{b_j}(\varepsilon)\,\d
      \varepsilon.
    \end{equation} It is an important (elementary) 
     property   that
    $h_{b_j,\varepsilon}(x)=h_{b_j,\varepsilon}(x_{b_j})=h_{b_j,\varepsilon}(y')$ is
    non-degenerately  convex in $y'\in\bX_{b_j}\setminus \set{0}$, meaning that  the  restricted Hessian
    \begin{equation}\label{eq:HesRestri}
      \parb{\nabla_{y'}^2
        h_{b_j,\varepsilon}(y')}_{|\bX_{b_j}\cap\set{y}^\perp}\text{ is
      positive definite at any }y\in \bX_{b_j}\setminus \set{0}.
    \end{equation}  Note that the  corresponding eigenvalues are close
    to one (and hence positive) thanks to the smallness of
    $\epsilon_j\,(\leq \epsilon)$. Obviously the  non-degeneracy holds as well for the expression
    $ m_j(x_{b_j})$ given in \eqref{eq:good}.  
    \end{subequations}

Choose
 a quadratic partition $\xi_1,\dots, \xi_J\in
    C^\infty(\mathbf{S}^{a_0})$ (viz $\Sigma_j \,\xi_j^2=1$) subordinate to the
    covering \eqref{eq:deltaNBHb2} (recalling the discussion before \eqref{eq:Gamma22}. 
Then
we can write, using  the support properties \eqref{eq:suppa2} and
\eqref{eq:deltaNBHb2} and recalling the definition of
$m_{a,\kappa}^{a_0}$ and $m_{a_0,\kappa}^{a_0}$ from \eqref{eq:modmfuncs} as well as the notation $\hat x=x/\abs{x}$, 
\begin{equation*}
  m_{a,\kappa}(x)=m^{a_0}_{a,\kappa}(x)=\Sigma_{j\leq J}\,\,m_{a,j,\kappa}(x);\quad m_{a,j,\kappa}(x)=m^{a_0}_{a,j,\kappa}(x)=\xi^2_j(\widehat{x^{a_0}}) m^{a_0}_{a,\kappa}(x).
\end{equation*}
Thanks to Lemma \ref{lemma:ma1} \ref{item:12a} and \ref{item:14a}
\begin{align*}
  \chi^2_{\kappa+}(x) m_{a,j,\kappa}(x)=&\xi^2_j(\widehat{x^{a_0}}) \chi^2_{\kappa+}(x) m_{a,\kappa}(x_{b_j}),\\&\chi_{\kappa+}(x):=\chi_+(\abs{x})\chi_+(|x^{a_0}|/\kappa|x|),
\end{align*} and we can also record (recalling the definition
\eqref{eq:Ffield}) that 
\begin{subequations}
\begin{align}\label{eq:Hes02}
  \begin{split}
 &p\cdot\parb{\chi^2_{\kappa+}(x)\nabla^2m_{a,\kappa}(x)} p=\Sigma_{j\leq J}\,\,
                                 p\cdot\parb{\xi^2_j(\widehat{x^{a_0}})\chi^2_{\kappa+}(x) \nabla^2m_{a,\kappa}(x_{b_j})}p,\\&
=\Sigma_{j\leq J}\, \,G^*_{b_j}\parb{  \xi^2_j(\widehat{x^{a_0}})\chi^2_{\kappa+}(x)  \vG_{j,\kappa}}G_{b_j};\quad 
                              \vG_{j,\kappa}=\vG_{j,\kappa}(x_{b_j})\text{ 
                                                                  bounded}.   
  \end{split}
\end{align}
In turn using the convexity property of $m_j$ and the previous
discussion (cf. \eqref{eq:good} and
\eqref{eq:HesRestri})  we conclude that \begin{equation}\label{eq:Hes_est2}
   G^*_{b_j}\xi^2_j(\widehat{x^{a_0}})\chi^2_{\kappa+}(x)
  G_{b_j}\leq  2 p\cdot \parb{\chi^2_{\kappa+}(x)\nabla^2m_j}(x) p.
\end{equation}

Here the right-hand side  is the `leading term' of  the commutator
$\tfrac 12 \i[
p^2,M_{j,\kappa}]$, where $M_{j,\kappa}$ is given by \eqref{eq:M_a0} for the
modification of 
$m_j$ given by the  function
\begin{equation*}
  \chi_+(2|x|)\chi_+(2|x^{a_0}|/\kappa|x|)m_j(x).
\end{equation*}
  More precisely, in the sense of \eqref{eq:1712022} and \eqref{eq:1712022kap}, for any real $f\in C_\c^\infty(\R)$
\begin{align}\begin{split}
	f(H_{a_0})\parbb{2 p\cdot \parb{\chi^2_{\kappa+}(x)&\nabla^2m_j}(x) p-\tfrac 12 \i[
H_{a_0},\chi_{\kappa+}(x) M_{j,\kappa}\chi_{\kappa+}(x)] }f(H_{a_0})\\ &=\vO(\inp{x}^{-1-2\mu})+\vO^a_\kappa(\inp{x}^{-1}).	
\end{split}
\label{eq:1712022apll22}
\end{align} 
\end{subequations} 

\begin{subequations}
Due to the appearance of the second term to the right we can not directly use Lemma \ref{lemma:comSTA} as we derived \eqref{eq:2boundobtain33} from \eqref{eq:Hes_est} and \eqref{eq:1712022apll}. We shall in Section \ref{sec:Proof of stationary completeness} remedy this point for $H_{a_0}$, $a_0\in\vA\setminus\{a_{\min},a_{\max}\}$, using the  `$Q$-bound' \eqref{eq:QanontreshB} from Section \ref{sec:Non-threshold analysis}. With this bound and a required energy localization for $H^{a_0}$ the  scheme of proof of Lemma \ref{lemma:comSTA}  applies again, leading to  `$Q$-bounds' on the form
\begin{align}\label{eq:2boundobtain33NOT}
	\begin{split}
		\sup _{\Im z\neq 0}\norm{Q_{j,\kappa}{f_1}&
			(H_{a_0}){R_{a_0}(z)}}_{\vL(\vB,\vH^{d})}< \infty;\\
		&\quad Q_{j,\kappa}=\xi_j(\widehat{x^{a_0}}) \chi_{\kappa+}(x)G_{b_j}f(H^{a_0}),\,\,j\le J.
	\end{split}
\end{align} 
Here  $f_1$ is given as in \eqref{eq:2boundobtain33} and $f$ is any standard support function 
 supported in $\R\setminus \vT_\p(H^{a_0})$. Once $f$ is fixed (we
 explain in Sections  \ref{sec:Squeezing argument} and \ref{sec:Proof
   of stationary completeness}  how to  choose it for our purposes), 
 the parameter     $\kappa>0$ is  fixed by \eqref{eq:QanontreshB}. More precisely  \eqref{eq:2boundobtain33NOT}  then follows from  applying  Lemma \ref{lemma:comSTA} with 
 \begin{equation}\label{eq:PropBasicNOTT}
 	\Psi=\Psi_{j,\kappa}=\tfrac 12 f_1(H_{a_0})f(H^{a_0})\chi_{\kappa+}(x) M_{j,\kappa}\chi_{\kappa+}(x)f(H^{a_0})f_1(H_{a_0}),
 \end{equation}  using \eqref{eq:BB^*a} and  \eqref{eq:QanontreshB} to
 verify the inputs \eqref{eq:Cyy} and \eqref{eq:conextra}.
  
\end{subequations}

 Finally  we introduce for $a\in\vA'$ functions $\xi^+_a$ and $\tilde\xi^+_a$ as follows.
First choose any $\xi_a\in
C^\infty(\mathbf{S}^{a_0})$ such that $\xi_a=1$ in
$\mathbf{S}^{a_0}\cap\mathbf  \Gamma^{a_0}_a(\varepsilon)$ and $\xi_a=0$ in
$\mathbf{S}^{a_0}\setminus  \mathbf \Gamma^{a_0}_a(\varepsilon/2)$. Choose then any $\tilde\xi_a\in
C^\infty(\mathbf{S}^{a_0})$ using this recipe   with $\varepsilon$
replaced by  $\varepsilon/2$. Finally  let  $\xi^+_{a,\kappa}(x)=\xi_a(\widehat{x^{a_0}})\chi_+(4\abs{x})\chi_+(4|x^{a_0}|/\kappa|x|)$ and  $\tilde\xi^+_{a,\kappa}(x)=\tilde\xi_a(\widehat{x^{a_0}})\chi_+(8\abs{x})\chi_+(8|x^{a_0}|/\kappa|x|)$,  and note  that
$\tilde\xi^+_a\xi^+_{a,\kappa}=\xi^+_{a,\kappa}$. Applied to $\varepsilon=\epsilon^{d}$ it follows from Lemma
\ref{lemma:ma1} that the  channel localization operators
$M_{a,\kappa}$ from  \eqref{eq:Masss} fulfill   
\begin{align}\label{eq:partM2}
  M_{a,\kappa}= M_{a,\kappa}\xi^+_{a,\kappa}= \xi^+_{a,\kappa}M_{a,\kappa}=M_{a,\kappa}\tilde\xi^+_{a,\kappa}=\tilde \xi^+_{a,\kappa}M_{a,\kappa},
\end{align} which in applications will provide `free factors' of
$\xi^+_{a,\kappa}$ and $\tilde\xi^+_{a,\kappa}$
where  convenient.

\section{$N$-body short-range scattering theory}\label{sec:-body effective potential and a $1$-body radial limit}

For any  $a\in  \vA'$ we introduce 
\begin{equation*}
  K_a(x_a,\lambda)=\sqrt
      \lambda \abs{x_a}\mand S_a(\xi_a,t)=t\xi^2;\quad (x_a,\lambda),(\xi_a,t)\in{\mathbf
X_a}\times \R_+.
\end{equation*}

For any  \emph{channel}
$\alpha =(a,\lambda^\alpha, u^\alpha)$, i.e.   $a\in\vA'$, $\lambda^\alpha\in \R$ and  
$(H^a-\lambda^\alpha)u^\alpha=0$ for a normalized $u^\alpha\in
\mathcal H^a$, we introduce the corresponding \emph{channel wave operators} 
\begin{equation}\label{eq:wave_op}
  W_\alpha^{\pm}=\slim_{t\to \pm\infty}\e^{\i tH}J_\alpha\e^{-\i
  ( S_a^\pm (p_a,t)+\lambda^\alpha t)},\quad J_\alpha \varphi
=u^\alpha\otimes
\varphi,
 \end{equation} where $p_a=-\i \nabla_{x_a}$ and $ S_a^\pm
 (\xi_a,\pm|t|)=\pm S_a(\pm\xi_a,|t|)=t\xi_a^2$. (With  long-range potentials included, not considered in this paper, the functions $K_a$ and $S_a$ need to be chosen slightly different, see \cite{Is, II,Sk3}.)  We outline a proof of the existence of the
 limits \eqref{eq:wave_op} in  Remark
 \ref{remark:R} \ref{item:2Sk} following \cite{Sk3, Ya3}. It is a general fact  that the existence of the channel 
wave operators  implies their orthogonality, see for example
\cite[Theorem XI.36]{RS}.

Let us for  any  channel
$\alpha =(a,\lambda^\alpha, u^\alpha)$ introduce the notation
\begin{equation*}
k_\alpha=p_a^2+\lambda^\alpha, \quad I^\alpha=(\lambda^\alpha,\infty)\quad\text{and}\quad \vE^\alpha=I^\alpha \setminus {
  \vT_{\p}(H)}.
\end{equation*} 
 Note  the intertwining property $H W_\alpha^{\pm}\supseteq
W_\alpha^{\pm}k_\alpha $  and the fact that $k_\alpha$ is diagonalized by
the unitary map  $F_\alpha:L^2(\mathbf X_a)\to L^2(I^\alpha
;\vG_a)$, $\vG_a=L^2(\mathbf{S}_a)$,  $\mathbf{S}_a=\mathbf X_a\cap\S^{d_a-1}$ with $d_a=\dim
\mathbf X_a$,  given by
\begin{align}\label{eq:Four}
  \begin{split}
  (F_\alpha \varphi)(\lambda,\omega)&=(2\pi)^{-d_a/2}2^{-1/2}
  \lambda_\alpha^{(d_a-2)/4}\int \e^{-\i  \lambda^{1/2}_\alpha \omega\cdot
  x_a}\varphi(x_a)\,\d
x_a;\\&\quad \quad\lambda_\alpha=\lambda-\lambda^\alpha,\quad \lambda\in I^\alpha.  
\end{split}
\end{align}

We denote
\begin{equation*}
 c_\alpha^\pm(\lambda)=\e^{\pm \i
  \pi (d_a-3)/4}\pi ^{-1/2}\lambda_\alpha^{1/4},
\end{equation*} and $F_\rho=F(\set{x\in \bX\mid\abs{x}<\rho})$ for
$\rho>1 $ (considered below as  multiplication operators).
\begin{proposition}[\cite {Sk3}]\label{prop:radi-limits-chann22} 
  For any channel $\alpha=(a,\lambda^\alpha, u^\alpha)$,
  $\lambda\in \vE^\alpha $, $\psi\in \vB$ and $g\in \vG_a$
  there exist the  weak limits 
  \begin{align}\label{eq:restrH}
    \begin{split}
   \inp{\Gamma^\pm_{\alpha}&(\lambda)\psi,g}=\lim_{\rho\to
    \infty}\,\overline{c_\alpha^\pm(\lambda)}\rho^{-1}\\&
    \inp[\big]{F_\rho R(\lambda\pm  \i 0)\psi, F_\rho \parb{u^\alpha \otimes
    \abs{x_a}^{(1-d_a)/2}
                                      \e^{\pm \i
                                         K_a(\abs{x_a}\cdot,\lambda_\alpha)}g(\pm\cdot)}}.   
    \end{split}
  \end{align}  Here    the   limits $\Gamma^\pm_{\alpha}(\lambda)$
  are  
weakly continuous $\vL(\vB , \vG_a)$-valued functions of  $\lambda\in\vE^\alpha$.  
\end{proposition}

The restrictions of the map $F_\alpha (W^\pm_\alpha)^*$
have  
strong almost everywhere interpretations, meaning more precisely
\begin{align*}
  F_\alpha (W^\pm_\alpha  )^*\psi= \int^\oplus_{
  I^{\alpha} }\parb{F_\alpha (W^\pm_\alpha)^*\psi}(\lambda)\,\d
  \lambda;\quad \psi\in \vH.
\end{align*} When applied to $\psi\in \vB \subseteq
\vH $ the following relationship to Proposition  \ref{prop:radi-limits-chann22} holds. 

\begin{thm}[{\cite{Sk3}}]\label{thm:wave_matrices} For any channel
  $(a,\lambda^\alpha, u^\alpha)$  and  any $\psi\in \vB \subseteq
\vH$ 
\begin{equation}\label{eq:adjFORM}
  \parb{F_\alpha
  (W^\pm_\alpha)^*\psi}(\lambda)=\Gamma_\alpha^\pm(\lambda)\psi\quad\text{for a.e.
  } \lambda\in \vE^\alpha.
\end{equation} In particular  for any  $\psi\in \vB $ the restrictions $\parb{F_\alpha (W^\pm_\alpha)^*\psi}(\cdot)$ are
weakly continuous $\vG_a$-valued functions on  $\vE^\alpha$.  
\end{thm}

\begin{defn}\label{defn:scatEnergy07}  
An  energy $\lambda
  \in \vE:=(\min \vT(H),\infty)\setminus\vT_\p(H)$  is  {stationary
    complete}  for $H$ if  
\begin{equation}\label{eq:ScatEnergy222331}
  \forall \psi\in  L^2_\infty:\,\, \sum_{\lambda^\beta<
  \lambda}\,\norm{\Gamma_\alpha^\pm(\lambda) \psi}^2= 
  \inp{\psi,\delta(H-\lambda){\psi}}.
\end{equation} 
\end{defn} Asymptotic completeness follows by integration provided
(\ref{eq:ScatEnergy222331}) is known for almost all $\lambda \in \vE$
(motivating the used terminology). The orthogonality of the
 channel wave operators \eqref{eq:wave_op}  implies (as demonstrated in
\cite[Subsection 9.2]{Sk3})  that
  \begin{equation}\label{eq:Besn}
     \forall \lambda
  \in \vE,\,\forall \psi\in  \vB:\,\, \sum_{\lambda^\beta<
  \lambda}\,\norm{\Gamma_\alpha^\pm(\lambda) \psi}^2\leq 
  \inp{\psi,\delta(H-\lambda){\psi}}.
  \end{equation} 

It is also known (see \cite[Proposition
9.16]{Sk3})  that a \emph{sufficient
  and necessary condition} for $\lambda\in \vE$ be stationary complete is
  given as follows:

For all $\psi\in  L^2_\infty$ there exists
  $(g_\beta)_{\beta}\in \vG:=\Sigma^\oplus_{\beta }\,\vG_b$  (here 
  $\beta=(b, \lambda^\beta, u^\beta)$ runs  over all channels)  such that,  as an 
    identity in  $\vB^*/\vB_0^*$ (equipped with the quotient topology),
    \begin{equation}\label{eq:asres29}
       R(\lambda+\i 0)\psi=2\pi \i \sum_{\lambda^\beta<\lambda} J_\beta
     v^{+}_{\beta,\lambda} [g_\beta],
\end{equation} where (recalling) $J_\beta \phi
=u^\beta\otimes
\phi$, here   with  $\phi$  taken to be the \emph{outgoing one-body $\beta$-
  quasi-mode} corresponding to the plus case of
\begin{equation}\label{eq:quasiM}
  v^{\pm}_{\beta,\lambda} [g]( x_b):= \mp \tfrac \i{2\pi} \parb{c_\beta^\pm(\lambda)}^{-1}\chi_+(\abs{x_b}) \abs{x_b}^{(1-d_b)/2}
                                      \e^{\pm \i
                                         K_b(x_b,\lambda_\beta)}g(\pm\hat
                                           x_b);\quad g\in \vG_b. 
\end{equation} 

It is an elementary property
 that the right-hand side of \eqref{eq:asres29} is well-defined in
$\vB^*/\vB_0^*$ for any
$(g_\beta)_{\lambda^\beta<\lambda}\in \vG$, cf. \eqref{eq:eqBnd0}. In fact,   cf. \cite
[(9.22)]{Sk3}, the  quasi-modes enjoy  an `orthogonality property'
entailing
\begin{equation}\label{eq:pars}
    \norm[\Big]{\sum_{\lambda^\beta<\lambda} J_\beta
    v^{+}_{\beta,\lambda} [g_\beta]}_{\rm
    quo}^2=\sum_{\lambda^\beta<\lambda} \parb{4\pi
    \lambda_\beta^{1/2}}^{-1}\,\norm{g_\beta}^2.
  \end{equation}

From a practical point of view
(here explaining our strategy of verifying \eqref{eq:asres29}) it
suffices for any given $\lambda
  \in \vE,\,\psi\in  L^2_\infty$ and $\epsilon>0$ to produce a vector
  $(g^\epsilon_\beta)_{\lambda^\beta<\lambda}\in \vG$ such that 
\begin{equation*}
    \norm[\big]{R(\lambda+\i 0)\psi-2\pi \i \sum_{\lambda^\beta<\lambda} J_\beta
     v^{+}_{\beta,\lambda} [g^\epsilon_\beta]}_{\rm
    quo}\leq \epsilon.
  \end{equation*}

It is  known  that if
\eqref{eq:asres29} holds for some $(g_\beta)_{\lambda^\beta<\lambda}\in \vG$, then necessarily
$g_\beta=\Gamma_\beta^+(\lambda) \psi$. 

The scattering matrix
$S(\lambda)=\parb{S_{\beta\alpha}(\lambda)}_{\beta\alpha}$ is a
priori   given for almost all $\lambda
  \in \vE$ by 
\begin{equation*}
  \hat
  S_{\beta\alpha}:=F_\beta(W_{\beta}^+)^*W_{\alpha}^-F_\alpha^{-1}=\int^\oplus_{
  I_{\beta\alpha} }
  S_{\beta\alpha}(\lambda)\,\d \lambda,\quad
  I_{\beta\alpha}=I^\beta\cap I^\alpha.
\end{equation*} (For $\lambda\notin I^\beta\cap I^\alpha$ we let
$S_{\beta\alpha}(\lambda)=0$.)
 It is known  that $S(\cdot)$ is 
a  weakly continuous $\vL(\vG)$-valued function (in fact
contraction-valued) on  $\vE$. At  stationary complete energies
the scattering matrix is characterized geometrically as follows.
 \begin{thm}[{\cite{Sk3}}]\label{Cor:besov-space-setting}  Let $\lambda\in
  \vE$ be    stationary complete  and $\alpha=(a,\lambda^\alpha, u^\alpha)$ be any
  channel with $\lambda^\alpha<\lambda$. Then the following
  existence and uniqueness results hold for any  $
  g\in \vG_a$.
  \begin{enumerate}[1)]
  \item \label{item:As10} Let   $
    u=\Gamma^-_{\alpha}(\lambda)^* g$,  and let 
    $(g_\beta)_{\lambda^\beta<\lambda}\in \vG$ be   given by
    $g_\beta=S_{\beta\alpha}(\lambda) g$. Then,  as an
    identity in  $\vB^*/\vB_0^*$,
    \begin{align}\label{eq:as}
       u=J_\alpha
    v^{-}_{\alpha,\lambda} [ g]+\sum_{\lambda^\beta<\lambda} J_\beta
    v^{+}_{\beta,\lambda} [g_\beta].
    \end{align} 
  \item \label{item:As20} Conversely, if  \eqref{eq:as}
  is fulfilled for some  $ u\in \vB^*\cap H^2_{\mathrm{loc}}(\bX)$
  with $(H-\lambda)u=0$
  and for some $(g_\beta)_{\lambda^\beta<\lambda}\in \vG$, 
  then  $ u=\Gamma^-_{\alpha}(\lambda)^* g$ and
  $g_\beta=S_{\beta\alpha}(\lambda) g$.
\end{enumerate}
\end{thm}
\begin{thm}[{\cite{Sk3}}]\label{thm:strongC} \begin{enumerate}[1)]
  \item\label{item:1str} For any channel $\alpha$ 
    \begin{enumerate}[a)]
    \item \label{item:g} the operators
      $\Gamma^\pm_{\alpha}(\lambda)\in \vL(\vB, \vG_a)$
      are strongly continuous at any stationary complete energy
      $\lambda\in\vE^\alpha$.
    \item \label{item:f} the operators
      $\Gamma^\pm_{\alpha}(\lambda)^*\in \vL( \vG_a,L^2_{-s})$, $s>1/2$,
      are strongly continuous in 
      $\lambda\in\vE^\alpha$.
\end{enumerate}
\item\label{item:2str} The $\vL(\vG)$-valued function
  $S(\lambda)=\parb{S_{\beta\alpha}(\lambda)}_{\beta\alpha}$  is
  strongly continuous at any stationary  complete $\lambda\in
  \vE$. Moreover the restriction of $S(\lambda)$ to
  $\Sigma^\oplus_{\lambda^\beta <\lambda}\,\vG_b$ (i.e. considered as
  an operator on the energetically
  open sector of $\vG$) is unitary at  any such energy $\lambda$.
\end{enumerate}
\end{thm}
\begin{remarks} \label{remark:R}
  \begin{enumerate}[i)]
  \item \label{item:1Sk}
Although we have referred to various
  results of \cite{Sk3}, strictly speaking the condition on the
  short-range potentials is somewhat stronger than those of Condition
  \ref{cond:smooth2wea3n12}. However the conditions of \cite{Sk3} are
  used to deal with the long-range case and the stronger condition
  imposed on the short-range pair-potentials appearing there are clearly not needed
  if the long-range part of the potentials is absent.

\item \label{item:2Sk}
One  can  show the well-definedness of  the channel wave operators 
 \eqref{eq:wave_op} using \cite{Sk3}. Indeed to show the existence of the limits,
 \begin{equation*}
   \lim_{t\to \pm\infty}\e^{\i tH}J_\alpha\e^{-\i
  (  S_a^\pm (p_a,t)+\lambda^\alpha t)}\varphi=\lim_{t\to
  \pm\infty}\e^{\i tH} \e^{-\i t H_a}J_\alpha\varphi,
 \end{equation*} it suffices to consider $\varphi$ localized as 
 $\varphi=f_1(k_\alpha)\varphi$ with $f_1$ being  any  narrowly supported standard support
 function obeying  $f_1=1$ in a
 neighbourhood of any  given $\lambda_0\not\in \vT_{\p}(H)$.
  Furthermore we can assume that
 the Fourier transform $\hat \varphi$ is supported away from
 collision planes, and with this assumption the proof reduces to the existence of the limits
 \begin{equation}\label{eq:aux}
   \slim_{t\to \pm\infty}\,{f_2}(H)\e^{\i tH}
  M_a \e^{-\i tH_a}{f_2}(H_a),
 \end{equation} where $M_a$ are  given by 
\eqref{eq:M_a0} for  functions  $m_a$ constructed as in Section
\ref{subsec: Geometric considerations for a0=a{max}} and $f_2\succ
f_1$,  also given supported  near  $\lambda_0$. At this point
the freedom of adjusting the small parameter $\epsilon>0$ in Lemmas
\ref{lemma:ma1} and \ref{lemma:m1} is important, cf. the beginning of \cite[Section
9]{Sk3}. Using now Section
\ref{subsec: Geometric considerations for a0=a{max}} the limits are
obtained by  Kato smoothness bounds, obtained by invoking the
`propagation observable' 
\eqref{eq:PropBasic} along with 
\eqref{eq:Hes0}--\eqref{eq:1712022apll}  and  the limiting
absorption principle \eqref{eq:LAPbnda}, applicably for $H$ as well as for   $H_a$.  This procedure is explained  in detail   for a more complicated setting in  \cite[Section 7.2]{Sk3}. 
 \end{enumerate}
\end{remarks}

\section{A phase-space partition of unity}\label{sec::A partition of unity}

Let  $\lambda\notin \vT_\p(H)$  and   $\psi\in L^2_\infty$  be given. 
  Recall that we aim at  deriving  the asymptotics of $R(\lambda+\i0)\psi$ in agreement with \eqref{eq:asres29}.

  In this section we explain the first step in a procedure that in
  Section \ref{sec:Proof of stationary completeness} will
  be  iterated
    to
  derive \eqref{eq:asres29}. As the reader will  see a certain `effective 
  partition of unity' $I\simeq \Sigma_{a\in \vA'}\,S_a$, will allow us to reduce the problem of
  asymptotics to that of $R_a(\lambda+\i0)\psi$, $a\in \vA'$. We let
  $a_0=a_{\max}$ and consider the corresponding operator $M_{a}$ from
  Subsection \ref{subsec: Geometric considerations for a0=a{max}}.
 
Recalling  \eqref{eq:MicroB0} 
  the problem is  reduced to  the asymptotics of $ \chi_+( B/\epsilon_0)R(\lambda+\i0)\psi$. The problem is further reduced thanks to the following elementary estimate, cf.  \cite[Lemma
  8.1]{Sk3} and its proof.
\begin{lemma}\label{lem:compBogB}
 Let the  positive parameter
  $\epsilon$ in the construction of $M_{a_0}$ be sufficiently small,
  and let $f\in C_\c^\infty(\R)$ be given. Then 
  \begin{equation}
    \label{eq:small}
    \chi_-(2M_{a_0}/\epsilon_0)\chi_+(B/\epsilon_0)f(H)=\vO(\inp{x}^{-1/2}).
  \end{equation}
\end{lemma}
\begin{proof} We can assume $f=f_1$ is a standard support
  function. Let $f_2$ be another support function  with 
$f_2\succ f_1$.

 Introducing 
  \begin{equation*}
    S=f_2(H)\chi_-(2M_{a_0}/\epsilon_0)
  \chi_+(B/\epsilon_0)\inp{x}^{1/2}f_1(H),
  \end{equation*} it suffices  to be show that
  $S=\vO(\inp{x}^{0})$.

  We estimate for any $\brp\in
  L^2_\infty\subseteq \vH$  by   repeated commutation using tacitly \eqref{82a0} (recall also  the
generic notation
$\inp{T}_{\varphi}=\inp{\varphi,T\varphi}$):
  \begin{align} \label{eq:commsmall0}
  \begin{split}
&\tfrac{\epsilon_0}6 \norm {{S\brp}}^2\\
&\leq \inp{\epsilon_0 I-M_{a_0}
  }_{S\brp}+C_1\norm{\brp}^2 \\
&\leq \inp{(\epsilon_0+C\sqrt{\epsilon}) I-B
  }_{{S\brp}}+C_2\norm{\brp}^2 \quad \quad \quad (\text{by
  }\eqref{eq:compa00} \mand \eqref{eq:compa})\\&
\leq \inp{(\epsilon_0+C\sqrt{\epsilon}
  -\tfrac 43\epsilon_0) I}_{{S\brp}}+C_3\norm{\brp}^2 
\\&\leq C_3\norm{\brp}^2\quad\quad \quad\quad\quad \quad\quad\quad \quad (\text{for }3C\sqrt{\epsilon}\leq \epsilon_0). 
 \end{split} 
\end{align}

By repeating  the estimation
\eqref{eq:commsmall0} with $S\brp$ replaced by 
$\inp{x}^{s}S\inp{x}^{-s}\brp$, $s\in \R\setminus\set{0}$,  we
conclude  the
 estimate \eqref{eq:defOrder} with $t=0$ as wanted. Note that the
 above 
  constant $C$ works in this case also.
\end{proof}

Thanks  to \eqref{eq:MicroB0} and Lemma \ref{lem:compBogB} we are lead to studying the
asymptotics of
\begin{equation*}
  \chi^2_+(2M_{a_0}/\epsilon_0)\chi_+(
B/\epsilon_0)R(\lambda+\i0)\psi.
\end{equation*} In turn, taking narrowly supported
  functions
$f_2\succ f_1$ such that $f_1=1$ in a
small neighbourhood of $\lambda$, it suffices to consider 
\begin{equation*}
  f_2(H)\chi^2_+(2M_{a_0}/\epsilon_0)f_2(H)\chi_+(B/\epsilon_0)R(\lambda+\i0)\psi\text { assuming in addition } \psi= f_1(H)\psi.
\end{equation*} 

 For given $\phi_1,\phi_2\in\vB^*$ we write 
$\phi_1 \simeq \phi_2$,  if $\phi_1 - \phi_2 \in\vB^*_0$. With this
convention we  are lead to rewrite,  using in the second  step  Lemma
\ref{lemma:m1} \ref{item:9b},  in the third   step  conveniently
\eqref{eq:partM},  commutation and \eqref{eq:LAPbnda}, and in the last step  the limiting
absorption principle  and the `$Q$-bounds' \eqref{eq:2boundobtain338} and  \eqref{eq:2boundobtain33}
for $H$ as well as for    $H_a$ 
(see  the comments after the computation)
\begin{align*}
&R(\lambda+\i0)\psi\\
  &\simeq f_2(H)\chi^2_+(2M_{a_0}/\epsilon_0)f_2(H)\chi_+(
B/\epsilon_0)R(\lambda+\i0)\psi\\&=f_2(H)h(M_{a_0})\,\sum_{a\in \vA'}\,M_af_2(H)\chi_+(
B/\epsilon_0)R(\lambda+\i0)\psi\quad   (\text{with } h(t)=\chi^2_+(2t/\epsilon_0)/t)
\\&\simeq\Sigma_{a\in \vA'}\,S_aR(\lambda+\i0)\psi\quad   (\text {with }S_a=  f_2(H_a)h(M_{a_0})M_af_2(H)\chi_+(B/\epsilon_0)f_2(H))
  \\&=\Sigma_{a\in \vA'}\,R_a(\lambda+\i0)\psi_a;\\
&\quad \quad \quad \quad \quad \quad \psi_a=S_a\psi-\i T_aR(\lambda+\i0)\psi, \quad T_a=\i\parb{H_aS_a-S_aH}.
\end{align*}
 The last step is formally correct by a  resolvent
equation, and indeed mathematical correct, but it needs justification eventually to be given in
Appendix \ref{sec:Strong bounds}. In the following we confine
ourselves to give a preliminary account of  the
needed  arguments:

 It is a straightforward application of \eqref{82a0} to show that  $S_a=\vO(\inp{x}^{0})$, 
thus in particular the operator  preserves $ L^2_\infty$. Hence the
first term $ S_a\psi\in L^2_\infty$. The second term of $\psi_a$ is
 more   subtle. We rewrite its contribution  in terms of  any support
  function
$f_3\succ f_2$ as 
\begin{align*}
  -\i R_a(\lambda+\i0)&(\bD_aS_a)R(\lambda+\i0)\psi;\\& \bD_aT=\i\parb{H_af_3(H_a)T-THf_3(H)}, \quad \bD T=\i[Hf_3(H),T].
\end{align*} We compute
\begin{align}\label{eq:expan3}
  \begin{split}
 \bD_a S_a&=T_1+T_2+T_3;\\
&T_1=f_2(H_a)\parb {\bD_a h(M_{a_0})}M_af_2(H)\chi_+(B/\epsilon_0)f_2(H)\\
&T_2=f_2(H_a)h(M_{a_0})\parb {\bD M_a}f_2(H)\chi_+(B/\epsilon_0)f_2(H)\\
&T_3=f_2(H_a)h(M_{a_0})M_af_2(H)\parb {\bD \chi_+(B/\epsilon_0)}f_2(H).   
  \end{split}
\end{align} The contributions from $T_1$ and $T_2$ are treated
similarly (note that thanks to \eqref{eq:partM} we can freely replace  $\bD_a$ by $\bD$). For $T_2$  we expand the commutator into a sum of
terms.  Most
of them may be treated by  
 yet another application of the limiting
 absorption principle \eqref{eq:LAPbnda} and 
\eqref{eq:BB^*a}. In fact proceeding  this way it only remains
 to examine
\begin{equation*}
  -\i4 R_a(\lambda+\i0)
  h(M_{a_0})f_2(H_a)p\cdot\parb{\chi_+(|x|)^2\nabla^2m_a(x)} p f_2(H)\chi_+(B/\epsilon_0)f_2(H)R(\lambda+\i0)\psi,
\end{equation*}  Using 
\eqref{eq:Hes0} (and its  notation) and
\eqref{eq:2boundobtain33}  this  expression  is given  (up  to
 a well-defined term in $\vB^*$ thanks to \eqref{eq:LAPbnda}) by
\begin{equation}\label{eq:repgood1}
  -\i 4R_a(\lambda+\i0)h(M_{a_0})f_2(H_a)\sum_{j\leq J}\, \,G^*_{b_j}\xi_j( \hat x)\chi_+(\abs{x})
  \psi_j;\quad \psi_j\in \vH^{d}.
\end{equation} 
By similar arguments for $H_a$ 
rather than for $H$ it follows  that for  $j=1,\dots, J$
\begin{equation}\label{eq:repgood2}
   R_a(\lambda+\i0)h(M_{a_0})f_2(H_a)G^*_{b_j}\xi_j( \hat x)\chi_+(\abs{x})
   \in \vL(\vH^{d},  \vB^*).
\end{equation} Consequently $-\i
R_a(\lambda+\i0)T_2R(\lambda+\i0)\psi$ is a well-defined element of  $\vB^*$.
  
The argument for $T_1$ is similar using Lemma
\ref{lemma:m1} \ref{item:9b} (and  of course again 
\eqref{82a0}, \eqref{eq:LAPbnda} and 
\eqref{eq:BB^*a}).   The argument for $T_3$ is also similar,
we use  the `$Q$-bound' \eqref{eq:2boundobtain338} rather than the  `$Q$-bound' \eqref{eq:2boundobtain33}.

   In conclusion we have deduced   an exact
    representation    of the form 
    \begin{equation}\label{eq:quotientFor}
      R(\lambda+\i0)\psi=\Sigma_{a\in \vA'}\,R_a(\lambda+\i0)\psi_a
        \text{ in }\vB^*/\vB_0^*,
    \end{equation} however we did not show that $\psi_a\in
    L^2_\infty$. In fact there is no reason to believe that  the
    constructed $\psi_a\in
    \vB$, although our procedure  reveals that   $\psi_a\in
    L^2_{1/2}$, cf. \eqref{eq:repgood1}.

  In Appendix \ref{sec:Strong bounds} we
give a  rigorous derivation and interpretation of \eqref{eq:quotientFor} and prove in
addition the following result, to be used in Section \ref{sec:Squeezing argument}.

\begin{lemma}\label{lemma:reduc} Let $\lambda\notin \vT_\p(H)$ and 
  $\psi\in L^2_\infty$  be given. For the construction
  \eqref{eq:quotientFor} there exists a
  sequence of 
  vectors 
  $\psi_{a,n}\in L^2_\infty$, $a\in \vA'$,  such that
  $\psi_{a,n} \to \psi_a\in  L^2_{1/2}$ (in particular
  $(\psi_{a,n})_n $ is convergent in $\vH$)  and
  \begin{equation}\label{eq:deltaerror0}
    \norm[\big]{R_a(\lambda+\i0)\parb{\psi_a-\psi_{a,n}}}_{\vB^*}\to
    0\text{ for }n\to \infty.
  \end{equation} 
\end{lemma} 

We will  iterate the reduction scheme  of \eqref{eq:quotientFor} and Lemma \ref{lemma:reduc}  
by a proper use of 
the  partition 
\begin{equation}\label{eq:deling}
    \chi^2_+(\abs{x})=\chi^2_{\kappa-}(x)+ \chi^2_{\kappa+}(x);\quad\chi_{\kappa\pm}(x)=\chi_+(\abs{x})\chi_\pm(|x^{a_0}|/\kappa|x|),
  \end{equation} cf. Subsection \ref{subsec:
    Geometric considerations a0 ej}. This  is eventually done in Section
  \ref{sec:Proof of stationary completeness}, while Sections
  \ref{sec:Non-threshold analysis} and \ref{sec:Squeezing argument} are
  devoted to some preliminaries needed for appying \eqref{eq:deling}.

\section{Non-threshold analysis}\label{sec:Non-threshold analysis}

The following bounds may be viewed as a stationary version of the so-called minimal velocity bounds. In fact the latter bounds were used in \cite{Is3,Is5} to show the stationary ones for fastly decaying potentials in the $3$-body problem. The proof given below is completely stationary and requires only a minimum amount of regularity on the potentials (moreover it applies  with  long-range potentials included).
\begin{lemma}\label{lemma:non-thresh-analys}
  Let $ a\in \vA\setminus\{a_{\min},a_{\max}\}$ and   $\lambda\notin \vT_\p(H)$   be
  given. Then there exists 
  $c_1>0$ and for any $E\notin \vT_\p(H^a)$ 
an open  neighbourhood $U$ of $E$ such that with $\kappa= c_1 d^a  $,
where 
\begin{equation*}
   d^a= d(E,H^a)=4\dist\parb{E ,\set{\mu\in\vT (H^a)\mid \mu<E}},
\end{equation*} it follows:
\begin{subequations}
\begin{align}\label{eq:bB2}
   &\text{For all real }f\in C^\infty_\c(U)\mand\psi \in
   \vB, \text{ and with }C_2>0
   \text{ being  independent of } \psi,\nonumber\\
&\quad\quad \norm{Q R_a(\lambda+\i
  0)\psi}_{\vH}\leq C_2\norm{\psi}_{ \vB};\\&\quad\quad\quad\quad\quad Q=Q_{a,\kappa,f}=r^{-1/2}\chi_+(\abs{x})\chi_-(|x^{a}|/2\kappa|x|)f(H^a).\nonumber
\end{align}
Under the conditions of \eqref{eq:bB2} and with $f_1$ being  any narrowly supported
standard support function with 
$f_1=1$ in a
small neighbourhood of $\lambda$, 
\begin{equation}\label{eq:Qanontresh}
  \sup _{\Im z\neq 0}\norm{Q_{a,\kappa,f}{f_1}
    (H_a)R_a(z)}_{\vL(\vB,\vH)}\leq C_1<
 \infty.
\end{equation}  
\end{subequations}
\end{lemma}
\begin{proof} Clearly \eqref{eq:bB2} is a consequence of
  \eqref{eq:Qanontresh} (possibly with $C_2\geq C_1$), cf. Remark \refeq{remark:RQ} \ref{item:1Q}. Fix  any support function $f_1$ as in \eqref{eq:Qanontresh}  and pick  any support functions  $f_3\succ f_2\succ f_1$.    
	
	{\bf I} (a positive commutator  estimate).
  By the Mourre estimate for the operator
  $\sqrt{r^a} B^a\sqrt{r^a}$  (the
one referred to in \eqref{eq:Mourre}--\eqref{eq:mourreFull} now with $r$ and  $H$    replaced by $r^a$ 
 and  $H^a$, respectively)
\begin{equation*}
  1_{U}(H^a)\i\Big [H^a, A^a\Big ]1_{U}(H^a)\geq \tfrac { d^a}2 1_{U}(H^a);\quad A^a= \sqrt{r^a} B^a\sqrt{r^a},\,B^a=2\Re \parb{p^a \cdot \nabla^a r^a}. 
\end{equation*} This is for a small open  neighbourhood $U$ of $E$. We recall from \eqref{eq:Mourre} that 
\begin{equation*}
  B=\i\big [H_a,r\big ]=2\Re \parb{p \cdot \nabla r}.
\end{equation*} The operators $B^a$ and $B$ are  bounded relatively to
$H_a$, as to be used explicitly below. Fix any real $f\in C^\infty_\c(U)$.  Now we introduce for any small $\epsilon>0$ the  'propagation observable'
\begin{equation*}
  \Psi=f_1(H_a)f(H^a)\zeta_\epsilon \parb{r^{-1/2} A^ar^{-1/2}}f(H^a)f_1(H_a),
\end{equation*} where  $\zeta_\epsilon \in
C^\infty(\R)$ is real and increasing   with  $\zeta'_\epsilon=1$  on
$(-\epsilon,\epsilon)$ and $\sqrt{\zeta'_\epsilon}
\in C_\c^\infty((-2\epsilon,2\epsilon))$.
Since the argument is small on the support of the derivative $\zeta'_\epsilon $
(more precisely bounded by $2\epsilon$), 
we obtain the lower bound
\begin{align*}\i[H_a,\Psi]\geq &\tfrac { d^a-C\epsilon} {2r}
  f_1(H_a)f(H^a)\zeta'_\epsilon \parb{r^{-1/2}A^ar^{-1/2}}f(H^a)f_1(H_a)\\&+f_1(H_a)\vO(r^{-1-2\mu})f_1(H_a).
\end {align*} In fact the bound holds
with $C=4\norm{Bf_2(H_a)}$. The error $\vO(r^{-1-2\mu})$ arises from
commutation when  using   the representation formula \eqref{82a0}. 

We denote  $T_\epsilon=\sqrt{\zeta'_\epsilon} \parb{r^{-1/2}A^ar^{-1/2}}$
and  conclude by Lemma \ref{lemma:comSTA} that for  some positive constant $C_1'$ 
\begin{equation}
  \label{eq:bB3}
   \sup _{\Im z\neq 0}\norm{r^{-1/2}T_\epsilon f(H^a)f_1(H_a)R_a(z)}_{\vL(\vB,\vH)} \leq C_1';\quad \epsilon= d^a/2(C+1).
\end{equation} (Note that indeed  $ d^a>C\epsilon$, as needed.)

{\bf II} (an estimate for the classically forbidden region).
 We  recall that the function $r^a$ has a scaling parameter, possibly large, to assure the Mourre estimate used in Step I. This scaling might  be depending on $E$, which could cause a problem recalling that we need to produce a constant $c_1>0$ being independent of $E$. This problem is easily cured by using \eqref{eq:compa00} as follows. First we note that on the support of
 $\chi_+(\abs{x})\chi_-(|x^{a}|/2\kappa|x|)$
 \begin{equation*}
 	r^{-1/2}\sqrt{r^a} \leq 2\sqrt\kappa +C'r^{-1/2},\quad C'=C'(E).
 \end{equation*}  Secondly, another elementary direct computation shows that $\norm{B^af_3(H_a)}\le\breve C$ for a positive  constant $\breve C$ being independent of the considered $E\notin \vT_\p(H^a)$. 

Thanks to the above remarks we can now mimic the proof of Lemma \ref{lem:compBogB} and produce a constant $c>0$ being independent of $E$
 such that  
 \begin{equation}
 	\label{eq:smallnonthres}
 	\chi_+(\abs{x})\chi_-(|x^{a}|/2\kappa|x|)(I-T_\epsilon)f_2(H_a)=\vO(\inp{x}^{-1/2});\quad \kappa={c}\epsilon.
 \end{equation}  Recall for comparison that $\epsilon$ has  an explicit  $E$-dependence, implying the same for $\kappa$.
Clearly \eqref{eq:smallnonthres} yields that for some positive constant $C_1''$
\begin{equation}\label{2ndBound}
	\sup _{\Im z\neq 0}\norm{r^{-1/2}\chi_+(\abs{x})\chi_-(|x^{a}|/2\kappa|x|)(I-T_\epsilon)f(H^a)f_1(H_a)R_a(z)}_{\vL(\vB,\vH)} \leq C_1''. 
\end{equation} 

{\bf III} (combining bounds).
Finally we insert  $I=T_\epsilon+(I-T_\epsilon)$ to the left of the factor
$f(H^a)$ in $Q_{a,\kappa,f}$  and  conclude  from (\ref{eq:bB3})--(\ref{2ndBound}) that
for $\kappa={c}\epsilon$   the  desired estimate     \eqref{eq:Qanontresh} holds  with $c_1= c/2(C+1)$ and $ C_1= C_1'+ C_1''$.
\end{proof}
\begin{corollary}\label{corollary:non-thresh-analysB}
  Let $ a\in \vA\setminus\{a_{\min},a_{\max}\}$ and   $\lambda\notin \vT_\p(H)$   be
  given. Then there exists 
  $c_2>0$ such that for any   compact  subset $K$  of $\R\setminus
  \vT_\p(H^a)$  and with $\kappa= c_2d^a_K  $,
where
\begin{equation*}
  d^a_K=d(K,H^a)=\inf_{E\in K}d(E,H^a)=\inf_{E\in K}4\dist\parb{E ,\set{\mu\in\vT (H^a)\mid \mu<E}},
\end{equation*} it follows:
\begin{subequations}
\begin{align}\label{eq:bB2B}
  &\text{For all real }f\in C^\infty_\c(K)\mand\psi \in
   \vB, \text{ and with }C_2>0
   \text{ being  independent of } \psi,\nonumber\\
&\quad\quad \norm{Q R_a(\lambda+\i
  0)\psi}_{\vH}\leq C_2\norm{\psi}_{ \vB};\\&\quad\quad\quad\quad\quad Q=Q_{a,\kappa,f}=r^{-1/2}\chi_+(\abs{x})\chi_-(|x^{a}|/2\kappa|x|)f(H^a).\nonumber
\end{align}
Under the conditions of \eqref{eq:bB2B} and with $f_1$ being  any narrowly supported
standard support function with 
$f_1=1$ in a
small neighbourhood of $\lambda$, 
\begin{equation}\label{eq:QanontreshB}
  \sup _{\Im z\neq 0}\norm{Q_{a,\kappa,f}{f_1}
    (H_a)R_a(z)}_{\vL(\vB,\vH)}\leq C_1<
 \infty.
\end{equation}  
\end{subequations}
\end{corollary}
\begin{proof} This is immediate from Lemma
  \ref{lemma:non-thresh-analys} (with $c_2=c_1/2$,  for example) and a compactness and partition of
  unity argument, writing $f=\Sigma_1^J f_j$ for  suitable  functions
  $f_j$ being narrowly
  supported near some $E_j\in \supp f$, $j=1,\dots,J$, respectively.
  \end{proof}
\begin{corollary}\label{corollary:non-thresh-analysBb}
  Let $ a\in \vA\setminus\{a_{\min},a_{\max}\}$,  $\lambda\notin
  \vT_\p(H)$, $\epsilon>0$    and $\psi \in L^2_1$ be
  given. Then there exists $\kappa>0$ such that 
  \begin{equation}\label{eq:bB2Bb} \norm{ \chi_+(\abs{x})\chi_-(|x^{a}|/2\kappa|x|)
1_{\R\setminus
  \vT_\p(H^a)}(H^a)R_a(\lambda+\i 0)\psi }_{\vB^*/\vB_0^*}\leq \epsilon.
\end{equation}
\end{corollary}
\begin{proof} We can assume that $\lambda\ge \inf \sigma(H^a)$, and we can replace $\R\setminus\vT_\p(H^a)$  by a bounded open subset $U\subseteq
  \R\setminus\vT_\p(H^a)$, say explicitly   by  $U=
  \parb{\R\setminus\vT_\p(H^a)}\cap (\inf \sigma(H^a)-1,\lambda +1)$.

 Let $K_m\subseteq U$, $m\in \N$,  be the set  of  $E\in  U$ for which 
$\dist\parb{E, U^\c}\geq 1/m$. Then the sets $K_m$ are compact and constitute an 
   increasing  sequence. Take any   sequence of standard support functions $f_m\in C^\infty_\c(U)$  obeying 
$f_m=1$ in $K_m$.
 Then 
  \begin{align}\label{eq:2pp7} 
    \begin{split}
   \norm {(1_{U}-f_m)&(H^a)R_a(\lambda+\i 0)\psi }_{\vB^*}\leq C_m\norm
   {(1_{U}-f_m)(H^a)\inp{x_a}\psi}_{\vH};\\
& C_m=\sup_{\eta\in U\setminus K_m}\, \norm{(p^2_a
  +\eta-\lambda -\i
 0)^{-1} \inp{x_a}^{-1}}_{\vL(L^2(\bX_a),\,\vB(\bX_a)^*)}.   
    \end{split}
  \end{align} Note that  (\ref{eq:2pp7}) follows  by  using  the
  trivial inclusion 
  $\set{\abs{x}\leq \rho}\subset \set{\abs{x_a}\leq \rho}$ and the
  spectral theorem on multiplication operator form for $H^a$
  \cite[Theorem VIII.4]{RS} (amounting to a partial
  diagonalization).  The
  sequence $(C_m)_1^\infty$ is  bounded (since it is decreasing). Rewriting $\vH=L^2\parb{\bX_a,L^2(\bX^a)}$
  and invoking the Lebesgue dominated convergence theorem and the Borel
  calculus for $H^a$, it follows  that
\begin{equation*}
   \norm{(1_{U}-f_m)(H^a)R_a(\lambda+\i 0)\psi }_{\vB^*}\to 0\text{ for }m\to \infty.
\end{equation*} 

 In particular we can fix $m$ such that 
\begin{equation*}
   \norm{(1_{U}-f_m)(H^a)R_a(\lambda+\i 0)\psi }_{\vB^*}\leq \epsilon.
 \end{equation*} 
Applying  now  Corollary \ref{corollary:non-thresh-analysB}  with $\kappa=c_2d^a_{\supp f_m} $ and  $f=f_m$ we conclude that 
 \begin{align*}
   \chi_+(\abs{x})\chi_-(|x^{a}|/2\kappa|x|)f_m(H^a)R_a(\lambda+\i
   0)\psi \in\vB_0^*
 \end{align*}  as well as
\begin{align*}
	\norm{\chi_+(\abs{x})\chi_-(|x^{a}|/2\kappa|x|)(1_{U}-f_m)(H^a)R_a(\lambda+\i 0)\psi }_{\vB^*}\leq \epsilon,
\end{align*}
yielding in combination \eqref{eq:bB2Bb} for this $\kappa$.
\end{proof}

\section{A squeezing argument}\label{sec:Squeezing argument}
In this section our outset is   Lemma \ref{lemma:reduc}. Hence we fix
$a\in \vA'$, $\lambda\notin \vT_\p(H)$  and $\psi\in
L^2_\infty$. We will initiate a reduction procedure for deriving  the asymptotics of the constructed  vectors
$R_a(\lambda+\i0)\psi_a \in\vB^*$ from \eqref{eq:quotientFor}. Along
with a modified  version of the  reduction procedure
of Section \ref{sec::A partition of unity} involving  Corollary \ref{corollary:non-thresh-analysB}, this  procedure  can 
be iterated as to be elaborated  on  in Section \ref{sec:Proof of stationary completeness}. 
\begin{lemma}\label{lemma:squeezingLemmaS} Under the above conditions
  $a\in \vA'$ and  $\lambda\notin \vT_\p(H)$:
  \begin{enumerate} [1)]
  \item \label{item:sqe1}   For any $\breve \psi_a \in L^2_1$ there
    exists  a sequence of  standard support functions 
$f^a_m$ all being supported in $\R\setminus \vT_\p(H^a)$, and  there
    exists  $\breve\phi_a\in\vB^*$ whose asymptotics as an element of
    $\vB^*/\vB_0^*$ agrees with the right-hand side of
    \eqref{eq:asres29},   such that 
  \begin{equation}\label{eq:squeFun} 
     \parb{I-f^a_m(H^a)} R_a(\lambda+\i0)\breve \psi_a = \breve\phi_a +
     o(m^0)\text{ in }\vB^*/\vB_0^*.
    \end{equation}

\item \label{item:sqe2} For the  vector
  $R_a(\lambda+\i0)\psi_a \in\vB^*$ from \eqref{eq:quotientFor} (given
  in terms of  the   fixed $\psi\in
L^2_\infty$) there
    exist   sequences $(f^a_j)_j$ of  standard support functions 
being supported in $\R\setminus \vT_\p(H^a)$ and 
$(\breve \psi_{a,j})_j\subseteq  L^2_\infty$, and  there
    exists  $\phi_a\in\vB^*$ whose asymptotics as an element of
    $\vB^*/\vB_0^*$ agrees with the right-hand side of
    \eqref{eq:asres29},  such that 
\begin{equation}\label{eq:squeFun9} 
  R_a(\lambda+\i0)\psi_a = f^a_j(H^a) R_a(\lambda+\i0)\breve\psi_{a,j}  +\phi_{a}+ o(j^0)\text{ in }\vB^*/\vB_0^*.
    \end{equation}  
  \end{enumerate}
\end{lemma}
\begin{remark*} For our applications it is not important that the
  vector $\phi_{a}$ in \ref{item:sqe2} can be chosen independent of
  $j$. However for the sake of completeness of presentation we find it
  better to prove the sharper version, which amounts to taking 
\begin{equation}\label{eq:exactF}
   \phi_{a}=\sum_{\alpha=(a,\lambda^\alpha,
     u^\alpha),\,\lambda^\alpha<\lambda}
   2\pi\i  \,J_\alpha
     v^{+}_{\alpha,\lambda} [g_\alpha],\text{ where }
    g_\alpha=\Gamma_\alpha^{a+}(\lambda)\psi_a, \, (g_\alpha)_\alpha\in \vG.
 \end{equation}  We consider henceforth only channels $\alpha$ specified
 as in this  summation (in particular with the
 first 
 component fixed as $a$). The operator 
 $\Gamma_\alpha^{a+}(\lambda)$ is (here and below) the
 outgoing 
 restricted channel wave operator for $H_a$ at the energy $\lambda$, possibly
 defined as in Proposition \ref{prop:radi-limits-chann22} (with $H$
replaced by $H_a$). 
  \end{remark*}

\begin{proof}[Proof  of Lemma \ref{lemma:squeezingLemmaS}]  For
  \ref{item:sqe1} we will use the proof of  Corollary
  \ref{corollary:non-thresh-analysBb} and for \ref{item:sqe2} the approximations 
$R_a(\lambda+\i0)\psi_{a,n}$ from  \eqref{eq:deltaerror0}  ($\psi_{a,n} \in
L^2_\infty$ is constructed explicitly in \eqref{eq:n2}). We can assume that
  $\lambda\ge \inf \sigma(H^a)$.  Let $f_1$ be a narrowly supported standard support function with 
$f_1=1$ in a
small neighbourhood of $\lambda$
 and let $f_2\succ f_1$, also given narrowly
supported at $\lambda$ (these functions appear explicitly  in Section
\ref{sec::A partition of unity}, Appendix \ref{sec:Strong bounds} and
in \eqref{eq:QanontreshB}).

{\bf \ref{item:sqe1}} We introduce  the quantities
\begin{equation*}
  \breve\phi_{a}=\sum_{\alpha=(a,\lambda^\alpha,
     u^\alpha),\,\lambda^\alpha<\lambda}
   2\pi\i  \,J_\alpha
     v^{+}_{\alpha,\lambda} [\breve g_\alpha],\quad
    \breve g_\alpha=\Gamma_\alpha^{a+}(\lambda)\breve \psi_a\in \vG_a.
\end{equation*} Recalling   the Bessel inequality
 \eqref{eq:Besn}  (with $H$
replaced by $H_a$)  we can record  that
 \begin{subequations}
\begin{equation}\label{eq:ScatEnergy22330}
  \sum_{\lambda^\alpha<
  \lambda}\,\norm{\breve g_\alpha}^2 \leq 
  \inp{\delta(H_a-\lambda)}_{\breve \psi_a}<\infty,
\end{equation} and thanks to \eqref{eq:pars} therefore in turn,
\begin{equation}\label{eq:eqBnd0}
  C_1\sum_{\lambda^\alpha<
    \lambda}\,\norm{\breve g_\alpha}^2 \leq\norm[\big] {\breve\phi_{a}}^2_{\vB^*/\vB_0^*}\leq C_2\sum_{\lambda^\alpha<
    \lambda}\,\norm{\breve g_\alpha}^2<\infty.
\end{equation} In particular   $\breve\phi_{a}$ is a   well-defined 
element of $\vB^*/\vB_0^*$.
 \end{subequations}

Next we consider the increasing sequence of compact sets $(K_m)_m$ 
\begin{equation*}
  K_m\subseteq U:=\parb{\R\setminus\vT_\p(H^a)}\cap (\inf \sigma(H^a)-1,\lambda +1)
\end{equation*}
 from the
proof of Corollary \ref{corollary:non-thresh-analysBb}. Let 
$(f^a_m)_m$  be any   sequence of standard support functions $f^a_m\in C^\infty_\c(U)$  obeying 
$f^a_m=1$ in $K_m$.

  To show   \ref{item:sqe1} it suffices  to prove  that  the element of $\vB^*/\vB_0^*$
 \begin{equation*}
   \breve\phi_{a,m}:=\breve\phi_{a} -\parb{1_{\vT_\p(H^a)}+1_U-f^a_m}(H^a)R_a(\lambda+\i
  0){f_2} (H_a)\breve\psi_a =o(m^0)\text{ in }\vB^*/\vB_0^*.
 \end{equation*} 
 By  the proof of  Corollary
  \ref{corollary:non-thresh-analysBb}
  \begin{subequations}
  \begin{align}\label{eq:2pp} 
    \begin{split}
   \norm {(1_U-f^a_m)&(H^a)R_a(\lambda+\i
     0) f_2(H_a)\breve\psi_{a}}_{\vB^*}\\ &\leq C_m\norm
                         {(1_U-f^a_m)(H^a)\inp{x_a}f_2(H_a)\breve\psi_a }_{\vH};\\
& \quad \quad C_m=\sup_{\eta\in U\setminus K_m}\, \norm{(p_a^2-\lambda+\eta-\i
 0)^{-1} \inp{x_a}^{-1}}_{\vL(L^2(\bX_a),\,\vB(\bX_a)^*)}.   
    \end{split}
  \end{align} As recorded there the
  sequence $(C_m)_1^\infty$ is  bounded, and  by 
   the dominated convergence theorem and the Borel
  calculus for $H^a$  it follows  that
\begin{equation}\label{eq:firDec}
   \norm{(1_U-f^a_m)(H^a)R_a(\lambda+\i
     0) f_2(H_a)\breve\psi_a }_{\vB^*}\to 0\text{ for }m\to \infty.
\end{equation}

  Next,  we  show that 
\begin{equation}\label{eq:1pp}
    \breve\phi_{a}= 1_{\vT_\p(H^a)}(H^a)R_a(\lambda+\i
 0) f_2(H_a)\breve\psi_a \text{ in }\vB^*/\vB_0^*,
  \end{equation} 
\end{subequations}  which 
clearly in the combination 
  of \eqref{eq:firDec}  yields the wanted property 
  $\breve\phi_{a,m}  = o(m^0)$. 

First consider the simplest case where  the summation defining $\breve\phi_{a}$ only involves finitely many channels
  $\alpha$.  We  use that
  $\Gamma_\alpha^{a+}(\lambda)=
  \gamma_a^+(\lambda_\alpha)J^*_\alpha$, as expressed by the  
  trace  operator $ \gamma_a^+(\cdot )=\gamma_a^-(\cdot
  )=\gamma_{a,0}(\cdot )$  for the free Hamiltonian $p_a^2$  at the   positive energy 
  $\lambda_\alpha=\lambda-\lambda^\alpha$ (cf.  Proposition
  \ref{prop:radi-limits-chann22} and Theorem \ref{thm:wave_matrices}),  and  stationary completeness
   of positive energies 
  for   $p^2_a$ (cf.  \cite {Sk3}). The latter assertion implies that
  \begin{equation*}
    (p^2_a-\lambda_\alpha -\i 0)^{-1} \breve\psi_\alpha -
    2\pi\i \, v^{+}_{\alpha,\lambda}
    [\gamma_a^+(\lambda_\alpha)\breve\psi_\alpha ]\in \vB^*_0(\bX_a);\quad
    \breve \psi_{\alpha}=J^*_\alpha\breve\psi_a \in \vB(\bX_a).
  \end{equation*} In combination with  the identity $\Gamma_\alpha^{a+}(\lambda)=
  \gamma_a^+(\lambda_\alpha)J^*_\alpha$ we then conclude \eqref{eq:1pp} by expanding 
  both  sides  into  finite  sums of identifiable terms.

In the remaining  case where there are  infinitely many channels
  $\alpha$ in the summation  we pick,  given in terms by an arbitrary numbering of the  eigenstates of $H^a$ (with  eigenvalues $\lambda^\alpha<\lambda$),  an increasing
 sequence  of finite-rank projections $P^a_j\to 1_{\vT_\p(H^a)}(H^a)$ (strongly). By using a
 modification of  \eqref{eq:2pp} (and arguing similarly) we then
 obtain as in \eqref{eq:firDec}  that
 \begin{equation*}
   \norm{\parb{1_{\vT_\p(H^a)}(H^a)-P^a_j}R_a(\lambda+\i
 0) f_2(H_a)\breve\psi_a }_{\vB^*}\to 0\text{ for }j\to \infty.
\end{equation*} Thanks  to 
this property, a trivial modification  of (\ref{eq:eqBnd0})  and   the
arguments for  the finite summation  case also  the  infinite
summation  formula  follows.
We have proved \eqref{eq:1pp} and therefore  \ref{item:sqe1}.

{\bf \ref{item:sqe2}}  Let us first establish the weaker version of   
\eqref{eq:squeFun9} given by allowing   the
vector  $\phi_a$  to have a dependence on  $j$, say \eqref{eq:squeFun9} with
$\phi_a$ replaced by $\breve\phi_{a,j}$. Eventually
we  show that  this term  may be taken independently of  $j$ in agreement
with \eqref{eq:exactF}.

Consider the sequence of 
  vectors 
  $\psi_{a,n}\in L^2_\infty$ appearing  in Lemma
  \ref{lemma:reduc}. (It is constructed explicitly in  \eqref{eq:n2}.)
  We take  a  subseqence $(\psi_{a,n_j})_j$  such that
  \begin{subequations}
  \begin{equation}\label{eq:deltaerror01}
    \norm[\big]{R_a(\lambda+\i0)\parb{\psi_a-\psi_{a,n_j}}}_{\vB^*}\leq
    1/j;\quad j\in \N.
  \end{equation} For each $j$ the vector $\psi_{a,n_j}\in
  L^2_1$. Thanks to  \ref{item:sqe1}
 we can consequently pick $f^a_j $ supported in $\R\setminus
  \vT_\p(H^a)$ and $\breve\phi_{a,j}\in\vB^*$  whose asymptotics as an element of
    $\vB^*/\vB_0^*$ agrees with the right-hand side of
    \eqref{eq:asres29}, such that 
  \begin{equation}\label{eq:squeFun1} 
     \norm[\big] {\parb{I-f^a_j(H^a)} R_a(\lambda+\i0) 
       \psi_{a,n_j} - \breve\phi_{a,j}}_{\vB^*/\vB_0^*}\leq  1/j;\quad j\in \N.
    \end{equation}  
  \end{subequations}

The combination of \eqref{eq:deltaerror01} and \eqref{eq:squeFun1}
yields that 
\begin{equation}\label{eq:squeFun12} 
     \norm[\big]{R_a(\lambda+\i0) 
       \psi_{a} -f^a_j(H^a)R_a(\lambda+\i0) 
      \psi_{a,n_j} - \breve\phi_{a,j}}_{\vB^*/\vB_0^*}\leq  2/j;\quad j\in \N.
    \end{equation}
 This is \eqref{eq:squeFun9}  with $\breve\psi_{a,j}=\psi_{a,n_j} $,
 except however that (most likely) $\breve\phi_{a,j}$ depends of $j$.

 Hence it remains to  show that indeed we may replace this  vector  $\breve\phi_{a,j}$ 
 by the vector $\phi_a$  from \eqref{eq:exactF}. Introducing
\begin{equation*}
    \breve g_{\alpha,j}=\Gamma_\alpha^{a+}(\lambda)\parb{\breve\psi_{a,j}
      -\psi_a}, 
 \end{equation*} we need to estimate
 \begin{equation}\label{eq:sumError}
   \breve\phi_{a,j}-\phi_{a}=\sum_{\alpha=(a,\lambda^\alpha,
     u^\alpha),\,\lambda^\alpha<\lambda}
   2\pi\i  \,J_\alpha
     v^{+}_{\alpha,\lambda} [\breve g_{\alpha,j}].
 \end{equation} At this point we record that indeed $\breve  g_{\alpha,j}$ is a well-defined element of $\vG_a$, although this is
 not seen directly from  Proposition
  \ref{prop:radi-limits-chann22} (since this result 
 requires $\psi_a\in \vB$, and we only know $\psi_a\in L^2_{1/2}$). 
However thanks to \cite[Remark 9.15 1)]{Sk3} and the  construction of
$\psi_a$ in Appendix \ref{sec:Strong bounds},  the quantity $\breve g_{\alpha,j}$
is still a well-defined element of $\vG_a$ as given by the recipe 
\eqref{eq:restrH} of the proposition.

More generally the  sums \eqref{eq:exactF} and \label{sec:squeezing-argument}
 \eqref{eq:sumError}  are  well-defined elements of $\vB^*/\vB_0^*$ due  to the following  
  version
  of \label{sec:squeezing-argument-1} \label{sec:squeezing-argument-3}
  \eqref{eq:ScatEnergy22330} and \eqref{eq:eqBnd0} (see \cite
  [Subsection 9.2]{Sk3} for details)
  \begin{subequations}
  \begin{equation}\label{eq:ScatEnergy22336}
   C^{-1}\norm[\Big] {\sum_{\lambda^\alpha<\lambda} J_\alpha
     v^{+}_{\alpha,\lambda} [\breve  g_{\alpha,j}]}^2_{\vB^*/\vB_0^*}\leq \sum_{\lambda^\alpha<
  \lambda}\,\norm{\breve  g_{\alpha,j}}^2 \leq 
  \inp{\delta(H_a-\lambda)}_{\breve\psi_{a,j}-\psi_a}<\infty.
\end{equation} In fact  by the construction of
$\breve\psi_{a,j}=\psi_{a,n_j} $ given in  \eqref{eq:n2} it follows that 
\begin{equation}\label{eq:brevNapprox}
   \inp{\delta(H_a-\lambda)}_{\breve\psi_{a,j}-\psi_a}\to
   0\text{ for }j\to \infty.
 \end{equation}   
  \end{subequations}
  The latter  assertion follows from   combining  
 \eqref{eq:BB^*a} with a variation of the bounds 
\cite [(7.3a) and (7.3b)]{Sk3}. Hence  using the same  `propagation
observables' as for \eqref{eq:2boundobtain3377} (used also in
Section \ref{sec::A partition of unity} of course)  we obtain by mimicking the proof
of  \cite [Lemma 7.1]{Sk3} that 
\begin{equation*}\label{eq:kato2}
   Q_l{f_2} (H_a)\delta(H_a-\lambda){f_2} (H_a){ Q_k}^* \in
   \vL(\vH),
\end{equation*} from which we  readily deduce  \eqref{eq:brevNapprox}
using 
\eqref{eq:BB^*a} and \eqref{eq:2boundobtain3377}.

By combining  \eqref{eq:ScatEnergy22336} and \eqref{eq:brevNapprox} we conclude
that $\norm[\big] {\breve\phi_{a,j}-\phi_{a}}_{\vB^*/\vB_0^*}=o(j^0)$,
and therefore  with \eqref{eq:squeFun12} in turn  the desired bound \eqref{eq:squeFun9}.
\end{proof}

\section{Proof of stationary completeness}\label{sec:Proof of stationary completeness}

 The main result of the paper reads. 
\begin{thm}\label{them:stat-compl-enerMain} Suppose  Condition \ref{cond:smooth2wea3n12}.
  Then all $\lambda
  \in \vE=(\min \vT(H),\infty)\setminus {\vT_{\p}(H)}$ are stationary
  complete for $H$.
\end{thm}
\begin{proof} Let $\lambda\notin \vT_\p(H)$ and 
  $\psi\in L^2_\infty$  be given. We need to show that 
     $R(\lambda+\i0)\psi\in\vB^*/\vB_0^*$
     has  asymptotics   agreeing  with the right-hand side of
     \eqref{eq:asres29} in
 Section \ref{sec:-body effective potential and a $1$-body radial
   limit}, for which purpose  the  paragraph  right after \eqref{eq:pars} comes in handy.
 By combining \eqref{eq:quotientFor} and Lemma
   \ref{lemma:squeezingLemmaS} \ref{item:sqe2}  we are  left with for each
   $a\in \vA'$ to 
  analyze the asymptotics of the vector
  $f(H^a)R_a(\lambda+\i0)\breve \psi_a $  for a fixed 
  standard support function 
$f$  supported in $\R\setminus \vT_\p(H^a)$ and for a fixed  $\breve\psi_{a}\in L^2_\infty$. We can assume 
that  $a\neq {a_{\min}}$ and  that $\breve\psi_a
=f_1(H_a)\breve\psi_a\in L^2_\infty$. Here  $f_1$ is  a support function given as in
\eqref{eq:QanontreshB}, and we let $f_2\succ f_1$, also given narrowly
supported. With a slight abuse of the  notation of
\eqref{eq:quotientFor}  we abbrevate  $\breve\psi_a=\psi_a$.

With reference to Subsection \ref{subsec: Geometric considerations a0
  ej} we denote from now on $a=a_{0}$ (reserving the symbol `$a$' for
a different  purpose) and repeat the partition of unity procedure of Section \ref{sec::A partition of unity} with $H$ replaced by $H_{a_0}$.

We write with
$u_{a_0}:=R_{a_0}(\lambda+\i0)f_1(H_{a_0})\psi_{a_0}$ (recall that now $\psi_{a_0}
=f_1(H_{a_0})\psi_{a_0}\in L^2_\infty$) 
and for a small $\kappa >0$
\begin{equation}\label{eq:dec3}
  f(H^{a_0})u_{a_0}-\Sigma_{a\in
    \vA_{a_0}'}\,S^ {a_0}_{a,\kappa} \,f(H^{a_0})u_{a_0} +Tf(H^{a_0})u_{a_0}\in \vB_0^*,
\end{equation} where $T=\vO^ {a_0}_\kappa(\inp{x}^0)$ (using the
convention of \eqref{eq:1712022kap}) and 
\begin{equation*}
  S^{a_0}_{a,\kappa}=
  f_2(H_{a})h(M^{a_0}_{a_0,\kappa})M^{a_0}_{a,\kappa}f_2(H_{a_0})\chi_+(B/\epsilon_0)f_2(H_{a_0}).
\end{equation*}  Here $h(t)=\chi^2_+(2t/\epsilon_0)/t$ as in Section
\ref{sec::A partition of unity}  and indeed \eqref{eq:dec3} may be
obtained by mimicking this section  using in addition the partition 
\eqref{eq:deling} and various commutation to capture   `$\kappa$-localization' (the latter  is absent in  Section
\ref{sec::A partition of unity}). We skip the straightforward details
of proof.

 Thanks to  \eqref{eq:LAPbnda}, \eqref{eq:bB2B} and the assumed smallness of  $\kappa >0$  the third term of \eqref{eq:dec3} is in
 $L^2_{-1/2}$, in particular in  $\vB_0^*$ as the whole
 sum. Consequently it suffices to derive  the desired good asymptotics
of  the terms in the middle summation. 

So let us for fixed $a\in  \vA_{a_0}'$ 
look at $S^ {a_0}_{a,\kappa} \,f(H^{a_0})u_{a_0}$. Fixing  any support
function $\tilde{f}\succ f$, also supported in $\R\setminus
\vT_\p(H^{a_0})$, it suffices to derive asymptotics of
$\tilde{f}(H^{a_0}_a)S^ {a_0}_{a,\kappa} \,f(H^{a_0})u_{a_0}$ (here $H^{a_0}_a=-
\Delta_{x^{a_0}} +V^a$).
Recall that we consider only sufficiently small $\kappa >0$ (the smallness may depend on the choice of $\tilde{f}$).

In comparison with Section \ref{sec::A partition of unity}  and Appendix \ref{sec:Strong bounds}  we   need to compute
the operator 
\begin{equation*}
  T_{a,\kappa}:=\i\tilde{f}(H^{a_0}_a)\parb{H_aS^ {a_0}_{a,\kappa} -S^ {a_0}_{a,\kappa} H_{a_0}}f(H^{a_0}).
\end{equation*}
With   a proper substitution of \eqref{eq:deling}  we can basically  proceed as
before and  expand 
 $T_{a,\kappa}$ into  a finite sum of terms    on
 the form  
 \begin{equation*}
 	f_2 (H_a){ Q_k}^* B_k
 	Q_k f_2(H_{a_0}),
 \end{equation*}  where (possibly replaced by  `components')
 \begin{equation*}
  Q_k f_2 (H_{a_0})=
  \vO(\inp{x}^{-1/2}),\,\,Q_k f_2 (H_{a})=
 \vO(\inp{x}^{-1/2}) \mand  B_k\text{ is bounded}.
 \end{equation*} In all cases the appropriate versions of
 \eqref{eq:2boundobtain3377} (with $H$ 
 replaced by $H_{a_0}$)  are at disposal.  Note that we use   Remarks
 \ref{remark:RQ} and  in particular that
 the scheme of Lemma
 \ref{lemma:comSTA} applies and combines well with \eqref{eq:BB^*a},  
 \eqref{eq:Hes02}--\eqref{eq:2boundobtain33NOT} and \eqref{eq:QanontreshB}, as explained in
 Subsection \ref{subsec: Geometric considerations a0 ej}. (Strictly
 speaking here 
 for example \eqref{eq:2boundobtain33NOT} is used with $f_1$ replaced
 by  $f_2$ and not only for  the pair  $(H_{a_0}, f(H^{a_0}))$ but
 also for
 the pair  
$( H_{a}, \tilde{f}(H^{a_0}_a))$. The latter  is legitimate since $\vT_\p(H^{a_0}_a)\subseteq\vT_\p(H^{a_0})$.)  Hence for a suitable  vector
 $\psi^{a_0}_a\in L^2_{1/2}$
 \begin{equation*}
   \tilde{f}(H^{a_0}_a)S^
 {a_0}_{a,\kappa} \,f(H^{a_0})u_{a_0}-R_a(\lambda+\i0)\psi^{a_0}_a \in \vB^*_0,
 \end{equation*} reducing our analysis to the asymptotics of  any such
 term
 $R_a(\lambda+\i0)\psi^{a_0}_a$. 

 Finally we    repeat the
  above procedure now for $R_a(\lambda+\i0)\psi^{a_0}_a$, $a\in  \vA_{a_0}'$,
  rather than for the terms of \eqref{eq:quotientFor} treated above
  (note the appearing   symmetry). 
When  iterated clearly  the  procedure terminates  and eventually we   thus
conclude asymptotics   agreeing  with the right-hand side of
\eqref{eq:asres29}, first up to an arbitrary small correction  and then
(thanks to the
  discussion after \eqref{eq:pars}) without any correction term.
\end{proof}
\begin{remark*} 
The notion of a stationary complete energy $\lambda$   requires
$\lambda\notin \vT_\p(H)$. In particular if  $\lambda\notin \vT(H)$  is an eigenvalue this
$R(\lambda+\i 0)\psi$  ill-defined. However in this case the corresponding
eigenprojection, say denoted $P_\lambda$, 
maps to $L^2_\infty$ and the expression $(H-P_\lambda -\lambda -\i
0)^{-1}$ (along with its imaginary part) does have an interpretation,
cf. \cite {AHS}, and the scattering theories  for $H$ and
$H-P_\lambda$ coincide. Hence we could deal with $\lambda\in
\sigma_{\pp}(H)\setminus \vT(H)$ upon modifying the Parseval formula
  \eqref{eq:ScatEnergy222331} by using $\delta(H-P_\lambda -\lambda)$
  rather than $\delta(H -\lambda)$ in the formula, cf.  \cite[Remark
  9.3]{Sk3}. With this modified definition of stationary completeness
  at such $\lambda $ our procedure of proof applies, yielding similar
  results (we leave out the
  details). On the other hand $\lambda\notin
\vT(H)$ is essential.

\end{remark*}

\appendix
\section{Strong resolvent bounds,  proof of Lemma \ref{lemma:reduc}}\label{sec:Strong bounds}

We shall in this appendix justify our somewhat formal derivation of
\eqref{eq:quotientFor} and prove Lemma \ref{lemma:reduc}. For that purpose we need certain strong
resolvent bounds.

The Mourre estimate \eqref{eq:mourreFull} implies 
\eqref{eq:LAPbnda}, \eqref{eq:BB^*a} as well as \eqref{eq:MicroB0}. Let us here recall
partially stronger versions of these assertions from \cite{AIIS} (used also in
\cite{Sk3}).  

In agreement with the stated Mourre estimate we first
fix  $\lambda\not\in \vT_{\p}(H)$ and any positive $c< d(\lambda,H)$ (the latter given by
\eqref{eq:optimal}).   Then   for a suitable rescaled version of
$r(x)$ (also denoted by $r$) and a small  
open neighbourhood $U$ of $\lambda$, it follows  that
\begin{equation}\label{eq:mourreFull9}
  \forall \text{ real } f\in C_\c^\infty(U):\quad{f}(H)
  \i[H,r^{1/2}Br^{1/2}]{f}(H)\geq
 c\,{f}(H)^2.
\end{equation} 

We let  then $\vN$
denote any compact   neighbourhood of $\lambda$ with 
$\vN\subseteq U$ and introduce
$\vN_\pm=\set{z\in \C\mid \Re z\in \vN\mand 0<\pm \Im z\leq 1}$.

\begin{thm} [\cite{AIIS}]\label{microLoc}
  \begin{enumerate}[1)]
  \item \label{item:1bbbA} There exists $C>0$ such that for all $z\in \vN_\pm$
\begin{equation}\label{eq:BBstjerne}
\|R(z)\|_{\vL(\vB,\vB^*)}
\le C.
 \end{equation} 

\item \label{item:1bbbA2} For any  $s>1/2$ and $\alpha\in (0, \min\{\mu, s-1/2\})$ there 
exists $C>0$ such that 
for  $k\in\{0,1\}$, $z\in \vN_\pm$ and $z'\in \vN_\pm$, respectively,
\begin{align}
\|p^kR(z)-p^kR(z')\|_{\mathcal L(L^2_s,L^2_{-s})}
\le C|z-z'|^\alpha.
\label{eq:171125}
\end{align}
In particular, for any $\lambda\in \vN$ and $s>1/2$ the following boundary values exist:
\begin{align*}
p^kR(\lambda\pm \i 0):=\lim_{\epsilon\to 0_+}p^kR(\lambda\pm\i \epsilon)
\text{\ \ in }\mathcal L(L^2_s,L^2_{-s}),
\end{align*} 
respectively {{(here the limits are taken in the operator-norm topology)}}. The same boundary values are realized (in an extended
form) as 
\begin{align*}
p^kR(\lambda\pm \i 0)= \swslim_{\epsilon\to 0_+} p^kR(\lambda\pm\i \epsilon)
\text{\ \ in }\vL(\mathcal B,\mathcal B^*),
\end{align*}
respectively  {{(here the right-hand side operators act on any $\psi\in \mathcal B$ as
    the $\epsilon$-limits of $p^kR(\lambda\pm\i \epsilon)\psi$ in the weak-star
    topology of $\mathcal B^*$)}}.

\item \label{item:1bbbB} For any $\beta\in (0,\mu)$ and 
$F\in C^\infty(\mathbb R)$ with 
$$\mathop{\mathrm{supp}}F\subset \parbb{-\infty,\sqrt{c}}\mand 
F'\in C^\infty_{\mathrm c}(\mathbb R),$$
there exists $C>0$, such that 
for all $z\in \vN_\pm$ and $\psi\in L^2_{1/2+\beta}$ 
 \begin{equation}\label{eq:MIRO}
\|F(\pm B)R(z)\psi\|_{L^2_{-1/2+\beta}}
\le C\|\psi\|_{L^2_{1/2+\beta}},
 \end{equation} 
respectively.
  \end{enumerate}

\end{thm}

Recall from Section \ref{sec::A partition of unity}  that 
\begin{equation*}
  S_a=  f_2(H_a)h(M_{a_0})M_af_2(H)\chi_+(B/\epsilon_0)f_2(H).
\end{equation*}
We need to justify the  presumably convincing, although formal, computation  
\begin{align}\label{eq:res1}
  \begin{split}
  &S_aR(\lambda+\i0)\psi=\phi_a:=R_a(\lambda+\i0)\psi_a;\\
&\quad \quad \quad \quad \quad \quad \psi_a=S_a\psi-\i T_aR(\lambda+\i0)\psi, \quad T_a=\i\parb{H_aS_a-S_aH}.  
  \end{split}
\end{align}
 
We will prove 
that the function $ \phi_a$ in \eqref{eq:res1} has a
well-defined meaning as an element in $\vB^*$ using an   a priori interpretation 
different from \eqref{eq:BB^*a}. From the outset $ \phi_a$ is taken as
the  weak
limit,   say in
 $L^2_{-1}$,
\begin{align*}
   \phi_a= \lim_{\epsilon\to 0_+}  R_a(\lambda+\i
 \epsilon )  &S_a\psi   -\i \lim_{\epsilon\to 0_+}R_a(\lambda+\i
 \epsilon )T_aR(\lambda+\i \epsilon)\psi \\=R_a(\lambda+\i
 0 )  S_a\psi    -\i \lim_{\epsilon\to 0_+}&R_a(\lambda+\i
 \epsilon )\chi^2_-(2B/\epsilon_0)T_aR(\lambda+\i \epsilon)\psi\\&-\i \lim_{\epsilon\to 0_+}R_a(\lambda+\i
 \epsilon )\chi^2_+(2B/\epsilon_0)T_aR(\lambda+\i \epsilon)\psi. 
\end{align*}  
 It follows from \eqref{eq:BB^*a} (or from Theorem \ref{microLoc}) that the first term to the right is
 an element in $\vB^*$ (since $S_a\psi\in L^2_\infty$). By
 commutation we can write (for the second  term)
 \begin{equation*}
   \chi^2_-(2B/\epsilon_0)T_a=\inp{x}^{-1/2-\mu} B\inp{x}^{-1/2-\mu}\text{ with
   } B \text{ bounded}.
 \end{equation*} This allows us to compute
 \begin{align*}
   \lim_{\epsilon\to 0_+}&R_a(\lambda+\i
 \epsilon )\chi^2_-(2B/\epsilon_0)T_aR(\lambda+\i \epsilon)\psi\\&=\lim_{\epsilon\to 0_+}R_a(\lambda+\i
 \epsilon)\chi^2_-(2B/\epsilon_0)T_aR(\lambda+\i
                                                                      0)\psi\quad(\text{by
                                                                      }\eqref{eq:LAPbnda}\\
&=R_a(\lambda+\i 0)\chi^2_-(2B/\epsilon_0)T_aR(\lambda+\i
                                                                      0)\psi\in
  \vB^*\quad(\text{by }\eqref{eq:BB^*a}.
 \end{align*} We conclude that the first and second terms
 are   on  a  form consistent with \eqref{eq:BB^*a},
\begin{equation*}
      R_a(\lambda+\i
 0)\brp\in \vB^*
      \text{ for some }\brp\in (L^2_{1/2+\mu}\subseteq)\, \vB.
    \end{equation*}
 
The third term is different. 
The  expression
$T_aR(\lambda+\i 0)\psi$ defines  an element
of 
$L_{s}^2$, $s=1/2$, as to be demonstrated below, however  we do not  prove  better decay.
 Thanks to 
\eqref{eq:LAPbnda} and  Theorem \ref{microLoc} (for $H_a$ rather than
for $H$)  there exists the operator-norm-limit
\begin{equation}\label{eq:strongED_+SDY}
\vLlim_{\epsilon\to 0_+}\,\,
  \inp{x}^{-1}R_a(\lambda+\i
 \epsilon)\chi^2_+(2B/\epsilon_0)\inp{x}^{-s},\quad s=1/2.
\end{equation} 
 Next we invoke \eqref{eq:expan3} and further expand 
 $T_a$ into  a sum of terms    on
 the form  $f_2 (H_a){ Q_k}^* B_k
 Q_k f_2 (H)$  for  suitable $Q_k f_2 (H)\in
 \vO(\inp{x}^{-s})$ (or possibly only for `components') and  bounded $B_k$, say for $k=1,\dots, K$, cf. \eqref{eq:rel1and2}. This is
 doable  along the line outlined in Section \ref{sec::A partition of unity}. The most important property reads, cf. \eqref{eq:LAPbnda},
\eqref{eq:2boundobtain338} and 
\eqref{eq:2boundobtain33}, 
\begin{align}\label{eq:2boundobtain3377}
  \begin{split}
  &\sup _{\Im z\neq 0}\norm{Q_k {f_2}
   (H){R(z)}}_{\vL(\vB,\widetilde{\vH})}< \infty,\\
&\sup _{\Im z\neq 0}\norm{Q_k {f_2}
   (H_a){R_a(z)}}_{\vL(\vB,\widetilde{\vH})}< \infty;\quad (\widetilde{\vH}=\vH\text{ or }\widetilde{\vH}=\vH^{d}).  
  \end{split}
\end{align} 
 These bounds 
   allow us
to take  the following weak limit, cf. Remark \ref{remark:RQ} \ref{item:1Q},
\begin{equation}
  \label{eq:weak} Q_k f_2 (H)R(\lambda+\i 0)\psi:=\wv2lim_{\epsilon\to 0_+}\,\, Q_k f_2 (H)R(\lambda+\i \epsilon)\psi.
\end{equation} 
 Using in turn that $\inp{x}^{s}f_2 (H_a){ Q_k}^* B_k$
is bounded, \eqref{eq:strongED_+SDY}--\eqref{eq:weak}  we can compute
the above third term
 as follows. Taking limits in the weak sense in $L^2_{-1}$
\begin{align*}
  -\i &\lim_{\epsilon\to 0_+}R_a(\lambda+\i
 \epsilon )\chi^2_+(2B/\epsilon_0)T_aR(\lambda+\i
        \epsilon)\psi\\
&=-\i \lim_{\epsilon\to 0_+}R_a(\lambda+\i
  0)\chi^2_+(2B/\epsilon_0)T_aR(\lambda+\i \epsilon)\psi\quad\text{
  by \eqref{eq:strongED_+SDY}} \mand \eqref{eq:2boundobtain3377}\\
&=-\i R_a(\lambda+\i 0 )\chi^2_+(2B/\epsilon_0)T_aR(\lambda+\i
  0)\psi\quad\text{ by \eqref{eq:weak} }\\
&=-\i \lim_{\epsilon\to 0_+}R_a(\lambda+\i
 \epsilon )\chi^2_+(2B/\epsilon_0)T_aR(\lambda+\i
        0)\psi\quad\text{ by \eqref{eq:strongED_+SDY} }.
\end{align*} 

We may summerize our interpretation  of \eqref{eq:res1} as
\begin{equation*}
  \psi_a\in (L^2_{s}\subseteq) \,\vH \mand  \phi_a=R_a(\lambda+\i
  0)\psi_a=\lim_{\epsilon\to 0_+}\,R_a(\lambda+\i
  \epsilon)\psi_a\text{ weakly in  }L^2_{-1}.
\end{equation*} Moreover it follows from \eqref{eq:2boundobtain3377}, \eqref{eq:weak} and a commutation that the vector $\phi_a$ defined this way is actually an element of $\vB^*$, although  this is of course obvious from the representation
\begin{equation*}
\phi_a= \lim_{\epsilon\to 0_+}  S_aR(\lambda+\i
\epsilon ) \psi. 
\end{equation*}  

Anyway  the above considerations combined  with those  in the beginning of Section \ref{sec::A partition of unity} as well as with  Theorem \ref{microLoc} finish our  justification of   \eqref{eq:quotientFor}.

\begin{proof}[Proof of Lemma \ref{lemma:reduc}]
We need  to construct a sequence of vectors $\psi_{a,n}\in L^2_\infty$ such that 
      $\psi_{a,n} \to \psi_a\in L^2_{1/2}$ as well as 
    \begin{equation}\label{eq:com2}
      R_a(\lambda+\i0)\psi_{a,n}\to R_a(\lambda+\i0)\psi_a \text{ in
      }\vB^*\text { for }n \to \infty.
    \end{equation}
First note  that 
the above discussion  allows us to represent,  in terms of `$Q$-operators' 
$ Q_k$ (and  bounded $B_k$), 
\begin{subequations}
\begin{equation}
  \label{eq:form}
  \psi_a=
\sum_{k} \,f_2 (H_a){ Q_k}^* B_k
 Q_k f_2 (H) R(\lambda+\i
        0)\psi+\brp;
\quad \brp\in\vB.
\end{equation}   We are let to define,
cf.
 the proof of \cite [Lemma  9.12]{Sk3},
\begin{equation}\label{eq:n2}
   \psi_{a,n}=
  \sum_{k} \,f_2 (H_a){ Q_k}^*\chi_n B_k
  Q_k f_2 (H) R(\lambda+\i
  0)\psi+\chi_n\brp;\quad \chi_n=\chi_-(\inp{x}/n).
\end{equation} 
\end{subequations}
 Obviously $\chi_n\brp\to \brp$ in $\vB \,(\subseteq L^2_{1/2})$, and due to  the fact  that
 $\chi_n B_k \to  B_k $ strongly on $\widetilde{\vH}$   it follows from  (\ref{eq:weak}) that $L^2_{\infty}\ni\psi_{a,n} \to \psi_a\in L^2_{1/2}$. Moreover it follows from \eqref{eq:2boundobtain3377},   
 \eqref{eq:BB^*a}  and the mentioned strong convergence
  that 
\begin{equation*}
  \sup_{\epsilon>0}\,\norm [\big]{R_a(\lambda+\i
  \epsilon) \parb{\psi_{a}-
\psi_{a,n}}}_{\vB^*}\to 0\text{ for }n\to \infty.
\end{equation*}
 This  uniform convergence and \eqref{eq:BB^*a}  yield
 \eqref{eq:com2}.
\end{proof}


\begin{thebibliography}{DoGa}

\bibitem[AHS]{AHS} A. Agmon, I. Herbst, E. Skibsted,  \emph{Perturbation of embedded eigenvalues in the generalized
    $N$-body problem},
 Comm. Math. Phys. \textbf{122},  (1989), 411--438.


  \bibitem[AIIS]{AIIS}
T. Adachi, K. Itakura, K. Ito, E. Skibsted, 
\emph{New methods in spectral theory of $N$-body Schr{\"o}dinger
  operators},  Rev. Math. Phys. \textbf{33}  (2021), 48 pp. 


\bibitem[De]{De}
J. Derezi\'nski, \emph{Asymptotic completeness for $N$-particle long-range quantum
    systems}, Ann. of Math.  \textbf{38}   (1993),
  427--476.


\bibitem[DG]{DG}
J. Derezi{\'n}ski, C. G{\'e}rard, \emph{Scattering theory of
 classical and quantum {$N$}-particle systems}, Texts and Monographs in
  Physics,  Berlin, Springer 1997.

 \bibitem[En]{En} V. Enss, \emph{Long-range scattering of two- and
     three-body quantum systems}, Journ\'ees ``Equations aux d\'eriv\'ees
   partielles'', Saint Jean de Monts, Juin 1989, Publications Ecole
   Polytechnique, Palaiseau 1989, pp. 31. 

\bibitem[GM]{GM} J. Ginibre, M. Moulin, \emph{Hilbert space
    approach to the quantum mechanical three-body problem},
  Ann. Inst. Henri Poincar\'e \textbf{XXI}  (1974), 97--145.


\bibitem[Gr]{Gr} G.M. Graf, \emph{Asymptotic completeness for
    $N$-body short-range quantum systems: a new proof},
  Commun. Math. Phys. \textbf{132} (1990), 73--101.

\bibitem[GT]{GT} D. Gilbarg, N.S. Trudinger, \emph{Elliptic partial
    differential operators of second order}, Grundlehren der mathematischen Wissenschaften \textbf{224},  Berlin, Springer 1977.

  \bibitem[HS]{HS} W. Hunziker, I. Sigal, \emph{Time-dependent
      scattering theory of $N$-body quantum systems},  Rev. Math. Phys. \textbf{12} (2000), 1033--1084.

    
 



\bibitem[H{\"o}]{H1} L. H{\"o}rmander, \emph{The analysis of linear
 partial differential operators. {I--IV}}, Berlin, Springer
 1983--85.


\bibitem[II]{II}
T. Ikebe, H. Isozaki, \emph{A stationary approach to the existence and
  completeness of long-range operators}, Integral equations and
operator theory \textbf{5} 
   (1982), 18--49.


\bibitem[Is1]{Is}
 H. Isozaki, \emph{Eikonal equations and spectral representations for
    long range  Schr{\"o}dinger Hamiltonians},  J. Math. Kyoto Univ. \textbf{20} 
   (1980), 243--261.

   
\bibitem[Is2]{Is1}
 H. Isozaki, \emph{Structures of the S-matrices for  three-body
   Schr{\"o}dinger operators},   Commun. Math. Phys. \textbf{146}
 (1992), 241--258.

\bibitem[Is3]{Is2}
 H. Isozaki, \emph{Asymptotic properties of generalized eigenfunctions
   for  three-body
   Schr{\"o}dinger operators},   Commun. Math. Phys. \textbf{153}
 (1993), 1--21.

\bibitem[Is4]{Is3}
 H. Isozaki, \emph{Asymptotic properties of solutions to $3$-particle
   Schr{\"o}dinger equations},   Commun. Math. Phys. \textbf{222}  (2001), 371--413.

\bibitem[Is5]{Is5}
 H. Isozaki, \emph{Many-body Schr{\"o}dinger equation--Scattering Theory and Eigenfunction
Expansions}, Mathematical Physics Studies, Singapore, Springer  2023. 

\bibitem[Ka]{Ka} T. Kato, \emph{Smooth operators and commutators},
  Studia Math. \textbf{31} (1968), 535--546.

\bibitem[Ne]{Ne} R.G.  Newton, \emph{The asymptotic form of the
    three-particle wave functions and the cross sections},
  Ann. Phys.  \textbf{74} (1972), 324--351.


\bibitem[RS]{RS}
M.~Reed, 
B.~Simon, \emph{Methods of modern mathematical physics {I}--{I\hspace{-.1em}V}},
  New York, Academic Press 1972-78.

\bibitem[SS]{SS}
I. Sigal, A. Soffer, \emph{The $N$-particle scattering problem:
  asymptotic completeness for short-range systems},  Ann. of Math. (2) \textbf{126} (1987), 35--108.

\bibitem[Sk1]{Sk3}
 E. Skibsted, \emph{Stationary scattering theory: the $N$-body
   long-range case},  Commun. Math. Phys. (2023), 1--75.

\bibitem[Sk2]{Sk2}
E. Skibsted, \emph{Green functions and completeness; the $3$-body problem
	revisited}, preprint   30 May 2022,
      http://arxiv.org/abs/2205.15028v1  (an extended version in
      preparation).

\bibitem[SW]{SW}
 E. Skibsted, X.P. Wang, \emph{Spectral analysis of $N$-body
   Schr{\"o}dinger operators at two-body thresholds}, to appear in
 Mathematical Physics Studies,  Springer  2024.

\bibitem[Ta]{Ta} H. Tamura, \emph{Asymptotic  completeness  for
    four-body  Schr\"odinger operators   with  short-range
    interactions},  Publ. Res. Inst. Math. Sci. \textbf{29} (1993), 1--21.
  

\bibitem[Ya1]{Ya1} D.R. Yafaev, \emph{Resolvent estimates and scattering
    matrix for $N$-body Hamiltonians},
  Intgr. Equat. Oper. Th. \textbf{21} (1995), 93--126.

  

\bibitem[Ya2]{Ya3} D.R. Yafaev, \emph{Radiation conditions and
    scattering theory for $N$-particle Hamiltonians},
  Commun. Math. Phys. \textbf{154}  (1993), 523--554.

\bibitem[Ya3]{Ya4} D.R. Yafaev, \emph{Eigenfunctions of the continuous
   spectrum for $N$-particle Schr{\"o}dinger
  operator}, Spectral and scattering theory, Lecture notes in pure and
applied mathematics, Marcel Dekker, New York  1994, 259--286.
 

\bibitem[Zi]{Zi} L.  Zielinski, \emph{A proof of asymptotic
    completeness for $n$-body  Schr{\"o}dinger
  operators}, Comm. Partial Differential Equations \textbf{19} (1994), 455--522.

\end{thebibliography}
\end{document}